\documentclass[11pt]{article}
\usepackage{authblk}
\title{A Weyl Criterion for Finite-State Dimension and Applications}
\author[1]{Jack H. Lutz\footnote{This author's research was supported in part by National Science Foundation grant
1900716.}}
\affil[1]{Department of Computer Science\\
Iowa State University,
Ames, IA 50011, USA}
\affil[]{lutz@cs.iastate.edu}
\author[2]{Satyadev Nandakumar}
\author[2]{Subin Pulari}
\affil[2]{
  Department of Computer Science and Engineering\\
  Indian Institute of Technology Kanpur,
  Kanpur, Uttar Pradesh, India.
}
\affil[]{\{satyadev,subinp\}@cse.iitk.ac.in}

\usepackage{amsmath,amsthm,amssymb}
\usepackage[margin=1in]{geometry}
\usepackage[colorlinks=true]{hyperref}

\usepackage{lineno,hyperref}
\usepackage{lipsum}

\newcommand{\Mod}{\text{ }\mathrm{mod}\text{ }}
\newtheorem*{theorem*}{Theorem}
\newtheorem*{conjecture*}{Conjecture}
\newtheorem*{question*}{Question}
\newtheorem*{lemma*}{Lemma}
\newcommand{\T}{\mathbb{T}}

\newcommand{\seq}[1]{\langle {#1}_n \rangle_{n \in \N}}

\newcommand{\wto}{\Rightarrow}

\usepackage{amsfonts}
\usepackage{amsmath} 
\usepackage{amssymb}
\usepackage{enumerate}

\newcommand{\PROOF}{\begin{proof}}%{{\noindent\bf{\em Proof:\/}~~}}

\newcommand{\QED}{\end{proof}}

\newcommand{\N}{\mathbb{N}}
\newcommand{\Z}{\mathbb{Z}}
\newcommand{\Q}{\mathbb{Q}}
\newcommand{\R}{\mathbb{R}}

\renewcommand{\dim}{{\mathrm{dim}}}

\newcommand{\Dim}{{\mathrm{Dim}}}

\newcommand{\D}{{\mathcal{D}}}

\usepackage[svgnames]{xcolor}
\def\<{\left\langle}
\def\>{\right\rangle}

\theoremstyle{plain}
\newtheorem{theorem}{Theorem}[section]
\newtheorem{lemma}[theorem]{Lemma}
\newtheorem{corollary}[theorem]{Corollary}

\theoremstyle{definition}

\newtheorem{definition}[theorem]{Definition}

\begin{document}
\pagenumbering{gobble}
\setlength{\belowdisplayskip}{0pt} \setlength{\belowdisplayshortskip}{0pt}
\setlength{\abovedisplayskip}{4pt} \setlength{\abovedisplayshortskip}{0pt}

\maketitle

\begin{abstract}
Finite-state dimension, introduced early in this century as a
finite-state version of classical Hausdorff dimension, is a
quantitative measure of the lower asymptotic {\it density of
  information} in an infinite sequence over a finite alphabet, as
perceived by finite automata.  Finite-state dimension is a robust
concept that now has equivalent formulations in terms of finite-state
gambling, lossless finite-state data compression, finite-state
prediction, entropy rates, and automatic Kolmogorov complexity. The
1972 Schnorr-Stimm dichotomy theorem gave the first automata-theoretic
characterization of normal sequences, which had been studied in
analytic number theory since Borel defined them in 1909.  This theorem
implies, in present-day terminology, that a sequence (or a real number
having this sequence as its base-b expansion) is normal if and only if
it has finite-state dimension 1. One of the most powerful classical
tools for investigating normal numbers is the 1916 Weyl's criterion,
which characterizes normality in terms of exponential sums.  Such sums
are well studied objects with many connections to other aspects of
analytic number theory, and this has made use of Weyl's criterion
especially fruitful.  This raises the question whether Weyl's
criterion can be generalized from finite-state dimension 1 to
arbitrary finite-state dimensions, thereby making it a quantitative
tool for studying data compression, prediction, etc. i.e., \emph{Can
we characterize all compression ratios using exponential sums?}.

This paper does exactly this. We extend Weyl's criterion from a
characterization of sequences with finite-state dimension 1 to a
criterion that characterizes every finite-state dimension.  This turns
out {\em not} to be a routine generalization of the original Weyl
criterion. Even though exponential sums may diverge for non-normal
numbers, finite-state dimension can be characterized in terms of the
\emph{dimensions} of the \emph{subsequence limits} of the exponential
sums. In case the exponential sums are convergent, they converge to
the Fourier coefficients of a probability measure whose
\emph{dimension} is precisely the finite-state dimension of the
sequence. 

This new and surprising connection helps us bring Fourier analytic
techniques to bear in proofs in finite-state dimension, yielding a new
perspective. We demonstrate the utility of our criterion by
substantially improving known results about preservation of
finite-state dimension under arithmetic. We strictly generalize the
results by Aistleitner and Doty, Lutz and Nandakumar for finite-state
dimensions under arithmetic operations. We use the method of
exponential sums and our Weyl criterion to obtain the following new
result: \emph{If $y$ is a number having finite-state strong dimension
0, then $\dim_{FS}(x+qy)=\dim_{FS}(x)$ and $\Dim_{FS}(x+qy)=\Dim_{FS}(x)$ for any $x \in \R$ and $q \in \Q$.}
This generalization uses recent estimates obtained in the work of
Hochman \cite{Hochman2014} regarding the entropy of convolutions of
probability measures.
\end{abstract}

\section{Introduction}
Finite-state compressibility \cite{LZ78}, or equivalently,
finite-state dimension \cite{dai2004finite},
\cite{athreya2007effective}, \cite{bourke2005entropy} is a
quantification of the information rate in data as measured by
finite-state automata. This formulation, initially motivated by
practical constraints, has proved to be rich and mathematically
robust, having several equivalent characterizations. In particular,
the finite state-dimension of a sequence is equal to the compression
ratio of the sequence using information lossless finite-state
compressors (\cite{dai2004finite},
\cite{athreya2007effective}). Finite-state dimension has unexpected
connections to areas such as number theory, information theory, and
convex analysis \cite{KuipersNiederreiterUniform}, \cite{DLN06}.
Schnorr and Stimm \cite{SchnorrStimm72} establish a particularly
significant connection by showing that a number is Borel normal in
base $b$ (see for example, \cite{Niven}) if and only if its base $b$
expansion has finite-state compressibility equal to 1, \emph{i.e.}, is
incompressible (see also: \cite{becher2013normal},
\cite{bourke2005entropy}). Equivalently, a number $x \in [0,1)$ is
  normal if and only if $\dim_{FS}(x)$, the finite-state dimension of
  $x$ is equal to $1$.

A celebrated characterization of Borel normality in terms of
exponential sums, provided by Weyl's criterion \cite{Weyl1916}, has
proved to be remarkably effective in the study of normality. Weyl's
criterion on uniformly distributed sequences modulo 1 yields a
characterization that a real number $r$ is normal to base $b$ if and
only if for every integer $k$,
\begin{align}
  \label{weylsum}
  \lim_{n \to \infty} \frac{1}{n} \sum_{j=0}^{n-1} e^{2\pi i k (b^jr)} = 0.
\end{align}
This tool was used by Wall \cite{Wall49} in his pioneering thesis to
show that normality is preserved under certain operations like
selection of subsequences along arithmetic progressions, and
multiplication with non-zero rationals. Weyl's criterion facilitates
the application of tools from Fourier analysis in the study of Borel
normality. Weyl's criterion is used in several important constructions
of normal numbers including those given by Cassels \cite{cassels},
Erd\"os and Davenport \cite{erdosdavenport} etc. The criterion was
also instrumental in obtaining the construction of absolutely normal
numbers given by Schmidt in \cite{schmidtconstruction}.

The finite-state compression ratio/dimension of an arbitrary sequence
is a quantity in [0,1]. The classical Weyl's criterion provides a
characterization of numbers having finite-state dimension equal to $1$
in terms of exponential sums. This leads us to the natural question -
\emph{Can we characterize arbitrary compression ratios using
exponential sums?}.

This question turns out to be highly non-trivial.  It is not easy to
generalize Weyl's criterion to study arbitrary finite-state
compression ratios/dimension. The major conceptual hurdle arises from
the fact that for non-normal numbers, the Weyl sum averages in
(\ref{weylsum}) need not converge. The Weyl averages need not converge
even when the finite-state dimension and the strong dimension of a
sequence are equal.
  
We demonstrate this by explicitly constructing such a sequence in
Lemma \ref{lem:nonconvergentexample}. Using a new construction method
involving the controlled concatenation of two special sequences, we
demonstrate the existence of a sequence $x \in \Sigma^\infty$ with
non-convergent Weyl averages, while having finite-state dimension and
strong dimension both equal to $\frac{1}{2}$. The proof that this
constructed sequence satisfies the required properties uses new
techniques, which might be of independent interest.  Due to the
existence of such sequences, it is unclear how to \emph{extract} the
finite-state dimension of a sequence from non-convergent Weyl
averages. Indeed, it was unclear whether any generalization of the
Weyl's criterion to arbitrary finite-state dimensions even exists.

Our paper rescues this approach and gives such a characterization of
arbitrary finite-state compressibility/dimension by introducing one
important viewpoint, that turns out to be the major theoretical
insight.  Even when the exponential sums diverge, the theory of
\emph{weak convergence} of probability measures
(\cite{billingsley2013convergence}) enables us to consider the
collection of all probability measures having Fourier coefficients
equal to the the subsequence limits of the Weyl averages. The
\emph{dimensions} of the measures in the set of subsequence weak limit
measures gives a generalization of Weyl's criterion. For any $x$, let
$\dim_{FS}(x)$ and $\Dim_{FS}(x)$ denote the finite-state dimension
and finite-state strong dimension \cite{athreya2007effective} of $x$
respectively. We now informally state our Weyl's criterion for
finite-state dimension.

\begin{theorem*}[Informal statement of Theorem
    \ref{thm:nonconvergentcaseweylcriterion}] 
	Let $x \in [0,1)$. If for any subsequence $\langle n_m
        \rangle_{m=0}^{\infty}$ of natural numbers, there exist
        complex numbers $c_k$ such that for every $k \in \Z$,  
%%         \begin{align*}
        	$\lim_{m \to
          \infty}\frac{1}{n_m}\sum_{j=0}^{n_m-1} e^{2 \pi i k (b^jx) }
                = c_k$, 
%%         \end{align*}
then, there exists a probability measure $\mu$ on $[0,1)$ such that
  for every $k$, $c_k = \int e^{2 \pi i k y } d\mu$. Let
  $\mathcal{W}_x$ be the collection of all such probability measures
  $\mu$ on $[0,1)$ that can be obtained as the subsequence limits of
    Weyl averages. Then, $ \dim_{FS}(x) = \inf_{\mu \in \mathcal{W}_x}
    H^{-}(\mu)$ and $\Dim_{FS}(x) = \sup_{\mu \in \mathcal{W}_x}
    H^{+}(\mu)$.
\end{theorem*}

The \emph{correct} notion of \emph{dimensions} of the subsequence weak
limit measures in $\mathcal{W}_x$ which yields the finite-state
dimensions of $x$ turns out to be $H^{-}$ and $H^{+}$, the
\emph{lower} and \emph{upper average entropies} of $\mu$ as defined in
\cite{athreya2007effective} \footnote{These are analogues of the
well-known \emph{R\'enyi upper and lower dimensions} of measures as
defined in \cite{renyi1959dimension}. See the remark following
Definition \ref{def:upperandloweraverageentropies}.}. Therefore, this
new characterization enables us to \emph{extract} the finite-state
compressibility/dimension by studying the behavior of the Weyl sum
averages, thereby extending Weyl's criterion for normality to
arbitrary finite-state dimensions.

An interesting special case of our criterion is when the exponential
averages of a sequence are convergent. In this case, the averages
$\langle c_k \rangle_{k \in \Z}$ are precisely the \emph{Fourier
coefficients} of a \emph{unique} limiting measure, whose
\emph{dimension} is precisely the finite-state dimension of the
sequence. This relates two different notions of dimension to each
other. We give the informal statement of our criterion for this
special case.
\begin{theorem*}[Informal statement of Theorem
    \ref{thm:weylaveragesconvergentcase}] 
	Let $x \in [0,1)$. If there exist complex numbers $c_k$ for $k \in
\Z$ such that $\frac{1}{n}\sum_{j=0}^{n-1} e^{2 \pi i k (b^j x)} \to
c_k$ as $n \to \infty$, then, there exists a unique measure $\mu$ on
$[0,1)$ such that for every $k$, $c_k = \int e^{2 \pi i k y }
d\mu$. Furthermore, $\dim_{FS}(x) =\Dim_{FS}(x)= H^{-}(\mu) =
H^{+}(\mu)$.
\end{theorem*}

Our results also show that in case there is a unique weak limit measure,
the exponential sums (\ref{weylsum}) converge for every $k \in
\Z$. These give the first known relations between Fourier coefficients
and finite-state compressibility/dimension.

The proof of Weyl's criterion for finite-state dimension is not a
routine generalization of the available proofs of Weyl's criterion for
normality (see \cite{Weyl1916}, \cite{EinsiedlerWardErgodic},
\cite{SteinShakarchiFourier}) and requires several facts from the
theory of weak convergence of probability measures and new
relationships involving the exponential sums, the \emph{dimensions} of
weak limit measures and the finite-state dimension of the given
sequence. Certain additional technical difficulties we overcome
include working with two different topologies - the topology on the
torus $\T$ where Fourier coefficients uniquely determine a measure,
and another, Cantor space, which is required for studying
combinatorial properties of sequences, like normality.
 
\subsection{Applications of our criterion}

We illustrate how this framework can be applied in sections
\ref{sec:preservationoffsd}. These results justify that this framework
pioneers a new, powerful, approach to data compression.

It is not very surprising that when the Weyl averages converge, our
criterion has applications. Importantly, even in situations where the
Weyl averages do not converge, it is possible to derive non-trivial
consequences. We apply our techniques to substantially improve known
results about the preservation of finite-state dimension under
arithmetic and combinatorial operations.

%% To this end, we first derive a Fourier analytic proof of the result by
Doty, Lutz and Nandakumar \cite{DLN06} show that if $x$ is any real
and $q$ is any non-zero rational, then the finite-state dimensions and
strong dimensions of $x$, $qx$ and $x+q$ are equal. 
%% Our new proof is quintessentially Fourier analytic, using the
%% invariance of R\'{e}nyi dimension under specific pushforward
%% operations.
When $x$ is normal, a generalization is obtained by Aistleitner
\cite{Aistleitner2011}, which can be described as follows. Let $y$ be
any real such that the asymptotic density of zeroes in its expansion
is one. Then, for any rational $q$, we have $x+qy$ is normal. We
generalize these results by allowing both the following conditions
simultaneously,
\begin{enumerate}
\item $x$ is allowed to be any real, obtaining a result for all
  finite-state dimensions rather than only for normals as in
  Aistleitner \cite{Aistleitner2011} and
\item $y$ is allowed to be any real with finite-state strong dimension
  0 which satisfies a natural independence condition. This generalizes
  both the restrictions in Doty, Lutz and Nandakumar \cite{DLN06} and
  Aistleitner \cite{Aistleitner2011}, 
\end{enumerate}
and show that for any rational $q \in \Q$,
$\dim_{FS}(x+qy) = \dim_{FS}(x)$ and $\Dim_{FS}(x+qy) = \Dim_{FS}(x)$ if $x$ and $y$ are \emph{independent}
(see Definition \ref{def:independentstrings}) and $\Dim_{FS}(y)=0$.

Using our Weyl criterion along with the results in Hochman
\cite{Hochman2014}, we obtain the following inequalities. Let $x$ and
$y$ be real numbers in $\T$ such that $x$ and $y$ are independent.
Then for any $d,e \in \Z$, $\dim_{FS}(dx+ey) \geq \max
\{\dim_{FS}(dx), \dim_{FS}(ey) \}$ and $\dim_{FS}(dx+ey) \leq
\dim_{FS}(dx) + \Dim_{FS}(ey)$. Similarly, $\Dim_{FS}(dx+ey) \geq \max
\{\Dim_{FS}(dx), \Dim_{FS}(ey) \}$ and $\Dim_{FS}(dx+ey) \leq
\Dim_{FS}(dx) + \Dim_{FS}(ey)$.  Our main results are consequences of
these inequalities.

Further, using our Weyl criterion and techniques from Mance and Madritsch
\cite{madritsch2016construction} we obtain new methods for the
explicit construction of numbers having a specified finite-state
dimension. (See Appendix \ref{sec:munormality})

Lossless data compression is practically significant, and
theoretically sophisticated. We show how one of the major tools of
modern mathematics, Fourier analysis, can be brought to bear to study
compressibility of individual data sequences. We hope that our
criterion will facilitate the application of more powerful Fourier
analytic tools in future works involving finite-state
compression/dimension.

\subsection{Organization of the article}

After the preliminary sections, section
\ref{sec:weylscriterionandweakconvergence} gives Weyl's criterion on
Cantor space using weak convergence of measures. Next, we show the
necessity and the sufficiency of passing to \emph{subsequences} of
sequences of measures in order to generalize Weyl's criterion for
finite-state dimension. In section \ref{sec:preservationoffsd} we show
the applications of our Weyl criterion to yield new, general results
regarding the preservation of finite-state dimension under arithmetic
and combinatorial operations.

\section{Preliminaries}
For any natural number $b >1$, $\Sigma_b$ denotes the alphabet
$\{0,1,2,\dots b-1\}$. Throughout this paper, we work with base 2, but
our results generalize to all bases. We use $\Sigma$ to denote the
binary alphabet $\Sigma_2$. We denote the set of finite binary strings
by $\Sigma^*$ and the set of infinite sequences by
$\Sigma^\infty$. For any $w \in \Sigma^*$, let $C_w$ be the set of
infinite sequences with $w$ as a prefix, called a \emph{cylinder}. For
any sequence $x=x_0x_1x_2\dots$ in $\Sigma^\infty$, we denote the
substring $x_ix_{i+1}\dots x_{j}$ of $x$, by $x_i^{j}$. The Borel
$\sigma$-algebra generated by the set of all cylinder sets is denoted
by $\mathcal{B}(\Sigma^\infty)$.

Let $\T$ denote the one-dimensional torus or unit circle. i.e, $\T$ is
the unit interval $[0,1)$ with the metric $d(r,s)=\min\{\lvert r-s
  \rvert,1-\lvert r-s \rvert\}$. $\T$ is a compact metric space. The
  Borel $\sigma$-algebra generated by all open sets in $\T$ is denoted
  by $\mathcal{B}(\T)$. For any base $b$, let $v_b$ be the
  \emph{evaluation map} which 
  maps any $x \in \Sigma^\infty$ to its \textit{value} in $\T$ which
  is $\sum_{i=0}^{\infty} \frac{x_i}{b^{i+1}} \Mod 1$. We use the
  simplified notation $v$ to denote the base $2$ evaluation map $v_2$.
  Let $T$ be 
  the left shift transformation $T(x_0x_1x_2\dots)=x_1x_2x_3\dots$ on
  $\Sigma^\infty$.  For any base $b$ and $w\in \Sigma_b^*$, let
  $I^b_w$ denote the interval 
  $\left[ v_b(w0^\infty), v_b(w0^\infty)+b^{-\lvert w \rvert} \right)$ in
    $\T$. We use the simplified notation $I_w$ to refer to $I^2_w$.
    Let $\mathbb{D}$ be the set of all dyadic rationals in $\T$. 
    It is easy to see that $v:\Sigma^\infty \to \T$ has a well-defined
    inverse, denoted $v^{-1}$, over $\T \setminus \mathbb{D}$.
    
For any measure $\mu$ on $\T$ (or $\Sigma^\infty$), we refer to the
collection of complex numbers $ \int e^{2 \pi i k y} d\mu $ where $k$
ranges over $\Z$ as the \emph{Fourier coefficients of measure
$\mu$}. For measures over $\Sigma^\infty$, the function $e^{2 \pi i k
  y}$ inside the integral is replaced with $e^{2 \pi i k v(y)}$.  For
every measure $\mu$ on $\T$, we define the corresponding lifted
measure on $\Sigma^\infty$ as follows.
\begin{definition}[\emph{Lift $\hat{\mu}$ of a measure $\mu$ on $\T$}]
If $\mu$ is a measure on $\T$, then we define the \emph{lift}
$\hat{\mu}$ of $\mu$ to be the unique measure on $\Sigma^\infty$
satisfying $\hat{\mu}(C_w)=\mu(I_w)$ for every string $w \in
\Sigma^*$.
\end{definition}
The uniqueness of $\hat{\mu}$ follows from routine measure theoretic
arguments.

\begin{definition}
Let $x \in \Sigma^*$ have length $n$.  We define the \emph{sliding
count probability} of $w \in \Sigma^*$ in $x$ denoted $P(x,w)$, and
the \emph{disjoint block probability} of $w$ in $x$, denoted
$P^d(x,w)$, as follows.
\begin{align*}
P(x,w)=\frac{\lvert \{i \in [0,n-\lvert w \rvert] : x_i^{i+\lvert
    w \rvert-1} = w\}\rvert }{n-\lvert w
  \rvert+1}
\text{ and }
P^d(x,w)=\frac{\lvert \{i \in [0,\frac{n}{\lvert w \rvert}) :
  x_{\lvert w \rvert i}^{\lvert w \rvert(i+1)-1} = w\}\rvert
}{n/\lvert w \rvert}
\end{align*}
\end{definition}
Now, we define normal sequences in $\Sigma^\infty$ and normal numbers
on $\T$.
\begin{definition}
 A sequence $x \in \Sigma^\infty$ is  \emph{normal} if for
 every $w \in \Sigma^*$, $ \lim_{n \to \infty}
 P(x_0^{n-1},w)=2^{-\lvert w \rvert}$.  A number $r \in \T$ is
 \emph{normal} if and only if $r \not\in \mathbb{D}$ and $v^{-1}(r)$
 is a normal sequence in $\Sigma^\infty$.
\end{definition}
Equivalently, we can formulate normality using disjoint probabilities
\cite{KuipersNiederreiterUniform}. The following is the block entropy
characterization of finite-state dimension from
\cite{bourke2005entropy}, which we use instead of the original
formulation using $s$-gales \cite{dai2004finite},
\cite{athreya2007effective}.
\begin{definition}[\cite{dai2004finite}, \cite{bourke2005entropy}]
\label{def:finitestatedimension}
For a given block length $l$, we define the \emph{sliding block
entropy} over $x_0^{n-1}$ as %% follows.
%% \begin{align*}
$H_l(x_0^{n-1}) = -\frac{1}{l}\sum_{w \in \Sigma^l}
P(x_0^{n-1},w)\log(P(x_0^{n-1},w))$.
%% \end{align*}
 The \emph{finite-state dimension} of $x \in \Sigma^\infty$,
denoted $\dim_{FS}(x)$, and \emph{finite-state strong dimension} of
$x$, denoted $\Dim_{FS}(x)$, are defined as follows.
\begin{align*}
\dim_{FS}(x) = \inf_{l} \liminf_{n \to \infty}
H_l(x_0^{n-1})
\quad\text{and}\quad 
\Dim_{FS}(x) = \inf_{l} \limsup_{n \to \infty}
H_l(x_0^{n-1}). 
\end{align*}
\end{definition}

\textbf{Remark:} The fact that $\dim_{FS}(x)$ and $\Dim_{FS}(x)$ are
equivalent to the lower and upper finite-state compressibilities of
$x$ using lossless finite-state compressors, follows immediately from
the results in \cite{LZ78} and \cite{dai2004finite}.

\emph{Disjoint block entropy} $H_l^d$ is defined similarly by
replacing $P$ with $P^d$. Bourke, Hitchcock and Vinodchandran
\cite{bourke2005entropy}, based on the work of Ziv and Lempel
\cite{LZ78}, demonstrated the entropy characterization of finite-state
dimension using $H_l^d$ instead of $H_l$. Kozachinskiy and Shen
(\cite{kozachinskiy2019two}) proved that the finite-state dimension of
a sequence can be equivalently defined using sliding block entropies
(as in Definition \ref{def:finitestatedimension}) instead of disjoint
block entropies. It is clear from the definition that, for any $x \in
\Sigma^\infty$, $\dim_{FS}(x) \leq \Dim_{FS}(x)$. Any $x$ with
$\dim_{FS}(x) = \Dim_{FS}(x)$ is called a \emph{regular sequence}.

Upper and lower average entropies were defined in
\cite{athreya2007effective} for measures constructed out of infinite
bias sequences. We extend these notions to the set of all measures on
$\Sigma^\infty$ below.
\begin{definition}
\label{def:upperandloweraverageentropies}
For any probability measure $\mu$ on $\Sigma^\infty$, let
$\mathbf{H}_n (\mu)=-\sum_{w \in \Sigma^n}\mu(C_w)\log(\mu(C_w))$.
The \emph{upper average entropy}  of
$\mu$, denoted $H^{+}(\mu)$, and its \emph{lower average entropy},
denoted $H^-(\mu)$, are defined respectively as the limit superior and
the limit inferior as $n$ tends to $\infty$ of $H_n(\mu)/n$.
%% \begin{align*}
%%   H^{+}(\mu)=\limsup_{n \to \infty} \frac{\mathbf{H}_{n}(\mu)}{n}
%%   \quad\text{and}\quad
%% H^{-}(\mu)=\liminf_{n \to \infty} \frac{\mathbf{H}_{n}(\mu)}{n}.
%% \end{align*}

\end{definition}
Upper and lower average entropies are the Cantor space analogues of
\emph{R\'enyi upper and lower dimensions} of measures on [0,1) which
  were originally defined for measures on the real line in
  \cite{renyi1959dimension}. For any $x \in \T$ (or $x \in
  \Sigma^\infty$) , let $\delta_x$ denote the Dirac measure at $x$.
  i.e, $\delta_x (A)=1$ if $x \in A$ and $0$ otherwise for every $A
  \in \mathcal{B}(\T)$ (or $A \in \mathcal{B}(\Sigma^\infty)$). Given
  a sequence $\langle x_n \rangle_{n=0}^{\infty}$ of numbers in
  $\T$(or $\Sigma^\infty$), we investigate the behavior of exponential
  averages $\frac{1}{n} \sum_{j=0}^{n-1} e^{2 \pi i k x_j}$ by
  studying the weak convergence of sequences of averages of Dirac
  measures which are defined as follows.
\begin{definition}
Given a sequence $\langle x_n \rangle_{n=0}^{\infty}$ in $\T$ (or
elements in $\Sigma^\infty)$, we say that $\langle \nu_n
\rangle_{n=1}^{\infty}$ is the sequence of averages of Dirac measures
over $\T$ (or over $\Sigma^\infty$) constructed out of the sequence
$\langle x_n \rangle_{n=0}^{\infty}$ if, $ \nu_n = n^{-1}
\sum_{i=0}^{n-1} \delta_{x_i} $ for each $n \in \N$.
\end{definition}

\section{Weyl's criterion and weak convergence}
\label{sec:weylscriterionandweakconvergence}
Schnorr and Stimm \cite{SchnorrStimm72} (see also
\cite{becher2013normal}, \cite{bourke2005entropy}) showed a central
connection between normal numbers and finite-state compressibility, or
equivalently, finite-state dimension: a sequence $x \in \Sigma^\infty$
is normal if and only if its finite-state dimension is 1. Any $x \in
\Sigma^\infty$ has finite-state dimension (equivalently, finite-state
compressibility) between $0$ and $1$. In this sense, finite-state
dimension is a generalization of the notion of normality.  Another
celebrated characterization of normality, in terms of exponential
sums, was provided by Weyl in 1916. This characterization has resisted
attempts at generalization. In the present section, we show that the
theory of weak convergence of measures yields a generalization of
Weyl's characterization for arbitrary dimensions. We demonstrate the
utility of this new characterization to finite-state
compressibilty/finite-state dimension, in subsequent sections.
%% \subsection{Weyl's criterion for Cantor Space}

Weyl's criterion for normal numbers on $\mathbb{T}$ is the following.

\begin{theorem}[Weyl's criterion \cite{Weyl1916}]
\label{thm:weyls_criterion_on_torus}
  A number $r \in \mathbb{T}$ is normal if and only if for every $k \in
  \mathbb{Z}$, 
    $\lim_{n \to \infty}
    \frac{1}{n} \sum_{j=0}^{n-1} e^{2\pi i k (2^j r)}
    = 0$.
\end{theorem}

The insight in this theorem is the connection between a number $x$
being normal, and the concept of the collection of its shifts being
uniformly distributed in the unit interval. It is the latter concept
which leads to the cancellation of the exponential sums of all
orders. We now prove a formulation of this criterion on Cantor space,
which we require in our work.

\begin{theorem}[Weyl's criterion on $\Sigma^\infty$]
\label{thm:weyls_criterion_on_cantor_space}
  A sequence $x \in \Sigma^\infty$ is a normal sequence if and only if
  for every $k \in \mathbb{Z}$, $\lim_{n \to \infty} \frac{1}{n}
  \sum_{j=0}^{n-1} e^{2\pi i k (v(T^j x))} = 0$.
\end{theorem}
\begin{proof}
  Suppose $v(x) \in \mathbb{D}$ for some $x \in
  \Sigma^\infty$. Then $x$ is not a normal sequence. We also have that
  $v(T^j x)$ is either 0 or 1 for all sufficiently large $j$. In
  either case, for all $k$, $e^{2\pi i k v(T^j x)} = 1$, for all
  sufficiently large $j$. Thus, the  following holds.
    $$\lim_{n \to \infty}
  \frac{1}{n} \sum_{j=0}^{n-1} e^{2\pi i k (v(T^j x))} = 1.$$
  Hence the theorem
  holds for dyadic rationals.

  If $v(x)\notin \mathbb{D}$, then $x \in \Sigma^\infty$
  represents its unique binary expansion. Further, for every $j \in
  \mathbb{N}$, $v(T^j x) \notin \mathbb{D}$. Hence, $v(T^j x)$ is
  equal to $2^j v(x) \Mod 1$. We apply Weyl's theorem on $\mathbb{T}$
  to conclude that
  \begin{align}
    \lim_{n \to \infty}\frac{1}{n} \sum_{j=0}^{n-1} e^{2\pi i k (v(T^j x))} =
    \lim_{n \to \infty}\frac{1}{n} \sum_{j=0}^{n-1} e^{2\pi i k (2^j v(x))} = 0
  \end{align}
  if and only if $v(x)$ is normal. Now, the theorem follows since $x$ represents the unique binary expansion of $v(x)$.
\end{proof}

%% \subsection{Weyl's criterion, weak convergence of measures on
%%   $\Sigma^\infty$, and Fourier coefficients of measures}

The key to generalizing Weyl's criterion to sequences with
finite-state dimension less than 1 is to characterize convergence of
subsequences of exponential sums using weak convergence of probability
measures on $\Sigma^\infty$ (see Billingsley
\cite{billingsley2013convergence}). Over $\T$, this equivalent
characterization is well-known (see Section 4.4 from
\cite{EinsiedlerWardErgodic}). Obtaining the same equivalence over
$\Sigma^\infty$ involves some technical hurdles due to the fact that
continuous functions over $\Sigma^\infty$ need not have a uniform
approximation using trigonometric polynomials. In order to overcome
these, we need to carefully study the relationship between the
convergence of Weyl averages and weak convergence over
$\Sigma^\infty$. We develop these relationships in the following
lemmas. At the end of this section we characterize Theorem
\ref{thm:weyls_criterion_on_cantor_space} in terms of weak convergence
of a sequence of measures over $\Sigma^\infty$.
\begin{definition}
A sequence $\seq{\nu}$ of probability measures on a metric space
$(X,d)$ \emph{converges weakly} to a probability $\mu$ on $(X,d)$,
denoted $\nu_n \wto \mu$, if for every bounded continuous function $f:
X \to \mathbb{C}$, we have $\lim_{n \to \infty} \int f d\nu_n = \int f
d\mu$.
\end{definition}
If a sequence of measures $\seq{\nu}$ on a metric space $(X,d)$ has
a weak limit measure, then the weak limit must be unique (see Theorem
1.2 from \cite{billingsley2013convergence}). Since $\T$ and
$\Sigma^\infty$\footnote{The metric on $\Sigma^\infty$ is $d(x,y) =
2^{-\min\{i \mid x_i \ne y_i\}}.$} are compact metric spaces, using
Prokhorov's Theorem (see Theorem 5.1 from
\cite{billingsley2013convergence}) we get that any sequence of
measures $\langle \nu_n \rangle_{n \in \N}$ on $\T$ (or
$\Sigma^\infty$), has a measure $\mu$ on $\T$ (or $\Sigma^\infty$) and
a subsequence $\langle \nu_{n_m} \rangle_{m \in \N}$ such that
$\nu_{n_m} \wto \mu$. We first establish a relationship between weak
convergence of measures on $\T$ and the convergence of measures of
dyadic intervals in $\T$. Since the set of all finite unions of dyadic
intervals in $\T$ is closed under finite intersections, we
obtain the following lemma using Theorem 2.2 from
\cite{billingsley2013convergence}.
\begin{lemma}
\label{lem:dyadicconvergenceimpliesweakconvergence}
If for every dyadic interval $I$ in $\T$, $\lim_{n \to
      \infty} \nu_{n}(I)=\mu(I)$, then $\nu_n \wto \mu$. 
\end{lemma}

The Portmanteau theorem (Theorem 2.1 from
\cite{billingsley2013convergence}) gives the following partial
converse.

\begin{lemma}
\label{lem:weakconvergenceimpliesdyadicconvergence}
Let $\nu_n \wto \mu$. Then $\lim_{n \to \infty} \nu_{n}(I)=\mu(I)$ for
dyadic interval $I=[d_1,d_2)$ if $\mu(\{d_1\})=\mu(\{d_2\})=0$.
\end{lemma}

We characterize convergence of exponential sums in terms of weak
convergence of probability measures, first on $\mathbb{T}$ and then on
the Cantor space $\Sigma^\infty$. Unlike Theorem
\ref{thm:weyls_criterion_on_torus}, the result on $\Sigma^\infty$ does
not follow immediately from that on $\mathbb{T}$. On $\mathbb{T}$, the
following theorem holds due to Prokhorov theorem, the fact that
continuous functions on $\T$ can be approximated uniformly using
trigonometric polynomials, and that Fourier coefficients of measures
over $\T$ are unique due to Bochner's theorem (see Theorem 4.19 from
\cite{folland2016course}).

\begin{theorem}
\label{thm:weakconvergenceequivalencetorus}
  Let $r \in \mathbb{T}$ and let $\langle \nu_n \rangle_{n=1}^{\infty}$ be the sequence of
  averages of Dirac measures constructed out of $\langle 2^n r \Mod
  1\rangle_{n =0}^{\infty}$. Let $\langle n_m \rangle_{m \in \N}$ be any subsequence of natural numbers. Then the following are equivalent.
  \begin{enumerate}
    \itemsep=0em
    \item For every $k \in \Z$, there is a $c_k \in \mathbb{C}$ such that
      $\lim_{m \to \infty} \frac{1}{n_m} \sum_{j=0}^{n_m-1} e^{2\pi i k
      (2^j r)}=c_k$.
    \item There is a unique measure $\mu$ such that $\nu_{n_m} \wto
        \mu$.
  \end{enumerate}
  Furthermore, if any of the above conditions are true, then $c_k=\int e^{2\pi i ky} d\mu$ for every $k \in \Z$ and $\mu$ is the unique measure on $\T$ having Fourier coefficients $\langle c_k \rangle_{k \in \Z}$.
\end{theorem}
\begin{proof}
Suppose condition 1 holds. Let $\mu$ be any subsequence weak limit of $\langle \nu_{n_m} \rangle_{m \in \N}$ which exists due to Prokhorov's Theorem. Since $e^{2 \pi i k y}$ is a continuous function on $\T$, from the definition of weak convergence it follows that $\frac{1}{n_m} \sum_{j=0}^{n_m-1} e^{2\pi i k (2^j r)}=\int e^{2\pi i ky} d\nu_{n_m}$ converges along a subsequence to $\int e^{2\pi i ky} d\mu$. Since 1 is true we get that $\lim_{m \to \infty} \frac{1}{n_m} \sum_{j=0}^{n_m-1} e^{2\pi i k
      (2^j r)}=\int e^{2\pi i ky} d\mu=c_k$ for every $k \in \Z$. Since, every continuous function $f$ on $\T$ can be approximated uniformly using trigonometric polynomials (see Corollary 5.4 from \cite{SteinShakarchiFourier}), using routine approximation arguments we get that $\lim_{m \to \infty} \int f d\nu_{n_m} =\int f d\mu$ for every continuous function $f$ on $\T$. Hence, $\nu_{n_m} \wto \mu$. Conditions 1 easily follows from 2 since for every $k \in \Z$, $e^{2\pi i k y}$ is a continuous function on $\T$.
\end{proof}

We require an analogue of this theorem for Cantor space. But the proof
above cannot be adapted because on Cantor space, there are continuous
functions which cannot be approximated uniformly using trigonometric
polynomials. For example, consider $\chi_{C_0}$. Observe that $\chi_{C_0}(0^\infty)=1 \neq
0=\chi_{C_0}(1^\infty)$. But since $v(0^\infty) = v(1^\infty)$, every
trigonometric polynomial has the same value on $0^\infty$ and
$1^\infty$. However, we recover the analogue by handling dyadic
rational sequences and other sequences in separate cases. Since the
set of all finite unions of cylinder sets in $\Sigma^\infty$ is closed
under finite intersections and since the characteristic functions of
cylinder sets are continuous on the Cantor space, we get the following
analogue of Lemma \ref{lem:dyadicconvergenceimpliesweakconvergence}
and \ref{lem:weakconvergenceimpliesdyadicconvergence} using Theorem
2.2 from \cite{billingsley2013convergence}.

\begin{lemma}
\label{lem:cylindersetconvergenceandweakconvergence}
For a sequence of measures $\langle \nu_n \rangle_{n \in \N}$ on $\Sigma^\infty$, $\nu_n \wto \mu$ if and only if $\lim_{n \to \infty} \nu_n(C_w) =
\mu(C_w)$ for every $w \in \Sigma^*$. 
\end{lemma}

In the following theorems we relate the convergence of measures of
cylinder sets to the convergence of Weyl averages on the Cantor space
using Theorem \ref{thm:weakconvergenceequivalencetorus} and Lemma
\ref{lem:cylindersetconvergenceandweakconvergence}. We state these theorems for convergence along any subsequence, since we require these more general results for studying the subsequence limits of Weyl averages.  

\begin{theorem}
\label{thm:cylinderconvergenceimpliesweylconvergence}
Let $x \in \Sigma^\infty$ and $\langle \nu_n \rangle_{n=1}^{\infty}$ be the sequence of averages
of Dirac measures on $\Sigma^\infty$ constructed out of $\langle T^n x
\rangle_{n=0}^{\infty}$. Let $\langle n_m \rangle_{m \in \N}$ be any subsequence of natural numbers. If $\lim_{m \to \infty}\nu_{n_m}(C_w)=\mu(C_w)$ for every $w \in \Sigma^*$, then for every $k \in \Z$,
\begin{align*}
\lim_{m \to \infty} \frac{1}{n_m} \sum_{j=0}^{n_m-1} e^{2\pi i k v(T^j x)}=\int e^{2\pi i k v(y)} d\mu.	
\end{align*}
\end{theorem}
Observe that $n_m^{-1} \sum_{j=0}^{n_m-1} e^{2\pi i k v(T^j x)}=\int e^{2\pi i k v(y) }d\nu_{n_m}$. Hence, the above claim follows from Lemma
\ref{lem:cylindersetconvergenceandweakconvergence} and the definition
of weak convergence since for every $k \in \Z$, $e^{2\pi i k v(y)}$ is
a continuous function on $\Sigma^\infty$\footnote{This follows easily
by observing that the valuation map $v:\Sigma^\infty \to \T$ is a
continuous function on $\Sigma^\infty$.}. While Fourier coefficients
uniquely determine measures over $\T$, Bochner's Theorem does not hold
over $\Sigma^\infty$. For example let $\mu_1=\delta_{0^\infty}$ and
let $\mu_2=\delta_{1^\infty}$. Then $\mu_1 \ne \mu_2$, but it is easy
to verify that for any $k \in \Z$, $ \int e^{2 \pi i k v(y) } d\mu_1 =
e^{2 \pi i k v( 0^\infty) } = 1 = e^{2 \pi i k v( 1^\infty) } = \int
e^{2 \pi i k v(y) } d\mu_2$. We need the following lemma to obtain
a converse of Theorem
\ref{thm:cylinderconvergenceimpliesweylconvergence}.

\begin{lemma}
  \label{lem:nonzerodyadicpointsarenotlimitpoints}
  Let $x \in \Sigma^\infty$ such that $v(x) \not \in \mathbb{D}$ and
  let $\langle \nu'_n \rangle_{n=1}^{\infty}$ be the sequence of
  averages of Dirac measures on $\T$ constructed out of the sequence
  $\langle 2^n v(x) \Mod 1 \rangle_{n=0}^{\infty}$. Let $d$ be any non-zero dyadic rational. If $\nu'_{n_m}
  \wto \mu'$ for some subsequence of natural numbers $\langle n_m
  \rangle_{m \in \N}$, then $\mu'(\{d\})=0$.
  \end{lemma}
  \begin{proof}
For every $k$, let $U_k$ be the interval $\left( d-1/2^{k+1},d+1/2^{k+1} \right)$. Let $w_1$ be the unique string ending with $0$ such that $v(w_1 1^\infty)=d$ and let $w_2$ be the unique string ending with $1$ such that $v(w_2 0^\infty)=d$ . Since $v(x) \not\in \mathbb{D}$, $v(x)$ has a unique binary expansion which is the sequence $x$. If $2^n v(x) \Mod 1 \in U_k$, then either $x_n^{n+k-1}=(w_1 1^\infty)_0^{k-1}$ or $x_n^{n+k-1}=(w_2 0^\infty)_0^{k-1}$. Let us consider the case when $x_n^{n+k-1}=(w_1 1^\infty)_0^{k-1}$. Since $d \neq 0$, we have $w_1 \neq 1^{\lvert w_1 \rvert}$ and hence $x_{n+i}^{n+i+k-1} \neq (w_1 1^\infty)_0^{k-1}$ and $x_{n+i}^{n+i+k-1} \neq (w_2 0^\infty)_0^{k-1}$ for any $i$ between $\max\{\lvert w_1 \rvert,\lvert w_2 \rvert\}$ and $k-\max\{\lvert w_1 \rvert,\lvert w_2 \rvert\}$. Therefore, $2^{n+i} v(x) \Mod 1 \not\in U_k$ for any $i$ between $\max\{\lvert w_1 \rvert,\lvert w_2 \rvert\}$ and $k-\max\{\lvert w_1 \rvert,\lvert w_2 \rvert\}$. Since this is true for any $n$, we get that for any $k \geq \max\{\lvert w_1 \rvert,\lvert w_2 \rvert\}$,
\begin{align*}
  	\limsup_{n \to \infty}\frac{\#\{2^i v(x) \Mod 1 \in U_k \mid 0 \leq i \leq n-1\}}{n} \leq \frac{2 \max\{\lvert w_1 \rvert,\lvert w_2 \rvert\}}{k}.
  	\end{align*}
 We get the same bound in the case when $x_n^{n+k-1}=(w_2 0^\infty)_0^{k-1}$. Hence for any $k \geq \max\{\lvert w_1 \rvert,\lvert w_2 \rvert\}$,
\begin{align*}
\limsup_{m \to \infty} \int \chi_{U_k} d\nu_{n_m} & = \limsup_{m \to \infty}\frac{\#\{2^i v(x) \Mod 1 \in U_k \mid 0 \leq i \leq n_m-1\}}{n_m} \\
&\leq \limsup_{n \to \infty}\frac{\#\{2^i v(x) \Mod 1 \in U_k \mid 0 \leq i \leq n-1\}}{n} \\
&\leq \frac{C}{k}.
\end{align*}
where the constant $C$ only depends on $d$. Now we define a sequence of functions $f_k$ as follows.
\begin{align*}
f_k(x) = \begin{cases}
 	1-2^{k+1}(x-d) & \text{if } d \leq x \leq d+1/2^{k+1}\\
 	1+2^{k+1}(x-d) &\text{if } d-1/2^{k+1} \leq x \leq d\\
 	0 & \text{otherwise.}
 \end{cases}	
\end{align*}
Each $f_k$ is a continuous function on $\T$. Since $\lVert f_k \rVert_\infty \leq 1$ for every $k$, we have $f_k \leq \chi_{U_k}$ and using this inequality we get that,
\begin{align*}
\mu'(\{d\}) &= \int \chi_{\{d\}} d\mu' \leq  \int f_k d\mu' \\
&= \lim_{m \to \infty } \int f_k d\nu_{n_m} \\
&= \limsup_{m \to \infty } \int f_k d\nu_{n_m} \\
&\leq \limsup_{m \to \infty } \int \chi_{U_k} d\nu_{n_m} \\
&\leq \frac{C}{k}.
\end{align*}
Since the above bound is true for every for any $k \geq \max\{\lvert w_1 \rvert,\lvert w_2 \rvert\}$, we get that $\mu'(\{d\}) =0$.
\end{proof}

 Using Lemma \ref{lem:weakconvergenceimpliesdyadicconvergence}, Theorem
 \ref{thm:weakconvergenceequivalencetorus} and
 Lemma \ref{lem:nonzerodyadicpointsarenotlimitpoints} we obtain the following
 partial converse of Theorem
 \ref{thm:cylinderconvergenceimpliesweylconvergence}.

\begin{theorem}
\label{thm:weylconvergenceimpliescylinderconvergence}
Let $x \in \Sigma^\infty$ and let $\langle n_m \rangle_{m \in \N}$ be any subsequence of natural numbers. Let $\langle c_k \rangle_{k \in \Z}$ be complex numbers such that
$
\lim_{m \to \infty} \frac{1}{n_m} \sum_{j=0}^{n_m-1} e^{2\pi i k v(T^j x)} = c_k
$
for every $k \in \Z$. Then there exists a unique measure $\mu$ on $\T$ having Fourier coefficients $\langle c_k \rangle_{k \in \Z}$ and $\lim_{m \to \infty}\nu_{n_m}(C_w)=\hat{\mu}(C_w)$ for every $w \in \Sigma^*$ such that $w \neq 1^{\lvert w \rvert}$ and $w \neq 0^{\lvert w \rvert}$.
\end{theorem}
\begin{proof}
	We first consider the case when $v(x)$ is a dyadic rational in $\T$. In this case, it is easy to verify that for every $k \in \Z$,
$\lim_{n \to \infty} \frac{1}{n} \sum_{j=0}^{n-1} e^{2\pi i k v(T^j x)} =1$. The unique measure on $\T$ having all Fourier coefficients equal to $1$ is $\mu=\delta_{0}$ and we have $\mu'=\delta_{0^\infty}$. In this case it is easy to verify that for every $w$ that is not equal to $1^{\lvert w \rvert}$ or $0^{\lvert w \rvert}$, $\lim_{m \to \infty}\nu_{n_m}(C_w)=0=\hat{\mu}(C_w)$. 
Now, we consider the case when $v(x)$ is not a dyadic rational in $\T$. In this case we have that $v(T^j x)$ is not a dyadic rational for all $j \geq 0$. In this case, it is easily verified that $v(T^j x)=2^j v(x) \Mod 1$ for all $j \geq 0$. This gives us the following equality,
	\begin{align}\label{eqn:weylconvergenceimpliescylinderconvergence1}
		\frac{1}{n_m}\sum_{j=0}^{n_m-1} e^{2 \pi i k (v(T^j x)) } = \frac{1}{n_m}\sum_{j=0}^{n_m-1} e^{2 \pi i k (2^j v(x)) }.
	\end{align}
	Let $\langle \nu'_{n} \rangle_{n=1}^{\infty}$ be the sequence of averages of Dirac measures on $\T$ constructed out of the sequence $\langle 2^{n} v(x) \Mod 1  \rangle_{n=0}^{\infty}$. From \ref{eqn:weylconvergenceimpliescylinderconvergence1} and Theorem \ref{thm:weakconvergenceequivalencetorus}, we get that $\nu'_{n_m} \wto \mu$ where $\mu$ is the unique measure on $\T$ having Fourier coefficients $\langle c_k \rangle_{k \in \Z}$. 
	
	By definition, $\nu_{n_m} = \frac{1}{n_m} \sum_{j=0}^{n_m-1} \delta_{T^j x}$. For any $w \in \Sigma^\infty$, $\delta_{T^j x}(C_w)$ is $1$ if and only if $T^j x \in C_w$. Since, $v(T^j x)$ is not a dyadic rational for all $j \geq 0$, we get that $T^j x \in C_w$ if and only if $2^j v(x) \Mod 1 \in I_w$. This is because $v(T^j x) < v(w1^\infty)$ since $v(T^j x)$ is not a dyadic rational and $T^j x \in C_w$.  Then, the last observation lets us conclude that,
	\begin{align}\label{eqn:weylconvergenceimpliescylinderconvergence2}
	\nu_{n_m}(C_w) = \nu'_{n_m}(I_w)	
	\end{align}
	for all $m \geq 1$.
	
	Let $w$ be any string such that $w \neq 1^{\lvert w \rvert}$ and $w \neq 0^{\lvert w \rvert}$. Using Lemma \ref{lem:nonzerodyadicpointsarenotlimitpoints}, we get that $\mu(\{v(w0^\infty)\})=\mu(\{v(w1^\infty)\})=0$. Since $v(w0^\infty)$ and $v(w1^\infty)$ are the end points of $I_w$, using Lemma \ref{lem:weakconvergenceimpliesdyadicconvergence} we get that $\lim_{m \to \infty}\nu'_{n_m}(I_w)=\mu(I_w)$. Hence, from \ref{eqn:weylconvergenceimpliescylinderconvergence2}, we get that $\lim_{m \to \infty}\nu_{n_m}(C_w)=\lim_{m \to \infty}\nu'_{n_m}(I_w)=\mu(I_w)=\hat{\mu}(C_w)$. The proof of the claim is thus complete.
\end{proof}

For any $x \in \Sigma^\infty$, let $\langle \nu_n
\rangle_{n=1}^{\infty}$ be the sequence of averages of Dirac measures
on $\Sigma^\infty$ constructed out of the sequence $\langle T^n x
\rangle_{n=0}^{\infty}$. Now, for any $A \in
\mathcal{B}(\Sigma^\infty)$, $\nu_n (A)$ is the proportion of elements
in the finite sequence $x,Tx,T^2x, \dots T^{n-1}x$ which falls inside
the set $A$. From this remark, and the definitions of $\nu_n$ and the
sliding count probability $P$, the following lemma follows easily.

\begin{lemma}
\label{lem:deltameasuresandcoutning}
Let $w$ be any finite string in $\Sigma^*$ and let $l=\lvert w
\rvert$. Let $x$ be any element in $\Sigma^\infty$. If $\langle
\nu_n\rangle_{n=1}^{\infty}$ is the sequence of averages of Dirac
measures over $\Sigma^\infty$ constructed out of the sequence $\langle
T^n x\rangle_{n=0}^{\infty}$. Then for any $n$, $
\nu_n(C_w)=P(x_0^{n+l-2},w).  $
\end{lemma}
\begin{proof}
From the definition, $\nu_n(C_w)$ is the proportion of elements in the finite sequence $\langle T^n x \rangle_{i=0}^{n-1}$ which begins with the string $w$. This is equal to $P(x_0^{n+l-2},w)$ by the definition of $P$.
\end{proof}

We now give a new characterization of Weyl's criterion on Cantor Space
(Theorem \ref{thm:weyls_criterion_on_cantor_space}) in terms of weak
convergence of measures. In later sections, we generalize this to
characterize finite-state dimension in terms of exponential sums.

\begin{theorem}[Weyl's criterion on $\Sigma^\infty$ and weak convergence]
\label{thm:weylcriterioncantor}
  Let $x \in \Sigma^\infty$, and $\langle\nu_n \rangle_{n=1}^{\infty}$ be the sequence of averages of Dirac
  measures constructed out of $\langle T^n x \rangle_{n =0}^{\infty}$, and
  $\mu$ be the uniform measure on $\Sigma^\infty$. Then
  the following are equivalent.
  \begin{enumerate}
\itemsep=0em
  \item $x$ is normal.
  \item For every $w \in \Sigma^*$, the sliding block frequency
    $P(x_0^{n-1},w) \to 2^{-\lvert w \rvert}$ as $n \to \infty$.
  \item For every $k \in \Z$, $\lim_{n \to \infty} \frac{1}{n}
    \sum_{j=0}^{n-1} e^{2\pi i k v(T^j x)} = 0$.
  \item $\nu_n \wto \mu$.
  \end{enumerate}
\end{theorem}
\begin{proof}
  1 and 2 are equivalent by definition. The equivalence of 2 and 4 follows from Lemma \ref{lem:deltameasuresandcoutning} and Lemma \ref{lem:cylindersetconvergenceandweakconvergence}. 2 $\implies$ 3 follows directly from Theorem \ref{thm:cylinderconvergenceimpliesweylconvergence} and Lemma \ref{lem:deltameasuresandcoutning}. Now, we prove 3 $\implies$ 2. The uniform distribution $\mu'$ on $\T$ is the unique measure having all Fourier coefficients equal to $0$. Let $\langle \nu'_{n} \rangle_{n=1}^{\infty}$ be the sequence of averages of Dirac measures on $\T$ constructed out of the sequence $\langle 2^{n} v(x) \Mod 1  \rangle_{n=0}^{\infty}$. From Theorem \ref{thm:weakconvergenceequivalencetorus}, we get that $\nu'_n \wto \mu'$. Since $\mu'(\{y\})=0$ for any $y \in \T$, using Lemma \ref{lem:weakconvergenceimpliesdyadicconvergence} we get that $\nu'_n(I_w) \to \mu'(I_w)$ as $n \to \infty$ for every $w \in \Sigma^*$. If $v(x)$ is a dyadic rational, it is easy to verify that the Weyl averages converges to $1$ for every $k \in \Z$. Hence, $v(x)$ must not be a dyadic rational. As in the proof of Theorem \ref{thm:weylconvergenceimpliescylinderconvergence}, this implies that $\nu_{n}(C_w) = \nu'_{n}(I_w)$ for all $m \geq 1$ and every $w \in \Sigma^*$. From the previous observations, we get that $\lim_{n \to \infty} \nu_{n}(C_w) = \lim_{n \to \infty} \nu'_{n}(I_w) = \mu'(I_w) = 2^{-\lvert w \rvert}$ for every $w \in \Sigma^*$. Finally, we get $\lim_{n \to \infty} P(x_0^{n-1},w) = 2^{-\lvert w \rvert}=\mu(C_w)$ using Lemma \ref{lem:deltameasuresandcoutning}.
\end{proof}
  
\section{Divergence of exponential sums for non-normal numbers}
\label{sec:counterexampletoconvergence}
Weyl's criterion says that when $\dim_{FS}(x)=\Dim_{FS}(x)=1$ the
averages of the exponential sums for every $k$ converges to
$0$. However for $x$ with $\dim_{FS}(x) < 1$, the situation is
different. It is easy to construct a sequence $a$ with
$\dim_{FS}(a)<1$ and a $k\in\Z$ such that the sequence of Weyl
averages with parameter $k$ do not converge. It is natural to ask if
the condition $\dim_{FS}(x)=\Dim_{FS}(x)$ is sufficient to guarantee
convergence of the exponential sum averages. But we construct an $x$
with $\dim_{FS}(x)=\Dim_{FS}(x)=\frac{1}{2}$ such that for some $k$,
the sequence $\langle \sum_{j=0}^{n-1}e^{2 \pi i k (v(T^j x))}/n
\rangle_{n=1}^{\infty}$ diverges.

Entropy rates converging to a limit does not imply that the empirical
probability measures converge to a limiting distribution, and it is
the latter notion which is necessary for exponential sums to converge.
\begin{lemma}
\label{lem:nonconvergentexample}
There exists $x \in \Sigma^\infty$ with $\dim_{FS}(x)=\Dim_{FS}(x)=\frac{1}{2}$ such that for some $k \in \Z$, the sequence $\langle \sum_{j=0}^{n-1}e^{2 \pi i k (v(T^j x))}/n \rangle_{n=1}^{\infty}$ is not convergent.
\end{lemma}

The analogue of Lemma \ref{lem:subadditivityofentropy} need not hold in the setting of disjoint block entropies. However, an analogue of \ref{lem:subadditivtycorollary1} can be obtained in the disjoint setting, which we require in the proof of Lemma \ref{lem:nonconvergentexample}.

\begin{lemma}
\label{lem:disjointblockentropyinequality}
For any $l$ and $k$, $H^d_{kl}(x_0^{n-1}) \leq H^d_{l}(x_0^{n-1}) + o(n)/n$ where the speed of convergence of the error term depends only on $k$ and $l$.
\end{lemma}

\begin{proof}
	It is enough to show that $(kl)H^d_{kl}(x_0^{n-1}) \leq k(lH^d_l(x_0^{n-1})) +o(n)/n$. The required inequality follows by dividing both sides by $kl$. If $(kl)H^d_{kl}(x_0^{n-1}) \leq k(lH^d_l(x_0^{n-1}))$ is true for all $n$ such that $kl \vert n$ then by continuity of entropy we obtain $(kl)H^d_{kl}(x_0^{n-1}) \leq k(lH^d_l(x_0^{n-1})) + o(n)/n$. Due to the uniform continuity of entropy, the speed of convergence of the error term depends only on $k$ and $l$. Hence, without loss of generality we assume that $kl \vert n$.
	
	From the definition of disjoint block entropy,
	\begin{align*}
		(kl)H^d_{kl}(x_0^{n-1}) = -\sum_{w \in \Sigma^{kl}} P^d(x_0^{n-1},w)\log(P^d(x_0^{n-1},w)).
	\end{align*}
	For $0 \leq j \leq k-1$ and $w \in \Sigma^l$, define $P_j^d (x_0^{n-1},w)$ to be the fraction of $kl$-length disjoint blocks in $x_0^{n-1}$ such that within the block, $w$ appears as the $(j+1)$\textsuperscript{th} disjoint $l$-length block from the left. Formally,
	\begin{align*}
		P^d_j (x_0^{n-1},w) = \frac{\lvert \{0\leq i < n/kl :  x_{kli+jl}^{kli+jl+l-1} = w\}\rvert	}{n/kl} .
	\end{align*}
	Let us define corresponding entropies,
	\begin{align*}
	\widehat{H}_j(x_0^{n-1}) = 	-\frac{1}{l}\sum_{w \in \Sigma^l} P^d_j(x_0^{n-1},w)\log(P^d_j(x_0^{n-1},w))
	\end{align*}
	for $0 \leq j \leq k-1$. 
	
	Using the subadditivty of Shannon entropy, it follows that,
 	\begin{align}
 	\label{eqn:disjointlemmaeq1}
 	(kl)H_{kl}^d (x_0^{n-1}) \leq \sum_{j=0}^{k-1} l\widehat{H}_j(x_0^{n-1}).	
 	\end{align}
	Since $kl \vert n$, it can be seen from the definitions that for any $w \in \Sigma^l$,
	\begin{align*}
	P^d(x_0^{n-1},w) =\frac{1}{k} \sum_{j=0}^{k-1} P^d_j(x_0^{n-1},w).	
	\end{align*}
	Using the concavity of the function $x \log(1/x)$, it follows that,
	\begin{align}
	\label{eqn:disjointlemmaeq2}
	l \widehat{H}_l^d(x_0^{n-1}) \geq \frac{1}{k} \sum_{j=0}^{k-1} l\widehat{H}_j (x_0^{n-1}).
	\end{align}
	From \ref{eqn:disjointlemmaeq1} and \ref{eqn:disjointlemmaeq2} it follows that,
	\begin{align*}
	(kl)H^d_{kl}(x_0^{n-1}) \leq k(lH^d_l(x_0^{n-1})).	
	\end{align*}

\end{proof}

Now we prove Lemma \ref{lem:nonconvergentexample}.
\begin{proof}[Proof of Lemma \ref{lem:nonconvergentexample}]
	Let $y \in \Sigma^\infty$ be a fixed normal sequence. Define $a \in
\Sigma^\infty$ by $a_{2n}=0$, $a_{2n+1}=y_{n}$, for all $n \in \N$. Define $b \in \Sigma^\infty$ by $b_{4n}=b_{4n+3}=0$, and $b_{4n+1}=y_{2n}$, $b_{4n+2}=y_{2n+1}$, for all $n
\in \N$.  Equivalently,  $a$ is constructed by repeating the pattern $0\star 0\star$ infinitely many times and replacing the $\star$ symbols with successive digits from $y$. Similarly, $b$ is constructed by repeating the pattern $0\star \star 0$ infinitely many times and replacing the $\star$ symbols with successive digits of $y$. The
sliding block frequency of $01$ in $a$ is easily verified to be equal to $1/4$, whereas it is equal to $3/16$ in $b$.
	
	For any $l\geq 2$, consider the $2^l$ length disjoint blocks in the sequence $a$. It is easily verified that by the construction of $a$, $2^{2^l/2}$ different strings of length $2^l$ occurs in $a$ with equal probabilities as the number of blocks goes to infinity. Hence, for every positive $\epsilon$, positive integer $l$ and finite string $\alpha$ with $2^l \vert \lvert \alpha \rvert$, there exists an integer $M_l^{\alpha, a} (\epsilon)$ such that for all $n \geq M_l^{\alpha, a} (\epsilon)$,
	\begin{align*}
	H^d_{2^l}\left((\alpha a)_0^{n-1}\right) \geq \frac{1}{2}-\epsilon.
	\end{align*}
	Such a number exists for $l=1$ also, due to Lemma \ref{lem:disjointblockentropyinequality}. Due to similar reasons, analogous quantities exist for the sequence $b$ which we denote using $M_l^{\alpha, b} (\epsilon)$. Since the speed of convergence of the error term in Lemma \ref{lem:disjointblockentropyinequality} is independent of the string, using the same lemma,  for every $i$, there exists $J_i$ such that for any string $z$ and for all $n \geq J_i$,
	\begin{align*}
	H_{2^i}^d(z_0^{n-1}) \geq H_{2^{i+1}}^d (z_0^{n-1}) - \frac{1}{2^{i+1}}.	
	\end{align*}
	For every positive $\epsilon$ and a finite string $\alpha$ of even length, there exists an integer $L^{\alpha,a}(\epsilon)$ such that for all $n \geq L^{\alpha,a}(\epsilon)$,
	\begin{align*}
	P\left((\alpha a)_0^{n-1},01\right) \geq \frac{1}{4}-\epsilon.	
	\end{align*}
	Similarly, there exists an integer $L^{\alpha,b}(\epsilon)$ such that for all $n \geq L^{\alpha,b}(\epsilon)$,
	\begin{align*}
	P\left((\alpha b)_0^{n-1},01\right) \leq \frac{3}{16}+\epsilon.	
	\end{align*}
	We construct $x$ by specifying longer and longer prefixes of $x$ in a stage-wise manner. Initially, let the prefix constructed until stage $0$ be $\sigma=\lambda$. 
	
	In stage $i$, if $i$ is odd, we do the following. For a fixed
        $i$, let $K_i$ be a large enough integer such that for all $n
        \geq K_i$
	\begin{align}
	\label{eqn:kiinequality}
	\frac{\frac{\lvert \sigma \rvert+n}{2^i}}{\frac{\lvert \sigma \rvert+n}{2^i} + M_i^{\lambda,b}(\frac{1}{2^i})}\left( \frac{1}{2}-\frac{1}{2^{i+1}} \right) \geq 	\frac{1}{2}-\frac{1}{2^{i}}.
	\end{align}
	Let $N_i$ be any integer greater than $\max\{M_i^{\sigma,
          a}(2^{-(i+1)})-\lvert \sigma
        \rvert,J_i,L^{\sigma,a}(2^{-i}), 2^i K_i \}$ such that
        $2^{i+1} $ divides $ \lvert \sigma \rvert + N_i $. Let
        $\sigma_i$ be the $N_i$ length prefix of $a$. We attach
        $\sigma_i$ to the end of the string $\sigma$ constructed until
        stage $i-1$. Now, we set $\sigma$ equal to this longer string
        $\sigma \sigma_i$. If $i$ is even, we perform the same steps
        as above by interchanging the roles of $a$ and $b$. We set $x$
        to be the infinite sequence obtained by concatenating
        $\sigma_i$s, i.e, $x=\sigma_1 \sigma_2 \sigma_3 \sigma_4
        \dots$. Now, we show that $x$ satisfies the required
        properties.
	
	It can be easily seen that $x$ satisfies conditions \ref{item:nicondition} and \ref{item:njcondition} since each $N_i \geq L^{\sigma,a}(2^{-i})$ (or $L^{\sigma,b}(2^{-i})$ if $i$ is even). This forces the slide count probability of $01$ in $\sigma_1\sigma_2\sigma_3 \dots \sigma_i$ to be at least $1/4-2^{-i}$ in odd stages and at most $3/16 + 2^{-i}$ in even stages.
	
	Now, we show that $\dim_{FS}(x)=\Dim_{FS}(x)=1/2$.
 Towards this end, we first show that for any $i$,
\begin{align}
\label{eqn:h2ixequation}
H_{2^i}^d (x_0^{n-1}) \geq \frac{1}{2}-\frac{1}{2^i} 	
\end{align}
provided that $n \geq \lvert \sigma_1\sigma_2\sigma_3 \dots \sigma_i \rvert$. For any $\alpha \in \Sigma^*$ and $\beta \in
\Sigma^*$, we write $\alpha \sqsubseteq \beta$ if
$\alpha$ is a prefix of $\beta$. In order to show \ref{eqn:h2ixequation}, it is enough to show that for any $i$ and $\alpha$ such that $\sigma_1\sigma_2\sigma_3 \dots \sigma_i \sqsubseteq \alpha$ and $\alpha \sqsubseteq \sigma_1\sigma_2\sigma_3 \dots \sigma_i \sigma_{i+1}$,
\begin{align}
\label{eqn:h2ialphaequation}
H_{2^i}^d (\alpha) \geq \frac{1}{2}-\frac{1}{2^i}.
\end{align}
If \ref{eqn:h2ialphaequation} holds for all $i$, then \ref{eqn:h2ixequation} holds for all $i$. This is because if $k$ is the number such that $\lvert \sigma_1\sigma_2\sigma_3 \dots \sigma_{i+k} \rvert \leq n \leq \lvert \sigma_1\sigma_2\sigma_3 \dots \sigma_{i+k+1} \rvert$ then,
\begin{align*}
H_{2^{i+k}}^d (x_0^{n-1}) \geq \frac{1}{2}-\frac{1}{2^{i+k}}
\end{align*}
due to \ref{eqn:h2ialphaequation}. Now, since $\lvert \sigma_1\sigma_2\sigma_3 \dots \sigma_{i+k} \rvert \geq \lvert \sigma_{i+k-1}\rvert = N_{i+k-1} \geq J_{i+k-1}$,
\begin{align*}
H_{2^{i+k-1}}^d (x_0^{n-1}) &\geq H_{2^{i+k}}^d (x_0^{n-1}) - \frac{1}{2^{i+k}}\\
&\geq \frac{1}{2}-\frac{1}{2^{i+k-1}}.
\end{align*}
Continuing this process we get,
\begin{align*}
H_{2^{i}}^d (x_0^{n-1})	\geq \frac{1}{2}-\frac{1}{2^i}.
\end{align*}
This proves the claim in \ref{eqn:h2ixequation}. Now, we prove the claim in \ref{eqn:h2ialphaequation}. If $\alpha=\sigma_1 \sigma_2 \sigma_3 \dots \sigma_i$ then,
\begin{align}
\label{eqn:strongh2ialphaequation}
H_{2^i}^d (\alpha) \geq \frac{1}{2}-\frac{1}{2^{i+1}}.
\end{align}
This is because $\alpha$ is $\sigma_1 \sigma_2 \sigma_3 \dots
\sigma_{i-1}$ concatenated with the first $N_i$ bits of $a$ if $i$ is
odd. And, $\alpha$ is $\sigma_1 \sigma_2 \sigma_3 \dots \sigma_{i-1}$
concatenated with the first $N_i$ bits of $b$ if $i$ is even. If $i$
is odd, since
\begin{align*}
N_i \geq M_i^{\sigma_1 \sigma_2 \sigma_3 \dots \sigma_{i-1}, a}\left (\frac{1}{2^{i+1}}\right)-\lvert \sigma_1 \sigma_2 \sigma_3 \dots \sigma_{i-1} \rvert
\end{align*}
equation \ref{eqn:strongh2ialphaequation} follows due to the definition of $M_i^{\sigma_1 \sigma_2 \sigma_3 \dots \sigma_{i-1}, a}(2^{-(i+1)})$. A similar argument works if $i$ is even with $a$ replaced with $b$. For $\alpha \neq \sigma_1 \sigma_2 \sigma_3 \dots \sigma_{i}$, there are,
\begin{align*}
\frac{\lvert \sigma_1 \sigma_2 \sigma_3 \dots \sigma_{i} \rvert}{2^i} = \frac{\lvert \sigma_1 \sigma_2 \sigma_3 \dots \sigma_{i-1} \rvert}{2^i} + \frac{\lvert \sigma_i \rvert}{2^i}	
\end{align*}
disjoint blocks of length $2^i$ within the first $\lvert \sigma_1 \sigma_2 \sigma_3 \dots \sigma_{i} \rvert$ bits of $\alpha$. Let $n_\alpha =\lvert \alpha \rvert-\lvert \sigma_1 \sigma_2 \sigma_3 \dots \sigma_{i} \rvert$ be the number of remaining digits of $\alpha$. Since $2^{i+1}$ divides $\lvert \sigma_1 \sigma_2 \sigma_3 \dots \sigma_{i} \rvert$, $2^i$ also divides $\lvert \sigma_1 \sigma_2 \sigma_3 \dots \sigma_{i} \rvert$.  Hence, there are $\lfloor n_\alpha /2^i \rfloor$ disjoint blocks of length $2^i$ in the $n_\alpha$ length suffix of $\alpha$ which is by construction, a prefix of $a$ or $b$ (depending on whether $i$ is odd or even). Then, due to the concavity of Shannon entropy, 
\begin{align}
\label{eqn:concavityblockcounting}
H_{2^i}^d (\alpha) \geq \frac{\frac{\lvert \sigma_1 \sigma_2 \sigma_3 \dots \sigma_{i-1} \rvert}{2^i} + \frac{\lvert \sigma_i \rvert}{2^i}}{\frac{\lvert \sigma_1 \sigma_2 \sigma_3 \dots \sigma_{i-1} \rvert}{2^i} + \frac{\lvert \sigma_i \rvert}{2^i}+ \lfloor \frac{n_\alpha} {2^i} \rfloor} H_{2^i}^d (\sigma_1 \sigma_2 \dots \sigma_i) + \frac{\lfloor \frac{n_\alpha} {2^i} \rfloor}{\frac{\lvert \sigma_1 \sigma_2 \sigma_3 \dots \sigma_{i-1} \rvert}{2^i} + \frac{\lvert \sigma_i \rvert}{2^i}+ \lfloor \frac{n_\alpha} {2^i} \rfloor} H_{2^i}^d (b_0^ {n_\alpha-1})	
\end{align}
Above we assumed that $i$ is odd. In the even case, $H_{2^i}^d (b_0^{n_\alpha-1})$ must be replaced by $H_{2^i}^d (a_0^{n_\alpha-1})$. Now if $n_\alpha \leq M_i^{\lambda,b}(\frac{1}{2^i})$, considering only the first term on the right we have,
\begin{align*}
H_{2^i}^d (\alpha) &\geq 	\frac{\frac{\lvert \sigma_1 \sigma_2 \sigma_3 \dots \sigma_{i-1} \rvert}{2^i} + \lfloor \frac{N_i}{2^i} \rfloor}{\frac{\lvert \sigma_1 \sigma_2 \sigma_3 \dots \sigma_{i-1} \rvert}{2^i} + \lfloor \frac{N_i}{2^i} \rfloor+ M_i^{\lambda,b}(\frac{1}{2^i})}\left( \frac{1}{2}-\frac{1}{2^{i+1}} \right) \\
&\geq \frac{1}{2}-\frac{1}{2^i}.
\end{align*}
The last inequality follows from \ref{eqn:kiinequality} since $N_i \geq 2^i K_i$. If $n_\alpha > M_i^{\lambda,b}(2^{-i})$ then,
\begin{align*}
H_{2^i}^d (b_0^{n_\alpha-1}) &\geq \frac{1}{2}-\frac{1}{2^i}
\end{align*}
from the definition of $M_i^{\lambda,b}(2^{-i})$. Now, using the above inequality and the inequality in \ref{eqn:strongh2ialphaequation} in \ref{eqn:concavityblockcounting} we get that \ref{eqn:h2ialphaequation} is true in the case when $n_\alpha > M_i^{\lambda,b}(2^{-i})$. The proof of the claim in \ref{eqn:h2ialphaequation} is thus complete which in turn implies that \ref{eqn:h2ixequation} is true. Having established \ref{eqn:h2ixequation}, it follows that $\dim_{FS}(x) \geq 1/2$.  

Now, we are left to show that $\Dim_{FS}(x) \leq 1/2$. By choosing $N_i$ to be such that $2^{i+1}$ divides $\lvert \sigma \rvert + N_i$, it is guaranteed that $2^{i+1}$ divides $\lvert \sigma_1 \sigma_2 \sigma_3 \dots \sigma_{i} \rvert$. Hence, for any fixed $i$, it follows that $2^i$ divides $\lvert \sigma_1 \sigma_2 \sigma_3 \dots \sigma_{k} \rvert$ for any $k \geq i-1$. Hence, for large enough $n$, the length $2^i$ disjoint blocks encountered in calculating the disjoint count probability $P_{2^i}^d (x_0^{n-1})$ shall be predominantly of the following two types. Either these blocks match the pattern $(0 \star 0 \star)^{2^{i-2}}$ or these blocks match the pattern $(0 \star \star 0)^{2^{i-2}}$. Occurrences of any block that does not match this pattern can only happen in the prefix $\sigma_1 \sigma_2 \sigma_3 \dots \sigma_{i-1}$. But these occurrences becomes negligible as $n$ becomes large. Hence,
\begin{align*}
H_{2^i}^d (x_0^{n-1}) &\leq \frac{\log(2^{2^i/2}+2^{2^i/2})}{2^i} + \frac{o(n)}{n} \\
&= \frac{2^i/2+1}{2^i} + \frac{o(n)}{n}. 	
\end{align*}
This implies that $$\limsup_{n \to \infty} H_{2^i}^d (x_0^{n-1}) \leq
(2^i/2+1)/2^i.$$ Using \ref{eqn:limitstrongdimension}, we get that that $\Dim_{FS}(x) = \lim_{i \to
  \infty}\limsup_{n \to \infty} H_{2^i}^d (x_0^{n-1})$.
%% SN - this refers to a commented discussion
%% due to \ref{eqn:limitstrongdimension}.
%% SN - end comment
Hence, it follows that,
\begin{align*}
\Dim_{FS}(x) &\leq \lim_{i \to \infty}	 \frac{2^i/2+1}{2^i}
= \frac{1}{2}.
\end{align*}
And hence, $x$ satisfies all the required properties.
\end{proof}

Generalizing the construction of diluted sequences in
\cite{dai2004finite}, we define an $x$ with $v(x)\in \T\setminus\mathbb{D}$
and $\dim_{FS}(x)=\Dim_{FS}(x)=1/2$, but where for some $k \in \Z$,
the sequences of Weyl sum averages diverge. The idea of dilution is as
follows. Let $y \in \Sigma^\infty$ be normal. Define $a \in
\Sigma^\infty$ by $a_{2n}=0$, $a_{2n+1}=y_{n}$, $n \in \N$. Then
$\dim_{FS}(a)=\Dim_{FS}(a)=1/2$. Note that $b \in \Sigma^\infty$
defined by $b_{4n}=b_{4n+3}=0$, and $b_{4n+1}=y_{2n}$,
$b_{4n+2}=y_{2n+1}$, $n \in \N$ is also a regular sequence with
$\dim_{FS}(b)=\Dim_{FS}(b)=1/2$. But, the sliding block frequency of
$01$ in $a$ is $1/4$, whereas it is $3/16$ in $b$. We leverage the
existence of such distinct sequences with equal dimension. The
disjoint blocks of $x$ alternate between the above two patterns in a
controlled manner to satisfy the following conditions.
\begin{enumerate}
  \itemsep=0em
\item \label{item:dimcondition} $\dim_{FS}(x)=\Dim_{FS}(x)=1/2$
\item \label{item:nicondition} There is an increasing sequence of
  indices $\langle n_i \rangle_{i=1}^{\infty}$ such that 
  $
  \lim_{i \to \infty} P(x_0^{n_i-1},01)=1/4.	
  $
\item \label{item:njcondition} There is an increasing sequence of
  indices $\langle n_i \rangle_{i=1}^{\infty}$ such that 
	$
	\lim_{j \to \infty} P(x_0^{n_j-1},01)=3/16.	
	$
\end{enumerate}

%% Let $\langle \nu_n \rangle_{n=1}^{\infty}$ be the sequence of
%% averages of Dirac average measures constructed out of the sequence
%% $\langle T^n x\rangle_{n=0}^{\infty}$. Let us assume that the Weyl
%% averages $\langle n^{-1} \sum_{j=1}^{n} e^{2 \pi i k (v(T^j x)) }
%% \rangle_{n=1}^{\infty}$ are convergent for every $k \in \Z$. Let
%% $\langle \nu'_n \rangle_{n=1}^{\infty}$ be the sequence of averages
%% of Dirac measures on $\T$ constructed out of the sequence $\langle
%% 2^n v(x) \Mod 1 \rangle_{n=0}^{\infty}$. Using the same steps in
%% the proof of Theorem
%% \ref{thm:weylconvergenceimpliescylinderconvergence}, we get that
%% $\nu_n' \wto \mu'$ where $\mu'$ is the unique measure on $\T$
%% having Fourier coefficients equal to the limits of the Weyl
%% averages. The $x$ that we construct will be such that $v(x) \in \T
%% \setminus \mathbb{D}$ and thus we can use Theorem
%% \ref{thm:weylconvergenceimpliescylinderconvergence} to infer that
%% $\nu(C_{01})$ is convergent. Using Lemma
%% \ref{lem:deltameasuresandcoutning}, we infer that $\lim_{n \to
%% \infty}P(x_0^{n-1},01)$ exists. But, we know from conditions
%% \ref{item:nicondition} and \ref{item:njcondition} that
%% $P(x_0^{n-1},01)$ is not convergent. Hence, we arrive at a
%% contradiction. Therefore, our assumption that the Weyl averages
%% $\langle n^{-1} \sum_{j=1}^{n} e^{2 \pi i k (v(T^j x)) }
%% \rangle_{n=1}^{\infty}$ are convergent for every $k \in \Z$ must be
%% wrong.

Let $\langle \nu_n \rangle_{n=1}^{\infty}$ be the sequence of averages
of Dirac measures constructed out of $\langle T^n
x\rangle_{n=0}^{\infty}$, and $\langle \nu'_n \rangle_{n=1}^{\infty}$,
those from $\langle 2^n v(x) \Mod 1 \rangle_{n=0}^{\infty}$. Assume
that $\langle n^{-1} \sum_{j=0}^{n-1} e^{2 \pi i k (v(T^j x)) }
\rangle_{n=1}^{\infty}$ converge for every $k \in \Z$. Using the same
steps in the proof of Theorem
\ref{thm:weylconvergenceimpliescylinderconvergence}, we get that
$\nu_n' \wto \mu'$ where $\mu'$ is the unique measure on $\T$ having
Fourier coefficients equal to the limits of the Weyl averages. Since
$v(x) \in \T \setminus \mathbb{D}$, Theorem
\ref{thm:weylconvergenceimpliescylinderconvergence} implies that
$\nu(C_{01})$ is convergent. Using Lemma
\ref{lem:deltameasuresandcoutning}, we infer that $\lim_{n \to
  \infty}P(x_0^{n-1},01)$ exists. But, we know from conditions
\ref{item:nicondition} and \ref{item:njcondition} that
$P(x_0^{n-1},01)$ is not convergent. Hence, we arrive at a
contradiction. Therefore, for some $k \in \Z$, the Weyl averages
$\langle n^{-1} \sum_{j=0}^{n-1} e^{2 \pi i k (v(T^j x)) }
\rangle_{n=1}^{\infty}$ diverge. The above construction is easily adapted to show that for any
rational number $p/q \in (0,1)$, there exists $x \in \Sigma^\infty$ with
$\dim_{FS}(x)=\Dim_{FS}(x)=p/q$ such that some Weyl average of $x$
diverges.
\begin{theorem}
\label{thm:nonconvergentexamplerational}
For any rational number $p/q \in (0,1)$, there exists $x \in \Sigma^\infty$ with $\dim_{FS}(x)=\Dim_{FS}(x)=p/q$ such that for some $k \in \Z$, the sequence $\langle \sum_{j=0}^{n-1}e^{2 \pi i k (v(T^j x))}/n \rangle_{n=1}^{\infty}$ is not convergent.
\end{theorem}
\begin{proof}
	 If $2p < q$ (equivalently $p/q < 1/2$) then the patterns $(0
         \star)^p 0^{q-2p}$ and $0^{q-p} \star^{p}$ can be alternated
         in the construction of $x$ in Theorem
         \ref{lem:nonconvergentexample} and the sliding count
         probability of the string $01$ can be made to oscillate
         between $p/2q$ and $(p+1)/4q$. These probabilities are equal
         if and only if $p=1$ which is avoided by assuming that $p$
         and $q$ are even (if they are not both even, then we perform
         the construction with the required dimension being
         $2p/2q$ and setting $p=2p$ and $q=2q$ in the design of the patterns). And if $2p > q$ (equivalently $p/q < 1/2$) then the
         patterns must be carefully selected. Let us assume without
         loss of generality that both $p$ and $q$ are even which
         implies that $(q-p) \geq 2$. Then, the patterns $(0
         \star)^{q-p} \star^{2p-q}$ and $0^{q-p} \star^{p}$ can be
         alternated in the construction of $x$ in Theorem
         \ref{lem:nonconvergentexample} and the sliding count
         probability of the string $01$ can be made to oscillate
         between $1/4$ and $(p+1)/4q$. These probabilities are equal
         if and only if $q-p=1$ which is avoided by assuming that $p$
         and $q$ are even as indicated above (if they are not both
         even, then we perform the construction with the required
         dimension being $2p/2q$ and setting $p=2p$ and $q=2q$ in the design of the patterns).
\end{proof}

\section{Weyl's criterion for finite-state dimension}
We saw in Lemma \ref{lem:nonconvergentexample} that Weyl averages may
diverge for $x$ having finite-state dimension less than $1$, even if
$x$ is regular. Hence, it is necessary for us to deal with divergent
Weyl averages and obtain their relationship with the finite-state
dimension of $x$. We know from Theorem \ref{thm:weylcriterioncantor}
that Weyl's criterion for normality (Theorem
\ref{thm:weyls_criterion_on_cantor_space}) is equivalently expressed
in terms of weak convergence of a sequence of measures over
$\Sigma^\infty$. In section \ref{subsec:weylcriterionforfsd}, we
generalize the weak convergence formulation to handle arbitrary finite
state dimension. Applying this, in section
\ref{subsec:weylaveragesconvergentcase}, we generalize the exponential
sum formulation.

%We demonstrate the utility of the exponential sum formulation by
%giving a new, Fourier-analytic, proof of Schnorr and Stimm's
%Theorem \cite{SchnorrStimm72}.

\subsection{Weak convergence and finite-state dimension}
\label{subsec:weylcriterionforfsd}

We know from Theorem \ref{thm:weylcriterioncantor} that $x \in
\Sigma^\infty$ is normal (equivalently, $\dim_{FS}(x)=1$) if and only
if $\nu_n \to \mu$, where $\mu$ is the uniform distribution over
$\Sigma^\infty$. In this subsection we give a generalization of this
formulation of Weyl's criterion which applies for $x$ having any
finite-state dimension. Lemma \ref{lem:nonconvergentexample} and
Theorem \ref{thm:cylinderconvergenceimpliesweylconvergence} together
imply that $\nu_n$'s need not be weakly convergent even if $x$ is
guaranteed to be regular. However, studying the subsequence limits of
$\langle \nu_n \rangle_{n=1}^{\infty}$ gives us the following
generalization of Weyl's criterion for arbitrary $x \in
\Sigma^\infty$.

\begin{theorem}
\label{thm:weylcriterionforfsd}
Let $x \in \Sigma^\infty$. Let $\langle \nu_n \rangle_{n=1}^{\infty}$
be the sequence of averages of Dirac measures on $\Sigma^\infty$
constructed out of the sequence $\langle T^n x
\rangle_{n=0}^{\infty}$. Let $\mathcal{W}_x$ be the collection of all
subsequence weak limits of $\langle \nu_n
\rangle_{n=1}^{\infty}$. i.e, 
$
\mathcal{W}_x = \{\mu \mid \exists \langle n_m \rangle_{m=0}^{\infty} \text{ such that } \nu_{n_m} \Rightarrow \mu\}
$. Then, $\dim_{FS}(x) = \inf_{\mu \in \mathcal{W}_x} H^{-}(\mu)$ and $\Dim_{FS}(x) = \sup_{\mu \in \mathcal{W}_x} H^{+}(\mu)$.	
\end{theorem}

We require the following technical lemmas for proving Theorem
\ref{thm:weylcriterionforfsd}. 
\begin{lemma}
\label{lem:subadditivityofentropy}	
For any $l$ and $m$,
$
(l+m)H_{l+m}(x_0^{n-1}) \leq lH_l (x_0^{n-1}) + mH_m (x_0^{n-1}) + o(n)/n
$ where the speed of convergence of the error term only depends on $l$
and $m$. 
\end{lemma}
The above lemma can be proved by using the techniques in the proof of
Lemma 1 from \cite{LZ78}. 
\begin{proof}[Proof of Lemma \ref{lem:subadditivityofentropy}]
\begin{align*}
	(l+m)H_{l+m}(x_0^{n-1}) &= -\sum_{w \in \Sigma^{l+m}} P(x_0^{n-1},w) \log(P(x_0^{n-1},w))\\
	&= -\sum_{u \in \Sigma^l} \sum_{v \in \Sigma^m} P(x_0^{n-1},uv) \left( \log\left( \frac{P(x_0^{n-1},uv)}{P(x_0^{n-1},u)} \right) + \log(P(x_0^{n-1},u)) \right)\\
	&= -\sum_{u \in \Sigma^l} P(x_0^{n-m-1},u) \log (P(x_0^{n-1},u)) - \sum_{u \in \Sigma^l} \sum_{v \in \Sigma^m} P(x_0^{n-1},uv)\log\left( \frac{P(x_0^{n-1},uv)}{P(x_0^{n-1},u)} \right)
\end{align*}
Let us analyze the first term on the RHS. 
\begin{align*}
  -\sum_{u \in \Sigma^l} P(x_0^{n-m-1}) \log (P(x_0^{n-1},u)) &=
  -\sum_{u \in \Sigma^l} P(x_0^{n-m-1},u) \log (P(x_0^{n-m-1},u)) \\
  &\quad\quad
  -\sum_{u \in \Sigma^l} P(x_0^{n-m-1},u) \log \frac{P(x_0^{n-1},u)}{P(x_0^{n-m-1},u)} \\
&= lH_l (x_0^{n-m-1}) + o(n)/n
\end{align*}
This is because when $n \to \infty$, $P(x_0^{n-m-1},u) \to P(x_0^{n-1},u)$, where the speed of convergence depends only on $m$. Hence, the error term depends only on $m$ and $\lvert \Sigma \rvert^l$.  Due to the same reason, using the continuity of entropy we get that,
\begin{align*}
lH_l (x_0^{n-m-1}) &= lH_l (x_0^{n-1}) + o(n)/n.
\end{align*}
Again, the error term in the above equation depends only on $m$ and $\lvert \Sigma^l \rvert$. Now, let us analyze the second term. Using the concavity of the $\log$ function,
\begin{align*}
 \sum_{u \in \Sigma^l} \sum_{v \in \Sigma^m} P(x_0^{n-1},uv)\log\frac{P(x_0^{n-1},u)}{P(x_0^{n-1},uv)}  &= \sum_{v \in \Sigma^m} P(x_0^{n-1},v) \sum_{u \in \Sigma^l} \frac{P(x_0^{n-1},uv)}{P(x_0^{n-1},v)}\log\left( \frac{P(x_0^{n-1},u)}{P(x_0^{n-1},uv)} \right)\\
 &\leq \sum_{v \in \Sigma^m} P(x_0^{n-1},v) \log\left( \sum_{u \in \Sigma^l}\frac{P(x_0^{n-1},u)}{P(x_0^{n-1},v)} \bigg/ \sum_{u \in \Sigma^l} \frac{P(x_0^{n-1},uv)}{P(x_0^{n-1},v)}\right)\\
 &= \sum_{v \in \Sigma^m} P(x_0^{n-1},v) \log\left( \frac{1}{P(x_0^{n-1},v)} \bigg/ \frac{P(x_{l}^{n-1},v)}{P(x_0^{n-1},v)}\right) \\
& =\sum_{v \in \Sigma^m} P(x_0^{n-1},v) \log
 \frac{1}{P(x_0^{n-1},v)}+\\
 &\quad\quad\quad
 \sum_{v \in \Sigma^m} P(x_0^{n-1},v) \log\frac{P(x_0^{n-1},v)}{P(x_{l}^{n-1},v)} \\
&= mH_m (x_0^{n-1}) + o(n)/n
\end{align*}
Where the last equality follows since $P(x_l^{n-1},v) \to P(x_0^{n-1},v)$ as $n \to \infty$ and the speed of convergence depends only on $l$. Hence, the error term depends only on $l$ and $\lvert \Sigma \rvert^m$. The required inequality follows from the bounds on both the terms on the RHS. We remark that every $o(n)/n$ term in the above bounds depends only on $l$ and $m$.
\end{proof}

The following is an immediate corollary of the above
inequality. 
\begin{lemma}
\label{lem:subadditivtycorollary1}
For any $l$ and $m$, $H_{ml}(x_0^{n-1}) \leq H_{l}(x_0^{n-1}) +
o(n)/n$ and $H_{l}(x_0^{n-1}) \geq H_{ml}(x_0^{n-1}) - o(n)/n$ where
the speed of convergence of the error term only depends on $l$ and
$m$. 
\end{lemma}

Now, we prove Theorem \ref{thm:weylcriterionforfsd} by proving the
equalities in the conclusion separately. 
\begin{lemma}
\label{lem:weylcriterionforfsd1}
	$\dim_{FS}(x) = \inf_{\mu \in \mathcal{W}_x} H^-(\mu)$
\end{lemma}
\begin{proof}
	We will first show that $\dim_{FS}(x)\leq \inf_{\mu \in \mathcal{W}_x} H^- (\mu)$. Let $s = \inf_{\mu \in \mathcal{W}_x} H^- (\mu)$. Let $\epsilon >0$ and $\mu'$ be any measure such that $H^-(\mu') < s+ \epsilon$. That is,
	\begin{align}
	\label{eqn:liminfequation1}
	\liminf_{l \to \infty} -\frac{1}{l} \sum_{w \in \Sigma^l} \mu'(C_w)\log(\mu'(C_w)) < s+\epsilon
	\end{align}
	Let $l'$ be any number such that
	\begin{align}
	\label{eqn:firstlemmaeq2}
	-\frac{1}{l'} \sum_{w \in \Sigma^{l'}} \mu'(C_w)\log(\mu'(C_w)) < s+\epsilon
	\end{align}
	which exists due to \ref{eqn:liminfequation1} (in fact there exist infinitely many $l'$ satisfying \ref{eqn:firstlemmaeq2}). Let $\langle \nu_{n_m} \rangle_{m=0}^{\infty}$ be any subsequence of $\<\nu_n\>_{n=0}^{\infty}$ such that $\nu_{n_m} \Rightarrow \mu'$. Hence, for any $w \in \Sigma^{l'}$, $\nu_{n_m}(C_w) \to \mu'(C_w)$ as $m \to \infty$. From, Lemma \ref{lem:deltameasuresandcoutning}, we get $P(x_0^{n_m+l'-2},w) \to \mu'(C_w)$ as $m \to \infty$. Hence, from \ref{eqn:firstlemmaeq2} and continuity of the entropy function we get that, there exist infinitely many $n_m$ such that
	\begin{align*}
	-\frac{1}{l'} \sum_{w \in \Sigma^{l'}} P(x_0^{n_m+l'-2},w)\log(P(x_0^{n_m+l'-2},w)) < s+\epsilon	
	\end{align*}
	This implies,
	\begin{align}
	\label{eqn:dimfsproofliminfstep}
	\liminf_{n \to \infty}-\frac{1}{l'} \sum_{w \in \Sigma^{l'}} P(x_0^{n-1},w)\log(P(x_0^{n-1},w)) < s+\epsilon	
	\end{align}
	The required inequality directly follows from above. Now we show that conversely, $\inf_{\mu \in \mathcal{W}_x} H^- (\mu) \leq \inf_l \liminf_{n \to \infty} H_l (x_0^{n-1})$. Let $s=\inf_l \liminf_{n \to \infty} H_l (x_0^{n-1})$. Let $\epsilon>0$ and $l'$ be any number such that $\liminf_{n \to \infty}H_{l'} (x_0^{n-1})<s+\epsilon$. Hence,
	\begin{align*}
	\liminf_{n \to \infty}-\frac{1}{l'} \sum_{w \in \Sigma^{l'}} P(x_0^{n-1},w)\log(P(x_0^{n-1},w)) < s+\epsilon	
	\end{align*}
	Hence, there are infinitely many $\langle n_m \rangle_{m=0}^{\infty}$ such that
	\begin{align}
	\label{eqn:firstlemmaeq4}
	-\frac{1}{l'} \sum_{w \in \Sigma^{l'}} P(x_0^{n_m-1},w)\log(P(x_0^{n_m-1},w)) < s+\epsilon	
	\end{align}
	For large enough $n_m$, due to the continuity of the entropy function we have,
	\begin{align}
	\label{eqn:firstlemmaeq3}
	-\frac{1}{l'} \sum_{w \in \Sigma^{l'}} P(x_0^{n_m+l'-2},w)\log(P(x_0^{n_m+l'-2},w)) < s+\epsilon	
	\end{align}
	Let $\langle \nu_{n_{m_i}} \rangle_{i=0}^{\infty}$ be any convergent subsequence of $\langle \nu_{n_m} \rangle_{m=0}^{\infty}$ such that $\nu_{n_{m_i}} \Rightarrow \mu'$ for some probability measure $\mu'$ as $i \to \infty$, which exists due to Prokhorov's theorem. We have from Lemma \ref{lem:deltameasuresandcoutning} that for any $w \in \Sigma^{l'}$,
	$
	P(x_0^{n_{m_i}+l'-2},w)= \nu_{n_{m_i}}(C_w).	
	$
	Hence we get,
	$
	\lim_{i \to \infty} P(x_0^{n_{m_i}+l'-2},w) = \mu'(C_w)	
	$
	for any $w \in \Sigma^{l'}$. Using continuity of entropy, \ref{eqn:firstlemmaeq3} implies,
	\begin{align}
	\label{eqn:firstlemmaeq5}
	-\frac{1}{l'} \sum_{w \in \Sigma^{l'}} \mu'(C_w)\log(\mu'(C_w)) < s+\epsilon
	\end{align}
	Now since \ref{eqn:firstlemmaeq4} is true, for any $k \in \N$ and large enough $n_m$ using Corollary \ref{lem:subadditivtycorollary1} we get that, 
	\begin{align}
	-\frac{1}{kl'} \sum_{w \in \Sigma^{kl'}} P(x_0^{n_m-1},w)\log(P(x_0^{n_m-1},w)) < s+\epsilon	
	\end{align}
	From \ref{eqn:firstlemmaeq4}, we get that the above is true for infinitely many $n_m$. The steps used in proving \ref{eqn:firstlemmaeq5} can be repeated (by choosing the same $\mu'$ as before) for $kl'$-length strings to obtain,
	\begin{align}
	\label{eqn:firstlemmaeq6}
	-\frac{1}{kl'} \sum_{w \in \Sigma^{kl'}} \mu'(C_w)\log(\mu'(C_w)) < s+\epsilon
	\end{align}
	Hence, \ref{eqn:firstlemmaeq6} is true for partitions having length in $\{kl'\}_{k=1}^{\infty}$. This implies,
	\begin{align*}
	\liminf_{l \to \infty} -\frac{1}{l} \sum_{w \in \Sigma^l} \mu'(C_w)\log(\mu'(C_w)) < s+\epsilon
	\end{align*}
	It follows that $\inf_{\mu \in \mathcal{W}_x} H^- (\mu) \leq s+\epsilon$. Letting $\epsilon \to 0$, we get the desired inequality.
\end{proof}

\begin{lemma}
\label{lem:weylcriterionforfsd2}
	$\Dim_{FS}(x) = \sup_{\mu \in \mathcal{W}_x} H^+(\mu)$
\end{lemma}
\begin{proof}
	We first show that $\Dim_{FS}(x)\leq \sup_{\mu \in \mathcal{W}_x} H^+(\mu)$. It is enough to show that if $\inf_l \limsup_{n \to \infty} H_l (x_0^{n-1}) > s$ then $\sup_{\mu \in \mathcal{W}_x} H^+(\mu) \geq s$. If $\inf_l \limsup_{n \to \infty} H_l (x_0^{n-1}) >s$ then for any length $l$,
	\begin{align}
	\label{eqn:weylcriterionforfsd2eqn1}
	\limsup_{n \to \infty}-\frac{1}{l} \sum_{w \in \Sigma^l} P(x_0^{n-1},w)\log(P(x_0^{n-1},w)) > s+\epsilon_1	
	\end{align}
	for some small $\epsilon_1>0$. Hence, for any $l'$ there exists infinitely many $n$ such that
	\begin{align*}
	-\frac{1}{l'} \sum_{w \in \Sigma^{l'}} P(x_0^{n-1},w)\log(P(x_0^{n-1},w)) > s+\epsilon_1.
	\end{align*}
	Let $\langle n^1_{m} \rangle_{m=0}^{\infty}$ be any increasing sequence such that
	\begin{align*}
	-\frac{1}{l'} \sum_{w \in \Sigma^{l'}} P(x_1^{n^1_{m}},w)\log(P(x_1^{n^1_{m}},w)) > s+\epsilon_1
	\end{align*}
	for all $m$. Since, \ref{eqn:weylcriterionforfsd2eqn1} is true for any length, using Corollary \ref{lem:subadditivtycorollary1}, we can choose a sequence $\langle n^2_{m} \rangle_{m=0}^{\infty}$ such that for each $m$,
	\begin{align*}
	-\frac{1}{2l'} \sum_{w \in \Sigma^{2l'}} P(x_1^{n^2_{m}},w)\log(P(x_1^{n^2_{m}},w)) > s+\epsilon_2
	\end{align*}
	and,
	\begin{align*}
	-\frac{1}{l'} \sum_{w \in \Sigma^{l'}} P(x_1^{n^2_{m}},w)\log(P(x_1^{n^2_{m}},w)) > s+\epsilon_2
	\end{align*}
	for some $0<\epsilon_2<\epsilon_1$. Similarly, for any $k>0$ we can choose $\langle n^k_{m} \rangle_{m=0}^{\infty}$ such that
	\begin{align*}
	-\frac{1}{2^il'} \sum_{w \in \Sigma^{2^il'}} P(x_1^{n^k_{m}},w)\log(P(x_1^{n^k_{m}},w)) > s+\epsilon_2
	\end{align*}
	for any $i \in \{0,1,2,3,\dots k-1\}$. Let $\langle a_{m}\rangle_{m=1}^{\infty}$ be any increasing sequence chosen such that $a_k$ is a member of $\langle n^k_{m} \rangle_{m=0}^{\infty}$. Now, consider $\langle \nu_{a_m} \rangle_{m=1}^{\infty}$ and let $\mu'$ be the weak limit point of any subsequence of $\langle \nu_{a_m} \rangle_{m=1}^{\infty}$, which exists due to Prokhorov's theorem. Now, for any $k \in \N$ and $m \geq k+1$,
	\begin{align*}
	-\frac{1}{2^kl'} \sum_{w \in \Sigma^{2^kl'}} P(x_1^{a_m},w)\log(P(x_1^{a_m},w)) > s+\epsilon_2.
	\end{align*}
	For any fixed $k$, using Corollary \ref{lem:subadditivtycorollary1} and continuity of the entropy function, for large enough $a_m$ along this sequence we have that,
	\begin{align*}
	-\frac{1}{2^kl'} \sum_{w \in \Sigma^{2^kl'}} P(x_1^{a_m+2^kl'-2},w)\log(P(x_1^{a_m+2^kl'-2},w)) > s+\epsilon_2.
	\end{align*}
	From Lemma \ref{lem:deltameasuresandcoutning} we get that for any $w \in \Sigma^{2^kl'}$, $P(x_0^{a_m+2^kl'-2},w) = \nu_{a_m}(C_w)$. Since, $\nu_{a_m}(C_w)$ converges along a subsequence to $\mu'(C_w)$, using continuity of the entropy function, we get, 
	\begin{align*}
	-\frac{1}{2^kl'} \sum_{w \in \Sigma^{2^kl'}} \mu'(C_w)\log(\mu'(C_w)) > s+\epsilon_2.
	\end{align*}
	Since the above is true for any $k$, we get that,
	\begin{align*}
	\limsup_{l \to \infty}-\frac{1}{l} \sum_{w \in \Sigma^l} \mu'(C_w)\log(\mu'(C_w)) > s+\epsilon_2.
	\end{align*}
	This implies that $\sup_{\mu \in \mathcal{W}_x} H^+(\mu) > s +\epsilon_2$ for $\epsilon_2>0$, which in turn implies that $\sup_{\mu \in \mathcal{W}_x} H^+(\mu) \geq s$. This completes the proof of the first part. Conversely, let us show that $s=\inf_l \limsup_{n \to \infty} H_l (x_0^{n-1}) \geq \sup_{\mu \in \mathcal{W}_x} H^+(\mu)$. For any $\epsilon>0$, there exists an $l'$ such that
	\begin{align*}
	\limsup_{n \to \infty}-\frac{1}{l'} \sum_{w \in \Sigma^{l'}} P(x_0^{n-1},w)\log(P(x_0^{n-1},w)) < s+\epsilon	
	\end{align*}
	Hence, for a small enough $\epsilon'>0$, there exists $ N(\epsilon,l')\in \N$ such that for all $n \geq N(\epsilon,l')$,
	\begin{align*}
	-\frac{1}{l'} \sum_{w \in \Sigma^{l'}} P(x_0^{n-1},w)\log(P(x_0^{n-1},w)) < s+\epsilon-\epsilon'	
	\end{align*}
	From the above, by using Corollary \ref{lem:subadditivtycorollary1}, it can be shown that for all $k>0$, there exists $ N(\epsilon, kl') \in \N$ such that for all $ n \geq N(\epsilon, kl')$,
	\begin{align*}
	-\frac{1}{kl'} \sum_{w \in \Sigma^{kl'}} P(x_0^{n-1},w)\log(P(x_0^{n-1},w)) < s+\epsilon	
	\end{align*}
	For $1<r<l'-1$, let $j$ be any number such $j=kl'+r$. Now, from Corollary \ref{lem:subadditivtycorollary1}, we have
	\begin{align*}
	H_j(x_0^{n-1})=H_{kl'+r}(x_0^{n-1}) \leq \frac{kl'}{kl'+r}H_{kl'}(x_0^{n-1}) + \frac{r}{kl'+r} H_{r}(x_0^{n-1}) + o(n)/n	
	\end{align*}
 	 Using the fact that $H_r$ is at most $1$, for large enough $j$ (equivalently for  large enough $k$) we have,
 	\begin{align*}
	H_j(x_0^{n-1})=H_{kl'+r}(x_0^{n-1}) \leq H_{kl'}(x_0^{n-1}) + \epsilon + o(n)/n	
	\end{align*}
 	Since this conclusion can be obtained for any $r$, for large enough $j$,  there exists $ N(\epsilon,j)$ such that for all $n \geq N(\epsilon,j)$,
 	\begin{align*}
 		-\frac{1}{j} \sum_{w \in \Sigma^{j}} P(x_0^{n-1},w)\log(P(x_0^{n-1},w)) < s+2\epsilon
 	\end{align*}

	Using the continuity of entropy, by considering large enough $N(\epsilon,j)$ we can ensure that for all $ n \geq N(\epsilon, j)$,
	\begin{align}
	\label{eqn:secondlemmaeq1}
	-\frac{1}{j} \sum_{w \in \Sigma^{j}} P(x_0^{n+j-2},w)\log(P(x_0^{n+j-2},w)) < s+2\epsilon	
	\end{align}
	Consider any $\mu' \in \mathcal{W}_x$. Let $\langle \nu_{n_m} \rangle_{m=0}^{\infty}$ be any subsequence of $\<\nu_n\>_{n=0}^{\infty}$ such that $\nu_{n_m} \wto \mu'$ as $m \to \infty$. Now, for any $j$ and $w \in \Sigma^{j}$, from Lemma \ref{lem:deltameasuresandcoutning} we have $P(x_0^{n_m+j-2},w)=\nu_{n_m}(C_w) \to \mu'(C_w)$ as $m \to \infty$. From \ref{eqn:secondlemmaeq1} using continuity of entropy, we get that,
	 $
	-j^{-1} \sum_{w \in \Sigma^j} \mu'(C_w)\log(\mu'(C_w)) < s+2\epsilon.
	$
	Since the above holds for all large enough $j$, we get
	\begin{align*}
	\limsup_{l \to \infty}-\frac{1}{l} \sum_{w \in \Sigma^l} \mu'(C_w)\log(\mu'(C_w)) < s+2\epsilon.
	\end{align*}
	Since the above holds for any $\mu' \in \mathcal{W}_x$ we get that,
	$
	\sup_{\mu \in \mathcal{W}_x} H^+(\mu) \leq s+2\epsilon.
	$
	By letting $\epsilon \to 0$ we obtain the desired inequality.
\end{proof}

Theorem \ref{thm:weylcriterionforfsd} now follows from Lemma \ref{lem:weylcriterionforfsd1} and Lemma \ref{lem:weylcriterionforfsd2}. Now, we prove an equivalent version of Theorem \ref{thm:weylcriterionforfsd} which we require in section \ref{subsec:weylaveragesconvergentcase}. From the definition of lower average entropy, Theorem \ref{thm:weylcriterionforfsd} shows that, $
 \dim_{FS}(x) = \inf_{\mu \in \mathcal{W}_x} \liminf_{l \to \infty} \mathbf{H}_{l}(\mu)/l$. We show that the limit inferior in this expression can be replaced by an infimum.
\begin{lemma}
\label{lem:liminfreplacedwithinf}
$
 \dim_{FS}(x) = \inf_{\mu \in \mathcal{W}_x} \inf_{l} \mathbf{H}_{l}(\mu)/l 
$
\end{lemma}
\begin{proof}
Observe that from Lemma \ref{lem:weylcriterionforfsd1}
\begin{align*}
\inf_{\mu \in \mathcal{W}_x} \inf_{l} \frac{\mathbf{H}_{l}(\mu)}{l} \leq \inf_{\mu \in \mathcal{W}_x} \liminf_{l \to \infty} \frac{\mathbf{H}_{l}(\mu)}{l}
\end{align*}
we get that $\inf_{\mu \in \mathcal{W}_x} \inf_{l} \mathbf{H}_{l}(\mu)/l \leq \dim_{FS}(x)$. Hence, it is enough to show that $ \dim_{FS}(x) \leq \inf_{\mu \in \mathcal{W}_x} \inf_{l} \mathbf{H}_{l}(\mu)/l$. This can be shown using the same steps in the proof of the first part of Lemma \ref{lem:weylcriterionforfsd1} since the existence of an $l'$ satisfying \ref{eqn:firstlemmaeq2} is true if we assume that $s=\inf_{\mu \in \mathcal{W}_x} \inf_{l} \mathbf{H}_{l}(\mu)/l$. The rest of the proof follows from \ref{eqn:firstlemmaeq2} using identical steps.
\end{proof}

\subsection{Weyl averages and finite-state dimension}
\label{subsec:weylaveragesconvergentcase}
%% In this section we relate the convergence of Weyl averages and
%% finite-state dimension using Theorem \ref{thm:weylcriterionforfsd}
%% and the results from Section
%% \ref{sec:weylscriterionandweakconvergence}. We obtain the main
%% result of this paper by relating the subsequence limits of Weyl
%% averages and the finite-state dimension of the given
%% sequence. Thereafter, we investigate sequences with convergent Weyl
%% averages as a special case of the main theorem. Using the
%% convergent case Weyl's criterion we obtain a relationship between
%% the regularity of the sequence and convergence of Weyl averages. We
%% also investigate Weyl's criterion for normality where the Weyl
%% averages are convergent to $0$ in light of the convergent case
%% Weyl's criterion.

We now obtain the main result of the paper by relating subsequence
limits of Weyl averages and finite-state dimension. In case the Weyl
averages converge, we show that the sequence is regular. In
particular, when the Weyl averages converge to 0, then the regular
sequence is normal. 

%This latter result provides a Fourier analytic
%proof of Schnorr and Stimm's theorem \cite{SchnorrStimm72}.

We know from Lemma \ref{lem:nonconvergentexample} that there exist
regular sequences with non-convergent Weyl averages. In the absence of
limits, we investigate the subsequence limits of Weyl averages in
order to obtain a relationship with the finite-state dimension. If for some $x \in \Sigma^\infty$, there exist a sequence of natural numbers $\langle n_m \rangle_{m \in \N}$ and constants
$\langle c_k \rangle_{k \in \Z}$ such that $ \lim_{m \to
  \infty}\frac{1}{n_m}\sum_{j=0}^{n_m-1} e^{2 \pi i k (v(T^j x)) } =
c_k$. Then, using Theorem
\ref{thm:weylconvergenceimpliescylinderconvergence}, we get that there
exists a measure $\mu$ on $\T$ such that $c_k=\int e^{2 \pi i k y}
d\mu$ and $\lim_{m \to \infty}\nu_{n_m}(C_w)=\hat{\mu}(C_w)$ for every
$w \neq 0^{\lvert w \rvert}$ and $w \neq 1^{\lvert w \rvert}$. %% We
%% cannot in general obtain convergence of $\nu_{n_m}(C_{0^l})$ and
%% $\nu_{n_m}(C_{1^l})$, since there exist sequences with convergent Weyl
%% averages having no definite limit distribution for the strings $0^l$
%% and $1^l$ for all $l$.
But, $\nu_{n_m}(C_{0^l})$ and $\nu_{n_m}(C_{1^l})$ need not converge.
Simple examples of such strings can be obtained by concatenating
increasingly large runs of $0$'s and $1$'s in an alternating stage
wise manner. However, the probabilities of the strings $0^l$ and $1^l$
have negligible effect on the finite-state dimension as $l$ gets
large.  Using Theorem \ref{thm:weylcriterionforfsd} we obtain the
following.

\begin{theorem}[Weyl's criterion for finite-state dimension]
\label{thm:nonconvergentcaseweylcriterion}
	Let $x \in \Sigma^\infty$. If for any $\langle n_m
        \rangle_{m=0}^{\infty}$ there exist constants $c_k$ for $k \in
        \Z$ such that $\lim_{m \to
          \infty}\frac{1}{n_m}\sum_{j=0}^{n_m-1} e^{2 \pi i k (v(T^j
          x)) } = c_k$, for every $k \in \Z$, then there exists a
        measure $\mu$ on $\T$ such that for every $k$, $c_k = \int
        e^{2 \pi i k y } d\mu$. Let $\widehat{\mathcal{W}}_x$ be
        the collection of the lifted measures $\hat{\mu}$ on
        $\Sigma^\infty$ for all $\mu$ on $\T$ that can be obtained as
        subsequence limits of Weyl averages. Then,
\begin{align*}
  \dim_{FS}(x) = \inf_{}
  \{H^{-}(\hat{\mu}) \mid \hat{\mu} \in \widehat{\mathcal{W}}_x\}
  \quad\text{and}\quad
  \Dim_{FS}(x) = \sup
  \{H^{+}(\hat{\mu}) \mid \hat{\mu} \in \widehat{\mathcal{W}}_x\}
\end{align*}        
\end{theorem}

Hence, the finite-state dimension and finite-state strong dimension are related to the lower and upper average entropies of the subsequence limits of the Weyl averages. 

We require the following definitions and technical lemmas for proving Theorem \ref{thm:nonconvergentcaseweylcriterion}. We define,

\begin{align*}
\widetilde{H}_l(x_0^{n-1})= 	\frac{-1}{l }\sum_{w \in \Sigma^l \setminus \{0^l,1^l\}} P(x_0^{n-1},w)\log(P(x_0^{n-1},w)).
\end{align*}

Using the above notion of entropy, we define,
\begin{align*}
\widetilde{\dim}_{FS}(x)=\liminf_{l \to \infty} \liminf_{n \to \infty} \widetilde{H}_l(x_0^{n-1}) \\
 \widetilde{\Dim}_{FS}(x)=\liminf_{l \to \infty} \limsup_{n \to \infty} \widetilde{H}_l(x_0^{n-1}).
\end{align*}

The following lemma easily follows from the definitions.
\begin{lemma}
\label{lem:newentropyandoldentropy}
$
\widetilde{H}_l(x_0^{n-1}) \leq H_l(x_0^{n-1}) \leq \widetilde{H}_l(x_0^{n-1})+2/l.
$
\end{lemma}

We use Lemma \ref{lem:newentropyandoldentropy} to prove the following lemma.

\begin{lemma}
For any $x \in \Sigma^\infty$, $\widetilde{\dim}_{FS}(x)=\dim_{FS}(x)$.
\end{lemma}
\begin{proof}
From \ref{eqn:limitdimension}, $\dim_{FS}(x)=\lim_{l \to \infty} \liminf_{n \to \infty} H_l(x_0^{n-1})=\liminf_{l \to \infty} \liminf_{n \to \infty} H_l(x_0^{n-1})$. From Lemma \ref{lem:newentropyandoldentropy}, we get that $\widetilde{\dim}_{FS}(x)\leq \dim_{FS}(x)$. Conversely using Lemma \ref{lem:newentropyandoldentropy} we get,
\begin{align*}
\dim_{FS}(x) &= \liminf_{l \to \infty} \liminf_{n \to \infty} H_l(x_0^{n-1})\\
&\leq \liminf_{l \to \infty} \liminf_{n \to \infty} \left( \widetilde{H}_l(x_0^{n-1})+2/l \right)\\
&= \liminf_{l \to \infty}  \left( \liminf_{n \to \infty}\widetilde{H}_l(x_0^{n-1})+2/l \right)\\
&= \liminf_{l \to \infty}   \liminf_{n \to \infty}\widetilde{H}_l(x_0^{n-1})\\
&= \widetilde{\dim}_{FS}(x).
\end{align*}
In the second last equality we used the fact if $a_l$ and $b_l$ are sequences such that $ \lim_{l \to \infty}b_l=0$, then $\liminf_{l \to \infty}(a_l+b_l)=\liminf_{l \to \infty} a_l$. 
\end{proof}

We prove the analogous lemma for finite-state strong dimension.

\begin{lemma}
For any $x \in \Sigma^\infty$, $\widetilde{\Dim}_{FS}(x)=\Dim_{FS}(x)$.
\end{lemma}
\begin{proof}
From \ref{eqn:limitdimension}, $\Dim_{FS}(x)=\lim_{l \to \infty} \limsup_{n \to \infty} H_l(x_0^{n-1})=\liminf_{l \to \infty} \limsup_{n \to \infty} H_l(x_0^{n-1})$. From Lemma \ref{lem:newentropyandoldentropy}, we get that $\widetilde{\Dim}_{FS}(x)\leq \Dim_{FS}(x)$. Conversely using Lemma \ref{lem:newentropyandoldentropy} we get,
\begin{align*}
\Dim_{FS}(x) &= \liminf_{l \to \infty} \limsup_{n \to \infty} H_l(x_0^{n-1})\\
&\leq \liminf_{l \to \infty} \limsup_{n \to \infty} \left( \widetilde{H}_l(x_0^{n-1})+2/l \right)\\
&= \liminf_{l \to \infty}  \left( \limsup_{n \to \infty}\widetilde{H}_l(x_0^{n-1})+2/l \right)\\
&= \liminf_{l \to \infty}   \limsup_{n \to \infty}\widetilde{H}_l(x_0^{n-1})\\
&= \widetilde{\Dim}_{FS}(x).
\end{align*}
In the second last equality we used the fact if $a_l$ and $b_l$ are sequences such that $ \lim_{l \to \infty}b_l=0$, then $\liminf_{l \to \infty}(a_l+b_l)=\liminf_{l \to \infty} a_l$. 
\end{proof}

Now, for any probability measure $\mu$ on $\Sigma^\infty$, we define
\begin{align*}
\widetilde{\mathbf{H}}_n (\mu)=-\sum_{w \in \Sigma^n \setminus \{0^n,1^n\}}\mu(C_w)\log(\mu(C_w)).
\end{align*}
And using the above definition, we define
\begin{align*}
\widetilde{H}^{+}(\mu)=\limsup_{n \to \infty} \frac{\widetilde{\mathbf{H}}_{n}(\mu)}{n}
\end{align*}
\begin{align*}
\widetilde{H}^{-}(\mu)=\liminf_{n \to \infty} \frac{\widetilde{\mathbf{H}}_{n}(\mu)}{n}.
\end{align*}

Analogous to Lemma \ref{lem:newentropyandoldentropy}, we now have the following lemma which easily follows from the definitions.

\begin{lemma}
\label{lem:newaverageentropyandoldaverageentropy}	
$\widetilde{\mathbf{H}}_n (\mu) \leq \mathbf{H}_n (\mu) \leq \widetilde{\mathbf{H}}_n (\mu) +2$.
\end{lemma}

Using Lemma \ref{lem:newaverageentropyandoldaverageentropy}, we prove the following lemmas.

\begin{lemma}
\label{lem:newlowerentropyandoldlowerentropy}
For any $\mu$ on $\Sigma^\infty$, $\widetilde{H}^{-}(\mu)=H^{-}(\mu)$.
\end{lemma}
\begin{proof}
From Lemma \ref{lem:newaverageentropyandoldaverageentropy}, it easily follows that $\widetilde{H}^{-}(\mu)\leq H^{-}(\mu)$. Conversely, using Lemma \ref{lem:newaverageentropyandoldaverageentropy} we get,
\begin{align*}
H^-(\mu) &=	\liminf_{n \to \infty} \frac{\mathbf{H}_{n}(\mu)}{n} \\
&\leq \liminf_{n \to \infty} \left(\frac{\widetilde{\mathbf{H}}_n(\mu)}{n} + \frac{2}{n}\right)\\
&= \liminf_{n \to \infty} \frac{\widetilde{\mathbf{H}}_{n}(\mu)}{n} \\
&= \widetilde{H}^{-}(\mu).
\end{align*}
In the second last equality we used the fact if $a_n$ and $b_n$ are sequences such that $ \lim_{n \to \infty}b_n=0$, then $\liminf_{n \to \infty}(a_n+b_n)=\liminf_{n \to \infty} a_n$. 
\end{proof}

\begin{lemma}
\label{lem:newupperentropyandoldupperentropy}
For any $\mu$ on $\Sigma^\infty$, $\widetilde{H}^{+}(\mu)=H^{+}(\mu)$.
\end{lemma}
\begin{proof}
From Lemma \ref{lem:newaverageentropyandoldaverageentropy}, it easily follows that $\widetilde{H}^{+}(\mu)\leq H^{+}(\mu)$. Conversely, using Lemma \ref{lem:newaverageentropyandoldaverageentropy} we get,
\begin{align*}
H^+(\mu) &=	\limsup_{n \to \infty} \frac{\mathbf{H}_{n}(\mu)}{n} \\
&\leq \limsup_{n \to \infty} \left(\frac{\widetilde{\mathbf{H}}_n(\mu)}{n} + \frac{2}{n}\right)\\
&= \limsup_{n \to \infty} \frac{\widetilde{\mathbf{H}}_{n}(\mu)}{n} \\
&= \widetilde{H}^{-}(\mu).
\end{align*}
In the second last equality we used the fact if $a_n$ and $b_n$ are sequences such that $ \lim_{n \to \infty}b_n=0$, then $\limsup_{n \to \infty}(a_n+b_n)=\limsup_{n \to \infty} a_n$. 
\end{proof}

Now, we prove one of the major technical lemmas used in the proof of Theorem \ref{thm:nonconvergentcaseweylcriterion}. In the following lemma, for any $x \in \Sigma^\infty$, let $\langle \nu_n \rangle_{n=1}^{\infty}$ be the sequence of averages of Dirac measures on $\Sigma^\infty$ constructed out of the sequence $\langle T^n x  \rangle_{n=0}^{\infty}$ and let $\langle \nu'_n \rangle_{n=1}^{\infty}$ be the sequence of averages of Dirac measures on $\T$ constructed out of the sequence $\langle 2^n v(x) \Mod 1  \rangle_{n=0}^{\infty}$.

\begin{lemma}
Let $x \in \Sigma^\infty$ be such that $v(x)$ is not a dyadic rational. Let $\mathcal{W}_x$ be the collection of all subsequence weak limits of $\langle \nu_n \rangle_{n=1}^{\infty}$ and let $\widehat{\mathcal{W}}_x$ be the set constructed in the statement of Theorem \ref{thm:nonconvergentcaseweylcriterion}. Then, 
\begin{align*}
\inf_{\mu \in \mathcal{W}_x} H^{-}(\mu)=\inf_{\hat{\mu} \in
  \widehat{\mathcal{W}}_x} H^{-}(\hat{\mu}) \quad\text{and}\quad
\sup_{\mu \in \mathcal{W}_x} H^{+}(\mu)
=\sup_{\hat{\mu} \in \widehat{\mathcal{W}}_x} H^{+}(\hat{\mu}).
\end{align*}
\end{lemma}
\begin{proof}
In order to show these equalities, it is enough to show that $\{H^-(\mu) \mid \mu \in \mathcal{W}_x\}=\{H^-(\hat{\mu}) \mid \hat{\mu} \in \widehat{\mathcal{W}}_x\}$ and $\{H^+(\mu) \mid \mu \in \mathcal{W}_x\}=\{H^+(\hat{\mu}) \mid \hat{\mu} \in \widehat{\mathcal{W}}_x\}$. In order to show these equalities, we show the following:
\begin{enumerate}
\itemsep=0em
	\item\label{item:newwandoldwlemmacondition1} For every $\mu \in \mathcal{W}_x$, there exists a $\hat{\rho} \in \widehat{\mathcal{W}}_x$ such that $\hat{\rho}(C_w)=\mu(C_w)$ for every $w \neq 1^{\lvert w \rvert}$ and $w \neq 0^{\lvert w \rvert}$.
	\item\label{item:newwandoldwlemmacondition2} For every $\hat{\mu} \in \widehat{\mathcal{W}}_x$, there exists a $\rho \in \mathcal{W}_x$ such that $\hat{\mu}(C_w)=\rho(C_w)$ for every $w \neq 1^{\lvert w \rvert}$ and $w \neq 0^{\lvert w \rvert}$.
\end{enumerate}
If we show the above conditions to be true, then the required conclusion follows from Lemma \ref{lem:newlowerentropyandoldlowerentropy} and Lemma \ref{lem:newupperentropyandoldupperentropy}. We first show condition \ref{item:newwandoldwlemmacondition1}. 

Let $\mu \in \mathcal{W}_x$ and let $\langle n_m \rangle_{m \in \N}$ be such that $\nu_{n_m} \wto \mu$. Hence, using Lemma \ref{lem:cylindersetconvergenceandweakconvergence}, $\lim_{m \to \infty} \nu_{n_m}(C_w) = \mu(C_w)$ for every $w \in \Sigma^*$. Since $v(x)$ is not a dyadic rational, we know from the proof of Theorem \ref{thm:weylconvergenceimpliescylinderconvergence} that for every $w \in \Sigma^*$ and $m \geq 1$, $\nu_{n_m}(C_w) = \nu'_{n_m}(I_w)$. Hence, we get $\lim_{m \to \infty} \nu'_{n_m}(I_w) = \mu(C_w)$ for every $w \in \Sigma^*$. Now, let $\rho$ be a subsequence weak limit of $\langle \nu'_{n_m} \rangle_{m \in \N}$ which exists due to Prokhorov's theorem. Consider any $w$ such that $w \neq 1^{\lvert w \rvert}$ and $w \neq 0^{\lvert w \rvert}$. Now, using Lemma \ref{lem:nonzerodyadicpointsarenotlimitpoints}, we get that $\rho(\{v(w0^\infty)\})=\rho(\{v(w1^\infty)\})=0$. Since $v(w0^\infty)$ and $v(w1^\infty)$ are the end points of $I_w$, using Lemma \ref{lem:weakconvergenceimpliesdyadicconvergence} we get that $\lim_{m \to \infty}\nu'_{n_m}(I_w)=\rho(I_w)$ for every $w \neq 1^{\lvert w \rvert}$ and $w \neq 0^{\lvert w \rvert}$. Since $\lim_{m \to \infty} \nu'_{n_m}(I_w) = \mu(C_w)$, this implies that $\rho(I_w)=\mu(C_w)$ for every $w \neq 1^{\lvert w \rvert}$ and $w \neq 0^{\lvert w \rvert}$. Since $\rho$ is a subsequence weak limit of $\langle \nu'_{n_m} \rangle_{m \in \N}$, due to Theorem \ref{thm:weakconvergenceequivalencetorus}, we get that $\hat{\rho} \in \widehat{\mathcal{W}}_x$. Now, condition \ref{item:newwandoldwlemmacondition2} follows since $\hat{\rho}$ is a measure in $\widehat{\mathcal{W}}_x$ with the required property.

Now, we show condition \ref{item:newwandoldwlemmacondition2}. Let $\hat{\mu} \in \widehat{\mathcal{W}}_x$. From the definition of $\widehat{\mathcal{W}}_x$ we can infer that there exists $\langle n_m \rangle_{m \in \N}$ such that for every $k \in \Z$, $\lim_{m \to \infty}\frac{1}{n_m}\sum_{j=0}^{n_m-1} e^{2 \pi i k (v(T^j x)) }  = \int e^{2\pi i k y} d\mu$. From Theorem \ref{thm:weylconvergenceimpliescylinderconvergence}, we get that $\lim_{m \to \infty}\nu_{n_m}(C_w)=\hat{\mu}(C_w)$ for every $w \neq 1^{\lvert w \rvert}$ and $w \neq 0^{\lvert w \rvert}$. Let $\rho$ be any subsequence weak limit of $\langle \nu_{n_m}\rangle_{m \in \N}$ which exists due to Prokhorov's Theorem. Then, using Lemma \ref{lem:cylindersetconvergenceandweakconvergence} we get that $\rho(C_w) = \hat{\mu}(C_w)$ for every $w \neq 1^{\lvert w \rvert}$ and $w \neq 0^{\lvert w \rvert}$. Since $\rho \in \mathcal{W}_x$, condition \ref{item:newwandoldwlemmacondition2} is true.
\end{proof} 

Now we prove Theorem \ref{thm:nonconvergentcaseweylcriterion}
\begin{proof}[Proof of Theorem \ref{thm:nonconvergentcaseweylcriterion}]

If $v(x)$ is a dyadic rational in $\T$, then it can be easily verified that the Weyl averages are convergent to $1$. The unique measure having all Fourier coefficients equal to $1$ over $\T$ is $\delta_0$. Since $\hat{\delta}_0=\delta_{0^\infty}$, it can be easily verified that $\dim_{FS}(x)=H^{-}(\delta_{0^\infty})=H^+(\delta_{0^\infty})=\Dim_{FS}(x)=0$. Hence, we consider the case when $v(x)$ is not a dyadic rational. We first define analogues of finite-state dimension by avoiding the strings $0^l$ and $1^l$ for all $l$ in calculating the sliding entropies. We define $\widetilde{H}_l(x_0^{n-1}) $ to be the normalized sliding entropy over $x_0^{n-1}$ as in the definition of  $H_l(x_0^{n-1})$, except that the summation is taken over $\Sigma^l \setminus \{0^l,1^l\}$ instead of $\Sigma^l$. Using this notion, we define $\widetilde{\dim}_{FS}(x)=\liminf_{l \to \infty} \liminf_{n \to \infty} \widetilde{H}_l(x_0^{n-1})$ and $\widetilde{\Dim}_{FS}(x)=\liminf_{l \to \infty} \limsup_{n \to \infty} \widetilde{H}_l(x_0^{n-1})$. Since, $\widetilde{H}_l(x_0^{n-1}) \leq H_l(x_0^{n-1}) \leq \widetilde{H}_l(x_0^{n-1})+2/l$, it can be shown using routine arguments that $\dim_{FS}(x)=\widetilde{\dim}_{FS}(x)$ and $\Dim_{FS}(x)=\widetilde{\Dim}_{FS}(x)$. Similarly we define $\widetilde{H}^+$ and $\widetilde{H}^-$ by reducing the range of the sum in the definition of $\mathbf{H}_l$ to $\Sigma^l \setminus \{0^l,1^l\}$ instead of $\Sigma^l$. Using a similar argument as in the case of sliding entropy, it can be shown that $\widetilde{H}^+$ and $\widetilde{H}^-$ are the same as $H^+$ and $H^-$ for any measure on $\Sigma^\infty$. Let $\langle \nu_n \rangle_{n=1}^{\infty}$ be the sequence of averages of Dirac measures on $\Sigma^\infty$ constructed out of the sequence $\langle T^n x  \rangle_{n=0}^{\infty}$. Let $\mathcal{W}_x$ be the set of all weak limits of $\nu_n$ as constructed in Theorem \ref{thm:weylcriterionforfsd}. Since $v(x)$ is not a dyadic rational, using Prokhorov's theorem for weak convergence of $\T$ and weak convergence over $\Sigma^\infty$, it can be shown that, $\inf_{\mu \in \mathcal{W}_x} H^{-}(\mu)=\inf_{\hat{\mu} \in \widehat{\mathcal{W}}_x} H^{-}(\hat{\mu})$ and $\sup_{\mu \in \mathcal{W}_x} H^{+}(\mu)
=\sup_{\hat{\mu} \in \widehat{\mathcal{W}}_x} H^{+}(\hat{\mu})$. The claim now follows from Theorem \ref{thm:weylcriterionforfsd}.
\end{proof}

Using Theorem \ref{thm:nonconvergentcaseweylcriterion}, we get the following theorem in the case when the Weyl averages are convergent.

\begin{theorem}[Weyl's criterion for convergent Weyl averages]
\label{thm:weylaveragesconvergentcase}
Let $x \in \Sigma^\infty$. If there exist $c_k \in \mathbb{C}$ for $k \in
\Z$ such that $\frac{1}{n}\sum_{j=0}^{n-1} e^{2 \pi i k (v(T^j x))} \to
c_k$ as $n \to \infty$, then, there exists a unique measure $\mu$ on
$\T$ such that for every $k$, $c_k = \int e^{2 \pi i k y }
d\mu$. Furthermore, $\dim_{FS}(x) =\Dim_{FS}(x)= H^{-}(\hat{\mu}) =
H^{+}(\hat{\mu})$.
\end{theorem}
\begin{proof}
Since the Weyl averages are convergent, any subsequence shall also
converge to the same limit. This implies that $\widehat{\mathcal{W}}_x$ is a singleton set and hence from Theorem
\ref{thm:nonconvergentcaseweylcriterion} it follows that $\dim_{FS}(x)
= H^{-}(\hat{\mu})$ and $\Dim_{FS}(x) = H^{+}(\hat{\mu})$. Let
$\widetilde{H}_l(x_0^{n-1})$ be defined as in the proof of Theorem
\ref{thm:nonconvergentcaseweylcriterion}. From the remarks at the
start of this section and Lemma \ref{lem:deltameasuresandcoutning}, we
have $\lim_{n \to \infty}\nu_n(C_w)=\lim_{n \to \infty}
P(x_0^{n-1},w)=\hat{\mu}(C_w)$ for every $w \neq 0^{\lvert w \rvert}$
and $w \neq 1^{\lvert w \rvert}$. Hence, $\limsup_{n \to \infty}
\widetilde{H}_l(x_0^{n-1})=\liminf_{n \to \infty}
\widetilde{H}_l(x_0^{n-1})$. Therefore, $\widetilde{\dim}_{FS}(x)=\widetilde{\Dim}_{FS}(x)$. Now, $\dim_{FS}(x)=\Dim_{FS}(x)$ follows
because $\dim_{FS}(x)=\widetilde{\dim}_{FS}(x)$ and $\Dim_{FS}(x)=\widetilde{\Dim}_{FS}(x)$ as given in the proof of Theorem
\ref{thm:nonconvergentcaseweylcriterion}.
\end{proof}

As a special case, we derive Weyl's criterion for normality, i.e,
for sequences $x$ such that $\dim_{FS}(x)=\Dim_{FS}(x)=1$ as a special
case of Theorem \ref{thm:weylcriterionforfsd} and Theorem \ref{thm:weylaveragesconvergentcase}. 

%Since $x$ is normal if
%and only if it satisfies Weyl's criterion, the following is a Fourier
%analytic proof of Schnorr and Stimm's result.
\begin{theorem}
\label{thm:normalweylcriterionasaspecialcase}
Let $x \in \Sigma^\infty$. Then $\lim_{n \to
  \infty}\frac{1}{n}\sum_{j=0}^{n-1} e^{2 \pi i k (v(T^j x)) }  = 0$
for every $k \in \Z$ if and only $\dim_{FS}(x)=\Dim_{FS}(x)=1$. 
\end{theorem}
\begin{proof}
The forward direction follows from Theorem
\ref{thm:weylaveragesconvergentcase} since the uniform distribution is
the unique measure on $\T$ having all Fourier coefficients equal to
0. Conversely, assume that $\dim_{FS}(x)=\Dim_{FS}(x)=1$. From Lemma
\ref{lem:liminfreplacedwithinf}, for any $\nu \in \mathcal{W}_x$, we
have $\inf_{l} \mathbf{H}_{l}(\nu)/l =1$. Since
$\mathbf{H}_{l}(\nu)\le l$, this implies that for every $l$,
$\mathbf{H}_{l}(\nu)=\sum_{w \in
  \Sigma^l}\nu(C_w)\log(\nu(C_w))=l$. From this we can infer that
$\nu$ is the uniform distribution on $\Sigma^\infty$. Hence, the
uniform distribution is the unique weak limit in the set
$\mathcal{W}_x$ defined in the statement of Theorem
\ref{thm:weylcriterionforfsd}. The claim now follows from the
definition of weak convergence since $e^{2 \pi i k v(y)}$ is a
continuous function over $\Sigma^\infty$.
\end{proof}

The conclusion of Theorem \ref{thm:weylaveragesconvergentcase} says
that $\dim_{FS}(x)=\Dim_{FS}(x)$. i.e, $x$ is a regular
sequence. Hence, Theorem \ref{thm:nonconvergentexamplerational} and
Theorem \ref{thm:weylaveragesconvergentcase} together yield the following.
\begin{corollary}
	If for each $k \in \Z$,
	$
	\lim_{n \to \infty}\frac{1}{n}\sum_{j=0}^{n-1} e^{2 \pi i k (v(T^j x)) }  = c_k
	$
	for a sequence of complex numbers $\langle c_k \rangle_{k \in \Z}$. Then, $x$ is a regular sequence. But there exist regular sequences having non-convergent Weyl averages.
\end{corollary}
%% SN - redundant
%% From Lemma \ref{lem:nonconvergentexample}, we know that there exists
%% $x$ with $\dim_{FS}(x)=\Dim_{FS}(x)=1/2$ having divergent Weyl
%% averages. Let $y$ be the diluted sequence with
%% $\dim_{FS}(x)=\Dim_{FS}(x)=1/2$ (see Theorem 6.2 from
%% \cite{dai2004finite}). Clearly, the sliding probability of every
%% finite string converges in $y$. Now, Lemma
%% \ref{lem:deltameasuresandcoutning} and Theorem
%% \ref{thm:cylinderconvergenceimpliesweylconvergence} lets us conclude
%% that $y$ has convergent Weyl averages. Thus if we are only given the
%% dimensions of a sequence, we cannot infer anything about the
%% convergence of the Weyl averages. However, the conclusion of Theorem
%% \ref{thm:nonconvergentcaseweylcriterion} imposes limits on the upper
%% and lower average entropies of the subsequence limit measures of any
%% sequence with given dimension.

\section{Preservation of finite-state dimension under real arithmetic}
\label{sec:preservationoffsd}

In this section, we demonstrate the utility of our framework by
proving the most general results yet regarding the preservation of
finite-state dimension under arithmetic operations like addition with
reals satisfying a natural independence condition, and multiplication
with non-zero rationals. These results strictly generalize all known
results regarding the preservation of finite-state dimension including
those of Doty, Lutz and Nandakumar \cite{DLN06} and Aistleitner
\cite{Aistleitner2011}.
%% we give a Fourier analytic proof of the result by Doty,
%% Lutz and Nandakumar \cite{DLN06} that finite-state dimension is
%% preserved under rational arithmetic. This result is then utilized to
%% provide completely new results generalizing those of Aistleitner
%% \cite{Aistleitner2011}. He showed that normality is preserved under
%% addition with a certain class of reals, and under combinatorial
%% alterations. We generalize this substantially to show that for
%% \emph{every real}, the finite-state dimension and strong dimensions
%% are preserved under these operations. 
Our Weyl criterion plays a pivotal role in these extensions. We
combine our Weyl criterion along with recent estimates by Hochman
\cite{Hochman2014} for the entropy of convolution of probability
measures.

It is easier to analyze addition and multiplication as operations over
$\T$.
%% Since the operations of addition and multiplication are more amenable
%% to analysis when seen as operations over elements of $\T$, 
Hence we first obtain an equivalent Weyl's criterion for finite-state
dimension in terms of measures over $\T$.
%%  rather than measures over $\Sigma^\infty$ as in Theorem
%% \ref{thm:nonconvergentcaseweylcriterion}
We now define the analogues
of upper and lower average entropies for measures over $\T$. This
turns out to be the notion of R\'enyi dimension as defined by Alfr\'ed
R\'enyi in \cite{renyi1959dimension}. Recall that for any $m$ and $w
\in \Sigma_m^n$, $I^m_w$ denotes the interval $\left[ v_m(w0^\infty),
  v_m(w0^\infty)+m^{-\lvert w \rvert} \right)$ in $\T$.

\begin{definition}[R\'enyi Dimension]
\label{def:renyidimension}
For any probability measure $\mu$ on $\T$ and \emph{partition factor}
$m$, let $\mathbf{H}^m_{n}(\mu)=-\sum_{w \in
  \Sigma_m^n}\mu(I^m_w)\log(\mu(I^m_w))$. The \emph{R\'enyi upper and
lower dimensions}  (see \cite{renyi1959dimension} and
\cite{young1982dimension}) are defined as follows, 
\begin{align*}
\overline{\dim}^m_{R}(\mu)=\limsup\limits_{n \to \infty}
\frac{\mathbf{H}^m_{n}(\mu)}{n \log m} \quad \text{and} \quad 
\underline{\dim}^m_{R}(\mu)=\liminf\limits_{n \to \infty}
\frac{\mathbf{H}^m_{n}(\mu)}{n \log m} 
\end{align*}
If $\overline{\dim}^m_{R}(\mu)=\underline{\dim}^m_{R}(\mu)$ then the
\emph{R\'enyi dimension} of $\mu$ is $\dim_{R}^m(\mu) =
\overline{\dim}^m_{R}(\mu)=\underline{\dim}^m_{R}(\mu)$. 
\end{definition}

From the above definition, it seems as if the notion of R\'enyi
dimension is dependent on the choice of the partition factor
$m$. However, R\'enyi upper and lower dimensions are quantities that
are independent of the partition factor. This important fact regarding
R\'enyi dimension seems to be a folklore result. However, for
completeness we give a full proof of this fact in the appendix.

\begin{lemma}
\label{lem:renyibaseinvariance}
Let $\mu$ be any measure on $\T$. For any two partition factors $m_1$ and $m_2$, $\overline{\dim}^{m_1}_{R}(\mu) = \overline{\dim}_{R}^{m_2}(\mu)$ and $\underline{\dim}^{m_1}_{R}(\mu) =\underline{\dim}_{R}^{m_2}(\mu)$. 
\end{lemma}

%\begin{lemma}
%\label{lem:renyibaseinvariance}
%Let $\mu$ be any measure on $\T$. For any two partition factors $m_1$ and $m_2$,
%\begin{align*}
%	\overline{\dim}^{n_1}_{R}(\mu) &= \limsup\limits_{n \to \infty}
%        \frac{\mathbf{H}_{m_1}^n(\mu)}{n \log m_1} = \limsup\limits_{n
%          \to \infty} \frac{\mathbf{H}_{m_2}^n(\mu)}{n\log
%          m_2}=\overline{\dim}_{R}^{m_2}(\mu) \quad\text{and}\\
%	\underline{\dim}^{m_1}_{R}(\mu) &= \liminf\limits_{n \to
%          \infty} \frac{\mathbf{H}_{m_1}^n(\mu)}{n\log m_1} =
%        \liminf\limits_{n \to \infty}
%        \frac{\mathbf{H}_{m_2}^n(\mu)}{n\log
%          m_2}=\underline{\dim}_{R}^{m_2}(\mu). 
%\end{align*}
%\end{lemma}

In the light of Lemma \ref{lem:renyibaseinvariance}, we suppress the
partition factor $m$ in the notations $\overline{\dim}^m_{R}(\mu)$, $
\underline{\dim}^m_{R}(\mu)$ and $\dim_{R}^m(\mu)$ and use
$\overline{\dim}_{R}(\mu)$, $ \underline{\dim}_{R}(\mu)$ and
$\dim_{R}(\mu)$ to refer to the corresponding quantities for a measure
$\mu$ on $\T$. Now,  we state an equivalent Weyl's criterion for
finite-state dimension for $r \in \T$ in terms of weak limit measures over $\T$ and
R\'enyi dimension of measures over $\T$.  

\begin{theorem}[Restatement of Weyl's criterion for finite-state
    dimension (Theorem \ref{thm:nonconvergentcaseweylcriterion})]
  \label{thm:weylaveragesconvergentcase_torus}
	Let $r \in \T$. If for any $\langle n_m
        \rangle_{m=0}^{\infty}$ there exist constants $c_k$ for $k \in
        \Z$ such that \footnote{ The $2^j$ term in (\ref{eqn:weylaveragesconvergentcase_torus}) must be
 replaced with $b^j$ while investigating the above criterion in any
 arbitrary base $b$}
        \begin{align}\label{eqn:weylaveragesconvergentcase_torus}
        \lim_{m \to
          \infty}\frac{1}{n_m}\sum_{j=0}^{n_m-1} e^{2 \pi i k 2^j r } = c_k 
        \end{align}  
          for every $k \in \Z$, then there exists a
        measure $\mu$ on $\T$ such that for every $k$, $c_k = \int
        e^{2 \pi i k y } d\mu$. Let $\mathcal{W}_r$ be
        the collection of all $\mu$ on $\T$ that can be obtained as
        subsequence limits of Weyl averages. Then, $\dim_{FS}(r) = \inf
  \{\underline{\dim}_{R}(\mu) \mid \mu \in \mathcal{W}_r\}$ and $\Dim_{FS}(r) = \sup
  \{\overline{\dim}_{R}(\mu) \mid \mu \in \mathcal{W}_r\}.$        
\end{theorem}
\begin{proof}
If $r$ is a dyadic rational in $\T$, the conclusion is trivial since it is easily verified that $\mathcal{W}_r$ consists only of the measure $\delta_0$. Now, the statement follows since $\dim_R(\delta_0)=0$.
For $r \in \T \setminus \D$, there exists a unique $x \in \Sigma^\infty$ such that $v(x)=r$. And we also have $v(T^j x)=2^j v(x)=2^j r$ for every $j \geq 0$. If we fix the partition factor $m$ to be equal to $|\Sigma|$, then the
quantities $\underline{\dim}_{R}(\mu)$ and $\overline{\dim}_{R}(\mu)$
for any measure $\mu$ on $\T$ coincide with the quantities
$H^-({\hat{\mu}})$ and $H^+({\hat{\mu}})$ of the lifted measure
$\hat{\mu}$. The equivalence of the above theorem with Theorem
\ref{thm:nonconvergentcaseweylcriterion} follows from these observations.	
\end{proof}

\label{sec:fourieranalyticproofofwallstheorem}

D.~D.~Wall in his thesis \cite{Wall49} proved that \emph{if $r \in [0,1]$ and $q$ is any non-zero rational
        number, then $r$ is a normal number if and only if $qr$ and
        $q+r$ are normal numbers}. Doty, Lutz and Nandakumar \cite{DLN06} generalized this result to
arbitrary finite-state dimensions and proved that the finite-state
dimension and finite-state strong dimension of any number are
preserved under multiplication and addition with rational numbers.

\begin{theorem}[\cite{DLN06}]
\label{thm:wallstheorem}
Let $r \in \T$ and $q$ be any non-zero rational number. Then
for any base $b$, $\dim^b_{FS}(r)=\dim^b_{FS}(q+r)=\dim^b_{FS}(qr)$
and $\Dim^b_{FS}(r)=\Dim^b_{FS}(q+r)=\Dim^b_{FS}(qr)$. 
%\begin{align*}
%\dim^b_{FS}(r)&=\dim^b_{FS}(q+r)=\dim^b_{FS}(qr)	\quad\text{and}\\
%\Dim^b_{FS}(r)&=\Dim^b_{FS}(q+r)=\Dim^b_{FS}(qr).	
%\end{align*}	
\end{theorem}

In the above $\dim^b$ and $\Dim^b$ denotes the finite-state dimension
and finite-state strong dimension of the number $r$ calculated by
considering the sequence representing the base-$b$ expansion of $r$.
For $r$ having multiple base $b$ expansions, this does not cause any
ambiguity since in this case the finite-state dimensions of $r$ are
$0$ with respect to any of the two possible expansions. 

In the specific case of normal sequences, Wall's result has been
generalized by Aistleitner in the following form. Let $\mathcal{C}$ be
the set of reals $y=0.y_0 y_1 \dots$ such that the ratio
$P(y_0^{n-1},0)$ goes to 1 as $n$ tends to $\infty$. Then we have the
following.

\begin{theorem}
\label{thm:aistleitnerfirsttheorem}
If $y$ is any real number in $\mathcal{C}$, then for any normal $r \in
\T$ and $q \in \Q$, the number $r+qy$ is normal.
\end{theorem}

We strictly generalize all these above results by formulating a
natural independence notion between two reals. We describe the
framework below. Given strings $x$ and $y$ in $\Sigma^\infty$ and
strings $u,w \in \Sigma^\ell$ for some $\ell \geq 1$, we define the
\emph{joint occurrence count} of $u$ and $w$ in $x$ and $y$ up to $n$
as,
\begin{align*}
N_{u,w}(x_0^{n-1},y_0^{n-1}) = \lvert \{i \in [0,n-\ell]:x_i^{i+\ell-1}=u \text{ and }y_i^{i+\ell-1}=w\} \rvert	
\end{align*}
And, the the \emph{joint occurrence probability} of $u$ and $w$ in $x$
and $y$ up to $n$ is defined as $P_{u,w}(x_0^{n-1},y_0^{n-1}) =
\frac{N_{u,w}(x_0^{n-1},y_0^{n-1})}{n-\ell+1}$.

Informally, we define two infinite strings $x$ and $y$ to be independent if for infinitely many lengths $l$, the occurrence probability distributions of $l$-length strings within x and y are \emph{independent} in the limit. The straightforward formulation of independence between $x$ and $y$ is 
\begin{align*}
	\lim_{n \to \infty} P_{u,w}(x_0^{n-1}, y_0^{n-1}) = \lim_{n \to \infty} P(x_0^{n-1}, u) P(y_0^{n-1}, w).
\end{align*}
 But these limits need not exist for general $x$ and $y$. Hence, the more admissible and useful definition is the following.

\begin{definition}
\label{def:independentstrings}	
Any two strings $x$ and $y$ in $\Sigma^\infty$ are said to be
\emph{independent} if for infinitely many $\ell \geq 1$ and for every $u,w \in \Sigma^\ell$,
\begin{align}
\label{eq:independencecondition}
\lim\limits_{n \to \infty} \left\lvert P_{u,w}(x_0^{n-1},y_0^{n-1}) -
P(x_0^{n-1},u)P(y_0^{n-1},w)\right\rvert = 0. 
\end{align}
\end{definition}

{\bf Note.} A basic intuition for our approach can be viewed as
follows. A standard result in probability theory (see for example,
Shiryaev \cite{Shiryaev}, 2nd. edition, Section II.8) is that, if $X$
and $Y$ are two independent random variables, then the distribution of
$X+Y$ is the convolution of the distributions of $X$ and $Y$. Moreover,
the Fourier coefficients of the convolution is the product of the
Fourier coefficients of the individual distributions. Our result may
be viewed as an analogous result using sequences.

The following theorem gives an important connection between the
exponential averages of the sum of independent reals and the
exponential averages of the individual reals.

\begin{theorem}
\label{thm:independenceimpliesweylproduct}
If $x$ and $y$ are real numbers in $\T$ such that $x$ and $y$ are
independent in the sense of condition \ref{def:independentstrings},
then for any integers $d$, $e$ and $q \in \Q$, 
\begin{align}
\label{eq:independentsumproofeq8}
\lim\limits_{n \to \infty}\left\lvert \frac{1}{n}\sum_{j=0}^{n-1} e^{2
  \pi i k 2^j (dx+ey) } - \frac{1}{n}\sum_{j=0}^{n-1} e^{2 \pi i k 2^j
  dx }   \frac{1}{n}\sum_{j=0}^{n-1} e^{2 \pi i k 2^j ey }
\right\rvert = 0. 
\end{align}
\end{theorem}
\begin{proof}
Observe that for any $k \in \Z$, 
\begin{align*}
\frac{1}{n}\sum_{j=0}^{n-1} e^{2 \pi i k 2^j (dx+ey) } = \frac{1}{n}\sum_{j=0}^{n-1} e^{2 \pi i k 2^j dx }  e^{2 \pi i k 2^j ey }
\end{align*}
With a slight abuse of notation we let $0.x_1 x_2 x_3 \dots$ and $0.y_1 y_2 y_3 \dots$ denote the base-$2$ expansions of $r$ and $y$ respectively. Then,
\begin{align*}
\frac{1}{n}\sum_{j=0}^{n-1} e^{2 \pi i k 2^j (dx+ey) } = \frac{1}{n}\sum_{j=0}^{n-1} e^{2 \pi i dk (0.x_j x_{j+1} x_{j+2} \dots) }  e^{2 \pi i ek (0. y_j y_{j+1} y_{j+2} \dots) }.
\end{align*}
Fix an arbitrary $\ell \geq 1$ such that the independence condition is satisfied. Using the inequality $\lvert e^{i \theta}-1 \rvert \leq \lvert \theta \rvert$, we get that
\begin{align*}
\left\lvert e^{2 \pi i dk (0.x_j x_{j+1} x_{j+2} \dots) } - e^{2 \pi i dk (0.x_j  \dots x_{j+\ell -1}) } \right\rvert \leq \frac{2 \pi \lvert dk \rvert}{2^{\ell -1}}. 
\end{align*}
Similarly,
\begin{align*}
\left\lvert e^{2 \pi i ek (0.y_j y_{j+1} y_{j+2} \dots) } - e^{2 \pi i ek (0.y_j  \dots y_{j+\ell -1}) } \right\rvert \leq \frac{2 \pi \lvert ek \rvert }{2^{\ell -1}}. 
\end{align*}
Using the above two inequalities we obtain that
\begin{align*}
\left\lvert e^{2 \pi i dk (0.x_j x_{j+1} x_{j+2} \dots) }e^{2 \pi i ek (0.y_j y_{j+1} y_{j+2} \dots) } - e^{2 \pi i dk (0.x_j  \dots x_{j+\ell -1}) }e^{2 \pi i ek (0.y_j  \dots y_{j+\ell -1}) }  \right\rvert \leq \frac{2 \pi (\lvert d \rvert+\lvert e \rvert)\lvert k \rvert }{2^{\ell -1}}. 
\end{align*}

Therefore, from the above observations we get
\begin{align}
\label{eq:independentsumproofeq1}
\left\lvert \frac{1}{n}\sum_{j=0}^{n-1} e^{2 \pi i k 2^j (dx+ey) } - \frac{1}{n}\sum_{j=0}^{n-1} e^{2 \pi i dk (0.x_j  \dots x_{j+\ell -1}) }e^{2 \pi i ek (0.y_j  \dots y_{j+\ell -1}) } \right\rvert \leq \frac{2 \pi (\lvert d \rvert+\lvert e \rvert)\lvert k \rvert }{2^{\ell -1}}. 
\end{align}
Observe that,
\begin{align}
\label{eq:independentsumproofeq2}
	\frac{1}{n}\sum_{j=0}^{n-1} e^{2 \pi i dk (0.x_j  \dots x_{j+\ell -1}) }e^{2 \pi i ek (0.y_j  \dots y_{j+\ell -1}) } &= \frac{1}{n} \sum\limits_{u,w \in \Sigma^\ell} N_{u,w}(x_0^{n+\ell-2},y_0^{n+\ell-2}) e^{2\pi i dk (0.u)} e^{2\pi i ek (0.w)} \nonumber \\
		&=\sum\limits_{u,w \in \Sigma^\ell} P_{u,w}(x_0^{n+\ell-2},y_0^{n+\ell-2}) e^{2\pi i dk (0.u)} e^{2\pi i ek (0.w)} 
\end{align}
Now,
\begin{align*}
&\left\lvert \sum\limits_{u,w \in \Sigma^\ell} P_{u,w}(x_0^{n+\ell-2},y_0^{n+\ell-2}) e^{2\pi i dk (0.u)} e^{2\pi i ek (0.w)} - \sum\limits_{u,w \in \Sigma^\ell} P(x_0^{n+\ell-2},u) e^{2\pi i dk (0.u)} P(y_0^{n+\ell-2},w) e^{2\pi i ek (0.w)}  \right\rvert\\
&\leq  \sum\limits_{u,w \in \Sigma^\ell} \left\lvert P_{u,w}(x_0^{n+\ell-2},y_0^{n+\ell-2}) - P(x_0^{n+\ell-2},u)  P(y_0^{n+\ell-2},w) \right\rvert.  
\end{align*}
Since the independence condition is satisfied, we get that
\begin{align}
\label{eq:independentsumproofeq3}
\lim\limits_{n \to \infty} \left\lvert \sum\limits_{u,w \in \Sigma^\ell} \left(P_{u,w}(x_0^{n+\ell-2},y_0^{n+\ell-2})  -  P(x_0^{n+\ell-2},u)  P(y_0^{n+\ell-2},w) \right) e^{2\pi i dk (0.u)} e^{2\pi i ek (0.w)}   \right\rvert =0	
\end{align}
Now,
\begin{align}
\label{eq:independentsumproofeq4}
&\sum\limits_{u,w \in \Sigma^\ell} P(x_0^{n+\ell-2},u) e^{2\pi i dk (0.u)} P(y_0^{n+\ell-2},w) e^{2\pi i ek (0.w)} \nonumber\\
&= \sum\limits_{u \in \Sigma^\ell} P(x_0^{n+\ell-2},u) e^{2\pi i dk (0.u)}  \sum\limits_{w \in \Sigma^\ell} P(y_0^{n+\ell-2},w) e^{2\pi i ek (0.w)} \nonumber\\
&= \frac{1}{n} \sum\limits_{u \in \Sigma^\ell} N(x_0^{n+\ell-2},u) e^{2\pi i dk (0.u)} \times  \frac{1}{n} \sum\limits_{w \in \Sigma^\ell} N(y_0^{n+\ell-2},w) e^{2\pi i ek (0.w)} \nonumber\\
&= \frac{1}{n} \sum_{j=0}^{n-1}  e^{2\pi i dk (0.x_{j} \dots x_{j+\ell-1})} \times  \frac{1}{n} \sum_{j=0}^{n-1}  e^{2\pi i ek (0.y_{j} \dots y_{j+\ell-1})}
\end{align}
Therefore, using (\ref{eq:independentsumproofeq2}),
(\ref{eq:independentsumproofeq3}) and
(\ref{eq:independentsumproofeq4}) we get that for any $\ell$
satisfying the independence condition, 
\begin{align}
\label{eq:independentsumproofeq5}
&\lim\limits_{n \to \infty}\left\lvert \frac{1}{n}\sum_{j=0}^{n-1} e^{2 \pi i dk (0.x_j  \dots x_{j+\ell -1}) }e^{2 \pi i ek (0.y_j  \dots y_{j+\ell -1}) } - \frac{1}{n} \sum_{j=0}^{n-1}  e^{2\pi i dk (0.x_{j} \dots x_{j+\ell-1})}   \frac{1}{n} \sum_{j=0}^{n-1}  e^{2\pi i ek (0.y_{j} \dots y_{j+\ell-1})} \right\rvert  \nonumber\\
&=0. 
\end{align}
Observe that,
\begin{align}
\label{eq:independentsumproofeq6}
&\left\lvert \frac{1}{n}\sum_{j=0}^{n-1} e^{2 \pi i dk (0.x_j x_{j+1} x_{j+2}  \dots ) } - \frac{1}{n} \sum_{j=0}^{n-1}  e^{2\pi i dk (0.x_{j} \dots x_{j+\ell-1})}   \right\rvert \nonumber \\
&\leq \frac{1}{n}\sum_{j=0}^{n-1} \left\lvert e^{2 \pi i dk (0.x_j x_{j+1} x_{j+2}  \dots ) }  -e^{2\pi i dk (0.x_{j} \dots x_{j+\ell-1})}  \right\rvert \nonumber\\
&\leq  \frac{1}{n}\sum_{j=0}^{n-1} \frac{2 \pi \lvert dk \rvert}{2^{\ell-1}} \nonumber \\
&=\frac{2 \pi \lvert dk \rvert}{2^{\ell-1}}.
\end{align}
Similarly,
\begin{align}
\label{eq:independentsumproofeq7}
&\left\lvert \frac{1}{n}\sum_{j=0}^{n-1} e^{2 \pi i dk (0.y_j y_{j+1} y_{j+2}  \dots ) } - \frac{1}{n} \sum_{j=0}^{n-1}  e^{2\pi i dk (0.y_{j} \dots y_{j+\ell-1})}   \right\rvert \leq \frac{2 \pi \lvert ek \rvert}{2^{\ell-1}}.
\end{align}

Finally using (\ref{eq:independentsumproofeq1}), (\ref{eq:independentsumproofeq5}), (\ref{eq:independentsumproofeq6}) and (\ref{eq:independentsumproofeq7}) we obtain that,
\begin{align*}
\lim\limits_{n \to \infty}\left\lvert \frac{1}{n}\sum_{j=0}^{n-1} e^{2 \pi i k 2^j (dx+ey) } - \frac{1}{n}\sum_{j=0}^{n-1} e^{2 \pi i k 2^j dx }   \frac{1}{n}\sum_{j=0}^{n-1} e^{2 \pi i k 2^j ey } \right\rvert \leq \frac{4 \pi (\lvert d \rvert+\lvert e \rvert)\lvert k \rvert }{2^{\ell -1}}.
\end{align*}
Since the independence condition holds for infinitely many $\ell$, we get
\begin{align*}
\lim\limits_{n \to \infty}\left\lvert \frac{1}{n}\sum_{j=0}^{n-1} e^{2 \pi i k 2^j (dx+ey) } - \frac{1}{n}\sum_{j=0}^{n-1} e^{2 \pi i k 2^j dx }   \frac{1}{n}\sum_{j=0}^{n-1} e^{2 \pi i k 2^j ey } \right\rvert = 0.
\end{align*}
\end{proof}

 For any measures $\mu_1$ and $\mu_2$ on $\T$, let $\mu_1 \ast \mu_2$
 denote the convolution of these two measures (see \cite{Rudin1962} or
 \cite{folland2016course}). 

\begin{lemma}
\label{lem:weaklimitconvolutionlemma1}
Let $x$ and $y$ be real numbers in $\T$ such that $x$ and $y$ are
independent in the sense of condition \ref{def:independentstrings} and
let $d,e \in \Z$. Then, for any $\mu \in \mathcal{W}_{dx+ey}$ there
exist $\mu_1 \in \mathcal{W}_{dx}$ and $\mu_2 \in \mathcal{W}_{ey}$
such that $\mu = \mu_1 \ast \mu_2$.
\end{lemma}
\begin{proof}
Consider any $\mu \in \mathcal{W}_{dx+ey}$. Let $\langle n_m
\rangle_{m=0}^{\infty}$ be a subsequence such that 
\begin{align*}
	\lim\limits_{m \to \infty} \frac{1}{n_m}\sum_{j=0}^{n_m-1} e^{2 \pi i k 2^j (dx+ey) } =c_k
\end{align*}
where $\langle c_k \rangle_{k \in \Z}$ are the Fourier coefficients of
$\mu$. Using Prokhorov's Theorem, we assume without loss of generality
that $\langle n_m \rangle_{m=0}^{\infty}$  is such that, 
\begin{align*}
	\lim\limits_{m \to \infty} \frac{1}{n_m}\sum_{j=0}^{n_m-1} e^{2 \pi i k 2^j dx } =c^1_k
\end{align*}
and,
\begin{align*}
	\lim\limits_{m \to \infty} \frac{1}{n_m}\sum_{j=0}^{n_m-1} e^{2 \pi i k 2^j ey } =c^2_k
\end{align*}
where $\langle c^1_k \rangle_{k \in \Z}$ and $\langle c^2_k \rangle_{k
  \in \Z}$ are the Fourier coefficients of measures $\mu_1 \in
\mathcal{W}_{dx}$ and $\mu_2 \in \mathcal{W}_{ey}$ respectively. Using
(\ref{eq:independentsumproofeq8}) we obtain that $c_k=c^1_k c^2_k$ for
every $k \in \Z$. 

  If $\langle c^\ast_k \rangle_{k \in \Z}$ denotes the Fourier
  coefficients of $\mu_1 \ast \mu_2$, then it follows that
  $c^\ast_k=c^1_k c^2_k =c_k$ for every $k \in \Z$ (see Theorem 1.3.3
  from \cite{Rudin1962}). Hence, using the Bochner's Theorem we obtain
  that $\mu = \mu_1 \ast \mu_2$.  
\end{proof}

\begin{lemma}
\label{lem:weaklimitconvolutionlemma2}
If $x$ and $y$ are real numbers in $\T$ such that $x$ and $y$ are
independent in the sense of condition \ref{def:independentstrings}.
Let $d,e \in \Z$  and $q \in \Q$. Then, for any $\mu_1 \in
\mathcal{W}_{dx}$ there exist $\mu \in \mathcal{W}_{dx+ey}$ and $\mu_2
\in \mathcal{W}_{ey}$ such that $\mu = \mu_1 \ast \mu_2$. 
\end{lemma}
\begin{proof}
Let $\langle n_m \rangle_{m=0}^{\infty}$ be a subsequence such that 
\begin{align*}
	\lim\limits_{m \to \infty} \frac{1}{n_m}\sum_{j=0}^{n_m-1} e^{2 \pi i k 2^j dx } =c^1_k
\end{align*}
where $\langle c^1_k \rangle_{k \in \Z}$ are the Fourier coefficients
of $\mu_1$. Using the Prokhorov's Theorem, we assume without loss of
generality that $\langle n_m \rangle_{m=0}^{\infty}$  is such that, 
\begin{align*}
  \lim\limits_{m \to \infty} \frac{1}{n_m}\sum_{j=0}^{n_m-1} e^{2 \pi i k 2^j (dx+ey) } =c_k
\end{align*}
and, 
\begin{align*}
	\lim\limits_{m \to \infty} \frac{1}{n_m}\sum_{j=0}^{n_m-1} e^{2 \pi i k 2^j ey } =c^2_k
\end{align*}
where $\langle c_k \rangle_{k \in \Z}$ and $\langle c^2_k \rangle_{k
  \in \Z}$ are the Fourier coefficients of measures $\mu \in
\mathcal{W}_{dx+ey}$ and $\mu_2 \in \mathcal{W}_{ey}$ respectively.
Using (\ref{eq:independentsumproofeq8}) we obtain that $c_k=c^1_k
c^2_k$ for every $k \in \Z$. Now the lemma follows using the same
argument in the proof of Lemma \ref{lem:weaklimitconvolutionlemma1}. 
\end{proof}

The proofs of the following bounds on R\'{e}nyi dimension of
convolutions crucially employ inequalities from Hochman
\cite{Hochman2014}.

\begin{lemma}
\label{lem:renyidimensionofconvolutionlowerbound}
For any measures $\mu_1$ and $\mu_2$ on $\T$,
$\underline{\dim}_{R}(\mu_1 \ast \mu_2) \geq
\max\{\underline{\dim}_{R}(\mu_1),\underline{\dim}_{R}(\mu_2)\}$ 
\end{lemma}
\begin{proof}
 Using Corollary 4.10 from \cite{Hochman2014} and the fact that $\mu_1
 \ast \mu_2 = \mu_2 \ast \mu_1$ (see \cite{Rudin1962} or
 \cite{folland2016course}) we obtain that, 
 \begin{align*}
 \frac{\mathbf{H}^2_n (\mu_1 \ast \mu_2)}{n} \geq \frac{\mathbf{H}^2_n
   (\mu_1)}{n} - O\left(\frac{1}{n}\right). 	 
 \end{align*}
 By applying $\liminf$ on both sides, we obtain that
 $\underline{\dim}_{R}(\mu_1 \ast \mu_2) \geq
 \underline{\dim}_{R}(\mu_1)$. Using Corollary 4.10 from
 \cite{Hochman2014} and the fact that $\mu_1 \ast \mu_2 = \mu_2 \ast
 \mu_1$ it follows that  $\underline{\dim}_{R}(\mu_1 \ast \mu_2) \geq
 \underline{\dim}_{R}(\mu_2)$. This completes the proof of the lemma. 
\end{proof}

The upper bound for the R\'{e}nyi dimension of the convolution is as
follows. 
\begin{lemma}
\label{lem:renyidimensionofconvolutionupperbound}
Let $\mu_1$ and $\mu_2$ be measures on $\T$. Then,
$\underline{\dim}_{R}(\mu_1 \ast \mu_2) \leq
\underline{\dim}_{R}(\mu_1)+\overline{\dim}_{R}(\mu_2)$ 
\end{lemma}
\begin{proof}
	As in \cite{Hochman2014} for any measure $\mu$ on $\T$, let
        $\sigma_m \mu$ denote the $m$-discretization of $\mu$ defined
        as $\sigma_m \mu = \sum_{w \in \Sigma^m} \mu(I_w)
        \delta_{v(w0^\infty)}$. Let $X$ denote the discrete random
        variables which takes the value $v(w0^\infty)$ with
        probability $\mu_1(I_w)$ for every $w \in \Sigma^m$. Let $Y$
        be the analogous random variable defined using $\mu_2$ such
        that $X$ and $Y$ are independent. Using routine information
        theoretic arguments involving the data processing inequality
        (see \cite{CovTho91}), it follows that $H(X + Y) \leq H(X) +
        H(Y)$. Since $X$ and $Y$ are independent, $\sigma_m \mu_1 \ast
        \sigma_m \mu_2$ denotes the distribution of the random
        variable $X+Y$. Therefore,
	\begin{align}
	\label{eqn:renyidimensionupperboundeq1}
	\frac{\mathbf{H}^2_m(\sigma_m \mu_1 \ast \sigma_m \mu_2)}{m}
        &\leq \frac{\mathbf{H}^2_m(\sigma_m \mu_1) +
          \mathbf{H}^2_m(\sigma_m \mu_2)}{m} = \frac{\mathbf{H}^2_m(
          \mu_1) + \frac{1}{m}\mathbf{H}^2_m( \mu_2)}{m}.
	\end{align}
	Using Lemma 4.8 from \cite{Hochman2014}, we have,
	\begin{align}
	\label{eqn:renyidimensionupperboundeq2}
	\left\lvert \frac{1}{m}\mathbf{H}^2_m(\sigma_m \mu_1 \ast \sigma_m \mu_2) - \frac{1}{m}\mathbf{H}^2_m(\mu_1 \ast \mu_2) \right\rvert \leq O\left( \frac{1}{m} \right).
	\end{align}
	From (\ref{eqn:renyidimensionupperboundeq1}) and (\ref{eqn:renyidimensionupperboundeq2}), we get 
	\begin{align*}
	\frac{\mathbf{H}^2_m(\mu_1 \ast \mu_2)}{m} &\leq \frac{\mathbf{H}^2_m( \mu_1)+\mathbf{H}^2_m( \mu_2)}{m} + O\left( \frac{1}{m} \right) 
	\leq \frac{\mathbf{H}^2_m( \mu_1)}{m} + \sup_{n \geq m} \frac{\mathbf{H}^2_n( \mu_2)}{n} + O\left( \frac{1}{m} \right). 
	\end{align*}
	Taking $\liminf$ on both sides we get that, and noting that
        the second and third terms above have limits, we get the
        required result.
%% 	\begin{align*}
%% 	\liminf\limits_{m \to \infty} \frac{\mathbf{H}^2_m(\mu_1 \ast \mu_2)}{m} &\leq \liminf\limits_{m \to \infty} \frac{\mathbf{H}^2_m( \mu_1)}{m} + \lim\limits_{m \to \infty}\sup\limits_{n \geq m} \frac{\mathbf{H}^2_n( \mu_2)}{n} + \lim\limits_{m \to \infty} O\left( \frac{1}{m} \right) 
%% %% 	&= \liminf\limits_{m \to \infty} \frac{\mathbf{H}^2_m( \mu_1)}{m} + \limsup\limits_{m \to. \infty} \frac{\mathbf{H}^2_m( \mu_2)}{m}.
%% 	\end{align*}
%% 	The required inequality directly follows from the above.
\end{proof}

By applying $\limsup$ instead of $\liminf$ in the proofs of the bounds above we get the following inequalities for R\'{e}nyi upper dimension of $\mu_1 \ast \mu_2$.

\begin{lemma}
\label{lem:boundsforrenyiupperdimension}
For any measures $\mu_1$ and $\mu_2$ on $\T$,
$\overline{\dim}_{R}(\mu_1 \ast \mu_2) \geq
\max\{\overline{\dim}_{R}(\mu_1),\overline{\dim}_{R}(\mu_2)\}$ and $\overline{\dim}_{R}(\mu_1 \ast \mu_2) \leq
\overline{\dim}_{R}(\mu_1)+\overline{\dim}_{R}(\mu_2)$ 
\end{lemma}

The following is our main result.

\begin{theorem}
\label{thm:independentaddtioninequalities}
Let $x$ and $y$ be real numbers in $\T$ such that $x$ and $y$ are
independent in the sense of condition \ref{def:independentstrings}.
Then for any $d,e \in \Z$, 
\begin{enumerate}
	\item\label{item:dimensionboundsitem1} $\dim_{FS}(dx+ey) \geq \max
\{\dim_{FS}(dx), \dim_{FS}(ey) \}$ and $\dim_{FS}(dx+ey) \leq
\dim_{FS}(dx) + \Dim_{FS}(ey)$
 	\item\label{item:dimensionboundsitem2} $\Dim_{FS}(dx+ey) \geq \max
\{\Dim_{FS}(dx), \Dim_{FS}(ey) \}$ and $\Dim_{FS}(dx+ey) \leq
\Dim_{FS}(dx) + \Dim_{FS}(ey)$
\end{enumerate}
 
\end{theorem}

\begin{proof}
We prove the bounds in \ref{item:dimensionboundsitem1}. in Consider any $\mu \in \mathcal{W}_{dx+ey}$. Using Lemma
\ref{lem:weaklimitconvolutionlemma1} we get that there exists $\mu_1
\in \mathcal{W}_{dx}$ and $\mu_2 \in \mathcal{W}_{ey}$ such that $\mu
= \mu_1 \ast \mu_2$. Now, it follows from Lemma
\ref{lem:renyidimensionofconvolutionlowerbound} that
$\underline{\dim}_{R}(\mu)=\underline{\dim}_{R}(\mu_1 \ast \mu_2) \geq
\underline{\dim}_{R}(\mu_1)$. On applying Theorem
\ref{thm:weylaveragesconvergentcase_torus} for $dx \in \T$, we get
$\underline{\dim}_{R}(\mu) \geq \dim_{FS}(dx)$. Since $\mu$ was
arbitrary, applying Theorem \ref{thm:weylaveragesconvergentcase_torus}
for $dx+ey \in \T$, we obtain $\dim_{FS}(dx+ey) \geq \dim_{FS}(dx)$.
The proof of $\dim_{FS}(dx+ey) \geq \dim_{FS}(ey)$ is similar. This
completes the proof of the first inequality.  

In order to show the second inequality, consider any $\mu_1 \in
\mathcal{W}_{dx}$. Using Lemma \ref{lem:weaklimitconvolutionlemma2},
there exist $\mu \in \mathcal{W}_{dx+ey}$ and $\mu_2 \in
\mathcal{W}_{ey}$ such that $\mu = \mu_1 \ast \mu_2$. Now using Lemma
\ref{lem:renyidimensionofconvolutionupperbound}, it follows that
$\underline{\dim}_{R}(\mu) = \underline{\dim}_{R}(\mu_1 \ast \mu_2)
\leq \underline{\dim}_{R}(\mu_1) + \overline{\dim}_{R}(\mu_2)$. On
applying Theorem \ref{thm:weylaveragesconvergentcase_torus} for the
points $dx+ey \in \T$ and $ey \in \T$, we get $\dim_{FS}(dx+ey) \leq
\underline{\dim}_{R}(\mu_1) + \Dim_{FS}(ey)$. Since $\mu_1$ was
arbitrary, applying Theorem \ref{thm:weylaveragesconvergentcase_torus}
for $dx \in \T$, we obtain $\dim_{FS}(dx+ey) \leq \dim_{FS}(dx) +
\Dim_{FS}(ey)$. The inequalities for finite-state strong dimension in \ref{item:dimensionboundsitem2} follows using similar arguments by using Lemma \ref{lem:boundsforrenyiupperdimension} instead of Lemmas \ref{lem:renyidimensionofconvolutionlowerbound} and \ref{lem:renyidimensionofconvolutionupperbound}.
\end{proof}

The following is an immediate corollary of the above theorem.

\begin{corollary}
\label{cor:dimensionofindependentsum}
If $x$ and $y$ are real numbers in $\T$ such that $x$ and $y$ are
independent in the sense of condition \ref{def:independentstrings},
then for any $q \in \Q$, 
\begin{enumerate}
	\item $\dim_{FS}(x+qy) \geq \max
\{\dim_{FS}(x),\dim_{FS}(y)\}$ and  $\dim_{FS}(x+qy)	\leq
\dim_{FS}(x)+\Dim_{FS}(y)$
	\item $\Dim_{FS}(x+qy) \geq \max
\{\Dim_{FS}(x),\Dim_{FS}(y)\}$ and  $\Dim_{FS}(x+qy)	\leq
\Dim_{FS}(x)+\Dim_{FS}(y)$
\end{enumerate}
\end{corollary}
\begin{proof}
 Let $e,d$ be integers such that $q=e/d$ in the reduced form. Since
 $\dim_{FS}(r+e/d \times y)=\dim_{FS}(dr+ey)$ and $\Dim_{FS}(r+e/d \times y)=\Dim_{FS}(dr+ey)$ due to Theorem
 \ref{thm:wallstheorem}, the required conclusion follows immediately
 from Theorem \ref{thm:independentaddtioninequalities}.
 \end{proof}
 
On considering the case when $\Dim_{FS}(y)=0$, we obtain the following
corollaries, generalizing earlier results by Doty, Lutz, Nandakumar
\cite{DLN06} and Aistleitner \cite{Aistleitner2011}, regarding the
preservation of finite-state dimension under addition with an
independent sequence having zero finite-state strong dimension.
  
 \begin{corollary}
\label{cor:preservationofdimensiononindependentaddition}
 	If $x$ and $y$ are real numbers in $\T$ such that $x$ and $y$
        are independent in the sense of condition
        \ref{def:independentstrings} with $\Dim_{FS}(x)=0$, then for
        any $q \in \Q$, $\dim_{FS}(x+qy)=\dim_{FS}(x)$ and $\Dim_{FS}(x+qy)=\Dim_{FS}(x)$. 
 \end{corollary}
 
%%%%%%%%%%%%%%%%%%%%%%%%%%%%%%%%
%% Aistleitner as a corollary %%
%%                            %%
%%%%%%%%%%%%%%%%%%%%%%%%%%%%%%%%

It is easy to verify that any string in $\mathcal{C}$ is independent
of any other string $x \in \Sigma^\infty$. Thus we obtain the
following generalization of Aistleitner's result to every dimension
\cite{Aistleitner2011}.

\begin{corollary}
\label{thm:aistleitnerfirsttheoremgeneralization}
	If $y$ is any real number in $\mathcal{C}$, then for any $x
        \in \T$ and $q \in \Q$, $\dim_{FS}(x+qy) = \dim_{FS}(x)$ and $\Dim_{FS}(x+qy) = \Dim_{FS}(x).$
\end{corollary}

\section*{Acknowledgements}
The authors would like to thank Michael Hochman for technical clarifications regarding his paper \cite{Hochman2014}.

\bibliography{main}

\appendix

\section*{Appendix}
In section A, we give a brief account of some
important equivalent characterizations of finite-state dimension.

\section{Equivalent characterizations of finite-state dimension}

Finite-state dimension was originally defined using finite-state
$s$-gales by Dai, Lathrop, Lutz and Mayordomo \cite{dai2004finite}. We
employed the equivalent characterization of finite-state dimension in
terms of block entropy rates given by Bourke, Hitchcock and
Vinodchandran \cite{bourke2005entropy} in Definition
\ref{def:finitestatedimension}. In this section, we give a brief
account of the original definition of finite-state dimension in terms
of finite-state $s$-gales and an equivalent characterization in terms
of finite-state compression ratios given in \cite{dai2004finite} and
\cite{athreya2007effective}. We give these formulations for the binary
alphabet $\Sigma=\{0,1\}$ for the sake of simplicity. It is routine to
extend these characterizations to arbitrary alphabets.

\subsection{Finite-state dimension using finite-state $s$-gales (\cite{dai2004finite}, \cite{athreya2007effective})}
We first define an $s$-gale
\begin{definition}[$s$-gale \cite{dai2004finite},\cite{LutzDimension2003}]
	Let $s \in [0,\infty)$. A function $d:\Sigma^* \to [0,\infty)$
            is an $s$-gale if it satisfies, $d(\lambda) < \infty$ and
	\begin{align*}
	d(w) = \frac{1}{2^s} \left(d(w0)+d(w1)\right)
	\end{align*}
	for every $w \in \Sigma^*$.
\end{definition}

Now, we define the success criteria for $s$-gales and the corresponding winning sets,
\begin{definition}[Success criteria for $s$-gales \cite{dai2004finite}, \cite{athreya2007effective}]
Let $s \in [0,\infty)$ and let $d$ be an $s$-gale.
\begin{enumerate}
	\item We say that $d$ \emph{succeeds on the sequence} $x \in \Sigma^\infty$ if,
	\begin{align*}
	\limsup\limits_{n \to \infty} d(x_0^{n-1})=\infty.	
	\end{align*}
 	And, the \emph{success set} of $d$ is defined as $S^\infty[d]=\{x \in \Sigma^\infty \mid d \text{ succeeds on } x\}$.
 	\item We say that $d$ \emph{succeeds strongly on the sequence} $x \in \Sigma^\infty$ if,
	\begin{align*}
	\liminf\limits_{n \to \infty} d(x_0^{n-1})=\infty.	
	\end{align*}
 	And, the \emph{strong success set} of $d$ is defined as $S_{str}^\infty[d]=\{x \in \Sigma^\infty \mid d \text{ succeeds strongly on } x\}$.
\end{enumerate}
\end{definition}

Now, we define finite-state gamblers.
\begin{definition}[Finite-state gamblers \cite{dai2004finite}, \cite{SchnorrStimm72}, \cite{Feder1991}]
	A \emph{finite-state gambler} is a $5$-tuple, $G=(Q,\delta,\beta,q_0,c_0)$ where,
	\begin{itemize}
		\item $Q$ is a non-empty set of \emph{states}.
		\item $\delta: Q \times \Sigma \to Q$ is the \emph{transition function}.
		\item $\beta:Q \to \Q \cap [0,1]$ is the \emph{betting function}.
		\item $q_0 \in Q$ is the \emph{initial state}.
		\item $c_0 \geq 0$ is the \emph{initial capital} of the gambler.
	\end{itemize}
\end{definition}

Let $\delta^*: Q \times \Sigma^* \to \Q$ denote the natural \emph{extension} of $\delta$ to finite strings in $\Sigma^*$ defined recursively as,
\begin{align*}
\delta(q,\lambda) &= q \\
\delta(q,wb) &=\delta(\delta^*(q,w),b).	
\end{align*}

Finite-state dimension was defined in \cite{dai2004finite} in terms of finite-state $s$-gales. We define finite-state $s$-gales corresponding to finite-state gamblers.

\begin{definition}[Finite-state $s$-gales \cite{dai2004finite}]
	An $s$-gale of a finite-state gambler $G=(Q,\delta,\beta,q_0,c_0)$ is the function $d^{(s)}_G: \Sigma^* \to [0,\infty)$ defined recursively as,
	\begin{align*}
	d^{(s)}_G(\lambda)&=c_0\\
	d^{(s)}_G(wb)&=2^s d^{(s)}_G(w) \left((1-b)(1-\beta(\delta^*(w)))+b \beta(\delta^*(w)) \right)	
	\end{align*}
	for every $w \in \Sigma^*$ and $b \in \Sigma$. A \emph{finite-state $s$-gale} is an $s$-gale $d$ for which there exists a finite-state gambler $G$ such that $d=d^{(s)}_G$.
\end{definition}

The following is the original definition of finite-state dimension in terms finite-state $s$-gales, given in \cite{dai2004finite}.
\begin{definition}[Finite-state dimension \cite{dai2004finite}]
	Let $x \in \Sigma^\infty$. The \emph{finite-state dimension} of $x \in \Sigma^\infty$ is defined as,
	\begin{align*}
		\dim_{FS}(x)=\inf\{s \in [0,\infty) \mid \exists \text{ a finite-state }s\text{-gale }d\text{ such that }x \in S^{\infty}[d]\}.
	\end{align*}	
\end{definition}

Similarly, the finite-state strong dimension was defined in \cite{athreya2007effective} by replacing $S^{\infty}[d]$ with $S_{str}^{\infty}[d]$.
\begin{definition}[Finite-state strong dimension \cite{athreya2007effective}]
	Let $x \in \Sigma^\infty$. The \emph{finite-state strong dimension} of $x \in \Sigma^\infty$ is defined as,
	\begin{align*}
		\Dim_{FS}(x)=\inf\{s \in [0,\infty) \mid \exists \text{ a finite-state }s\text{-gale }d\text{ such that }x \in S_{str}^{\infty}[d]\}.
	\end{align*}	
\end{definition}

We remark that finite-state dimension and finite-state strong dimension were defined in \cite{dai2004finite} and \cite{athreya2007effective} more generally for subsets of $\Sigma^\infty$. But, we only require the concept of finite-state dimensions of individual sequences in $\Sigma^\infty$ for developing our results.

\subsection{Finite-state compression and finite-state dimension (\cite{dai2004finite}, \cite{athreya2007effective})}

Finite-state dimension is also characterized in terms of compression ratios using information lossless finite-state compressors (\cite{dai2004finite},\cite{athreya2007effective},\cite{bourke2005entropy}). Let $\mathcal{C}$ be the collection of all information lossless finite-state compressors. Let $\mathcal{C}_k$ be the collection of all $k$-state information lossless finite-state compressors. The following compressibility characterization of finite-state dimension and finite-state strong dimension were given in \cite{dai2004finite} and \cite{athreya2007effective} respectively.
\begin{theorem}[\cite{dai2004finite},\cite{athreya2007effective},\cite{bourke2005entropy}]
	For any $x\in \Sigma^\infty$,
	\begin{align*}
	\dim_{FS}(x) = \inf\limits_{C \in \mathcal{C}} \liminf\limits_{n \to \infty} \frac{\lvert C(x_0^{n-1}) \rvert}{n}	
	\end{align*}
	and,
	\begin{align*}
	\Dim_{FS}(x) =  \inf\limits_{k \in \N} \limsup\limits_{n \to \infty} \min\limits_{C \in \mathcal{C}_k} 	\frac{\lvert C(x_0^{n-1}) \rvert}{n}.
	\end{align*}
\end{theorem}

\section{Preliminaries}
For a given block length $l$, we define the \emph{disjoint block
entropy} over $x_0^{n-1}$ as follows.
\begin{align*}
H^d_l(x_0^{n-1}) = -\frac{1}{l}\sum_{w \in \Sigma^l}
P^d(x_0^{n-1},w)\log(P^d(x_0^{n-1},w)).
\end{align*}

Kozachinskiy and Shen (\cite{kozachinskiy2019two}) also demonstrated
that\footnote{Though Shen and Kozachinskiy proved the equivalence
between disjoint block entropies and sliding block entropies for
finite-state dimension, the same techniques in
\cite{kozachinskiy2019two} proves the equivalences for finite-state
strong dimension by replacing $\liminf$'s with $\limsup$'s.},

\begin{align}
\label{eqn:limitdimension}
\dim_{FS}(x) = \lim\limits_{l \to \infty} \liminf\limits_{n \to \infty} H^d_l(x_0^{n-1}) = \lim\limits_{l \to \infty} \liminf\limits_{n \to \infty} H_l(x_0^{n-1})
\end{align}
\begin{align}
\label{eqn:limitstrongdimension}
\Dim_{FS}(x) = \lim\limits_{l \to \infty} \limsup\limits_{n \to \infty} H^d_l(x_0^{n-1}) = \lim\limits_{l \to \infty} \limsup\limits_{n \to \infty} H_l(x_0^{n-1})
\end{align}

These equalities are used in the proofs of certain results in the paper.

\section{Preservation of finite-state dimension under arithmetic and
  combinatorial operations}
  
 We require the following lemma for proving Lemma \ref{lem:renyibaseinvariance}. 
\begin{lemma}
\label{lem:renyibaselemma}
Let $m_1,m_2$ be any two partition factors. For any $l>0$, let $n$ be such that $m_1^n \leq m_2^l < m_1^{n+1}$. Then,
\begin{align*}
\left\lvert \frac{\mathbf{H}_{m_2}^l(\mu')}{l\log m_2}-  \frac{\mathbf{H}_{m_1}^n(\mu')}{n \log m_1}\right\rvert	 \leq \frac{2}{n}.
\end{align*}
\end{lemma}
\begin{proof}
The left hand side is equal to,
\begin{align*}
&\left\lvert \frac{n \log m_1  \mathbf{H}_{m_2}^l(\mu')-\log (m_1^n+m_2^l-m_1^n)\mathbf{H}_{m_1}^n(\mu')}{l\log m_2 \cdot n \log m_1 } \right\rvert\\ &= \left\lvert \frac{n \log m_1 \mathbf{H}_{m_2}^l(\mu')-\left(n \log m_1 + \log \left(1+\frac{m_2^l-m_1^n}{m_1^n}\right)\right)\mathbf{H}_{m_1}^n(\mu')}{l\log m_2 \cdot n \log m_1 } \right\rvert \\
&= \left\lvert \frac{n \log m_1 (\mathbf{H}_{m_2}^l(\mu')-\mathbf{H}_{m_1}^n(\mu'))- \log \left(1+\frac{m_2^l-m_1^n}{m_1^n}\right)\mathbf{H}_{m_1}^n(\mu')}{l\log m_2 \cdot n \log m_1 } \right\rvert.
\end{align*}
Using the triangle inequality and the fact that $m_2^l \geq m_1^n$, we can upper bound the right hand side above by,
\begin{align*}
\frac{\lvert \mathbf{H}_{m_2}^l(\mu')-\mathbf{H}_{m_1}^n(\mu') \rvert}{\lvert l\log m_2 \rvert}	+ \frac{ \log \left(1+\frac{m_2^l-m_1^n}{m_1^n}\right)  \mathbf{H}_{m_1}^n(\mu')}{\lvert n \log m_1 \rvert^2}.
\end{align*}
 Since, $\mathbf{H}_{m_1}^n(\mu') \leq \log(m_1^n) = n \log m_1$ and $m_2^l-m_1^n \leq m_1^n(m_1-1)$, the second term is at most $1/n$. Since $m_2^l \geq m_1^n$, the first term above can be upper bounded by,
 \begin{align}
 \label{eqn:renyibaseinvariancelemmaterm}
 	\frac{\lvert \mathbf{H}_{m_2}^l(\mu')-\mathbf{H}_{m_1}^n(\mu') \rvert}{n \log m_1}.
 \end{align}
 Now, we make the following two observations. First, $\mathbf{H}_{m_2}^l(\mu')$ is the Shannon entropy corresponding to the probability distribution $\mu'$ over the finite set $\{I_w^{m_2}: w \in \Sigma_{m_2}^l\}$. Similarly, $\mathbf{H}_{m_1}^n(\mu')$ is the Shannon entropy corresponding to the probability distribution $\mu'$ over the finite set $\{I_w^{m_1}: w \in \Sigma_{m_1}^n\}$. Secondly, any interval in $\{I_w^{m_1}: w \in \Sigma_{m_1}^n\}$ intersects with at most $m_1$ other intervals in $\{I_w^{m_2}: w \in \Sigma_{m_2}^l\}$ since $m_1^n \leq m_2^l < m_1^{n+1}$. Hence, it follows that $\lvert \mathbf{H}_{m_2}^l(\mu')-\mathbf{H}_{m_1}^n(\mu') \rvert \leq \log m_1$ since conditioned on any interval in $\{I_w^{m_1}: w \in \Sigma_{m_1}^n\}$, there are at most $m_1$ possibilities among the intervals in $\{I_w^{m_2}: w \in \Sigma_{m_2}^l\}$. So, we obtain that \ref{eqn:renyibaseinvariancelemmaterm} is at most $1/n$. The two bounds that we obtained above completes the proof of the lemma.
\end{proof}

Now, we prove Lemma \ref{lem:renyibaseinvariance}. 
\begin{proof}[Proof of Lemma \ref{lem:renyibaseinvariance}]
	We show that,
	\begin{align*}
	\liminf\limits_{n \to \infty} \frac{\mathbf{H}_{m_1}^n(\mu')}{n\log m_1} = \liminf\limits_{n \to \infty} \frac{\mathbf{H}_{m_2}^n(\mu')}{n \log m_2}.	
	\end{align*}
	We show that the left hand side is less than or equal to the right hand side. The opposite inequality can be shown in a similar way by interchanging the roles of $m_1$ and $m_2$. The required inequality follows if for every positive $\epsilon$, there exists infinitely many $N$ such that,
	\begin{align*}
		\frac{\mathbf{H}_{m_1}^N(\mu')}{N\log m_1} \leq \liminf\limits_{n \to \infty} \frac{\mathbf{H}_{m_2}^n(\mu')}{n\log m_2} + \epsilon.
	\end{align*}
	In order to show this, we consider any $L$ such that,
	\begin{align}
	\label{eqn:liminfinequalityrenyi}
		\frac{\mathbf{H}_{m_2}^L(\mu')}{L\log m_2} \leq \liminf\limits_{n \to \infty} \frac{\mathbf{H}_{m_2}^n(\mu')}{n\log m_2} + \frac{\epsilon}{2}.
	\end{align}
	The existence of infinitely many such $L$ is guaranteed by the definition of limit infimum. Now, if $N$ is such that $m_1^N \leq m_2^L < m_1^{N+1}$, using Lemma \ref{lem:renyibaselemma} we have,
	\begin{align*}
		\frac{\mathbf{H}_{m_1}^N(\mu')}{N\log m_1} \leq \frac{\mathbf{H}_{m_2}^L(\mu')}{L\log m_2} + \frac{1}{N} \leq \liminf\limits_{n \to \infty} \frac{\mathbf{H}_{m_2}^n(\mu')}{n\log m_2} + \frac{\epsilon}{2} + \frac{1}{N}.
	\end{align*}
	For large enough $L$, $N$ also gets large enough so that $1/N$ is at most $\epsilon/2$. For such an $N$ we have,
	\begin{align*}
		\frac{\mathbf{H}_{m_1}^N(\mu')}{N\log m_1} \leq \liminf\limits_{n \to \infty} \frac{\mathbf{H}_{m_2}^n(\mu')}{n\log m_2} + \epsilon.
	\end{align*}
	Since there are infinitely many $L$ satisfying \ref{eqn:liminfinequalityrenyi}, there exists infinitely many $N$ satisfying the last inequality. The proof of the part corresponding to lower R\'enyi dimension is complete. The other part can be proved in a similar way.
\end{proof}

\section{$\mu$-normality and finite-state dimension}
\label{sec:munormality}
There are several known techniques for explicit constructions of
normal numbers (see, for example, the monographs by Kuipers and
Niederreiter \cite{KuipersNiederreiterUniform}, or Bugeaud
\cite{BugeaudUniform}), but constructions of those with finite-state
dimension $s \in [0,1)$ follow two techniques: first, to start with a
  normal sequence, and to ``dilute'' it with an appropriate fraction
  of simple patterns, as we did in Section
  \ref{sec:counterexampletoconvergence}, and second, to start with a
  coin with bias $p$ such that $-p\log_2 p - (1-p)\log_2(1-p)=s$, and
  consider any typical sequence drawn from this distribution (see also
  \cite{Miller2011}). The second technique does not directly yield a
  computable normal. As an application of Theorem
  \ref{thm:weylaveragesconvergentcase}, we show that a construction
  due to Mance and Madritsch \cite{madritsch2016construction}
  explicitly yields such sequences, which are computable if the given measure is computable. This technique involves the notion
  of $\mu$-normality used to generalize the Champernowne sequence
  \cite{Cham33}.

\begin{definition}[Mance, Madritsch \cite{madritsch2016construction}]
Let $\mu$ be a measure on $\Sigma^\infty$. We say that $x
\in \Sigma^\infty$ is \emph{$\mu$-normal} if, for every $w \in
\Sigma^*$, $ \lim_{n \to \infty} P(x_0^{n-1},w) = \mu(C_w) $.
\end{definition}
Let $T$ be the left shift transformation $T(x_0x_1x_2\dots)=x_1x_2x_3\dots$ on $\Sigma^\infty$. We say that a measure $\mu$ on $\Sigma^\infty$ is \emph{invariant} with respect to $T$ if for every $A \in \mathcal{B}(\Sigma^\infty)$, $\mu(T^{-1}(A))=A$. If $x$ is $\mu$-normal, as a consequence of Lemma \ref{lem:deltameasuresandcoutning} and Lemma \ref{lem:cylindersetconvergenceandweakconvergence}, we get that $\nu_n \wto \mu$ where $\langle \nu_n \rangle_{n=1}^{\infty}$ is the sequence of averages of Dirac measures constructed out of $\langle T^n x \rangle_{n=0}^{\infty}$. Therefore, we get the following lemma as a consequence of Theorem \ref{thm:weylcriterionforfsd} and the fact that stationary processes have a well defined entropy rate (see Section 3 from \cite{khinchin1957mathematical}).
\begin{lemma}
\label{lem:fsdofmunormal}
Let $\mu$ be a measure on $\Sigma^\infty$ and $x \in
\Sigma^\infty$ be a $\mu$-normal sequence. Then, $\mu$ is invariant with respect to $T$ and $
\dim_{FS}(x)=\Dim_{FS}(x)=H^+(\mu)=H^-(\mu)$.
\end{lemma}
\begin{proof}[Proof of Lemma \ref{lem:fsdofmunormal}]

It is enough to show that $\mu(T^{-1}(C_w))=\mu(C_w)$ for every string $w \in \Sigma^*$. Since $w$ is arbitrary, routine approximation arguments can be used to prove that $\mu$ is thus invariant measure with respect to the left shift transformation.

In order to show that $\mu(T^{-1}(C_w))=\mu(C_w)$, let us first observe that, $T^{-1}(C_w)=C_{0w} \cup C_{1w}$. It can be easily verified that,
\begin{equation}
\label{eqn:slidecountdecomposition}
P(x_0^{n+l-2},w)=P(x_0^{n+l-2},0w)+P(x_0^{n+l-2},1w)+O\left(\frac{1}{n}\right)
\end{equation}
since the slide counts for $0w$ and $1w$ together misses out at most constantly many counts of $w$ at the start and end of $x_0^{n+l-2}$. Now, since $\nu_n \Rightarrow \mu$, using Lemma \ref{lem:deltameasuresandcoutning} we have,
\begin{equation*}
	\lim\limits_{n \to \infty}\nu_n (C_{0w}) = \lim\limits_{n \to \infty} P(x_1^{n+l-2},0w) = \mu(C_{0w})
\end{equation*}
and,
\begin{equation*}
	\lim\limits_{n \to \infty}\nu_n (C_{1w}) = \lim\limits_{n \to \infty}P(x_1^{n+l-2},1w) = \mu(C_{1w}).
\end{equation*}
We also have $\lim\limits_{n \to \infty}\nu_n (C_w) = \mu(C_w)$. Hence from \ref{eqn:slidecountdecomposition} we get, $\mu(C_w)=\mu(C_{0w})+\mu(C_{1w})$ which implies that $\mu(C_w)=\mu(T^{-1}(C_w))$.

It is well-known that stationary processes have a well defined entropy rate (see Section 3 from \cite{khinchin1957mathematical}). The same techniques used in proving this claim can be used to show that $H^+(\mu)=H^-(\mu)$ for any invariant measure $\mu$ which along with Theorem \ref{thm:weylcriterionforfsd} completes a proof of the lemma. 

\end{proof}

We remark that for any $\alpha \in [0,1]$ there exists an invariant
measure $\mu$ on $[0,1]$ such that $H^+ (\mu)=H^-(\mu)=\alpha$. We can
further assume that $\mu$ is a Bernoulli measure.
%% In order to show
%% this, let us consider tossing a biased coin such that the outcome of
%% each toss event has Shannon entropy equal to $\alpha$. That is, let
%% $p$ be any probability measure on $\{0,1\}$ such that $-\p(0)\log_2
%% p(0)-p(1)\log_2 p(1)=\alpha$. Now let $\mu$ be the product measure
%% $p^{\N}$ defined on the Bernoulli space. Due to the additivity of
%% Shannon entropy over independent random variables, it follows that,
%% $-\sum_{w\in \Sigma^n}\mu(C_w)\log(\mu(C_w)) &=-n(p(0)\log_2
%% p(0)+p(1)\log_2 p(1))=n\alpha$.  It follows from the definitions of
%% lower and upper average entropies that $H^+(\mu)=H^-(\mu)=\alpha$. 
For any invariant measure $\mu$ on $\Sigma^\infty$, in Section 3 of
\cite{madritsch2016construction}, Mance and Madritsch construct
$\mu$-normal numbers by generalizing the construction of the
Champernowne sequence. We summarize the construction of $\mu$-normal
sequences given in \cite{madritsch2016construction} below.

{\bf Construction.} (Mance, Madritsch
\cite{madritsch2016construction}) Let $\mathbf{p}_1,\mathbf{p}_2, \dots,
\mathbf{p}_{2^l}$ be any ordering of the set of $l$-length strings
$\Sigma^l$. Let $m_l = \min\{\mu(C_w) \mid w \in \Sigma^l \land
\mu(C_w) >0 \}$. Let $M$ be any constant such that $M \geq
m_l^{-1}$. Then, we define $\mathbf{p}_{l,M} = \mathbf{p}_1^{\lceil M
  \mu(\mathbf{p}_1) \rceil} \mathbf{p}_2^{\lceil M \mu(\mathbf{p}_2)
  \rceil} \mathbf{p}_3^{\lceil M \mu(\mathbf{p}_3) \rceil} \dots
\mathbf{p}_{2^l}^{\lceil M \mu(\mathbf{p}_{2^l}) \rceil}$.

Now given $\mu$, an invariant measure on $\Sigma^\infty$, we construct
a $\mu$-normal number as follows. Let $M_i = \lceil \max\{i^{2i} \log
i, (\inf\{\mu(C_w) \mid w \in \Sigma^i \land \mu(C_w) >0 \})^{-1}\}
\rceil$. Now, let $\ell_1=1$ and for $i \geq 2$ define,
\begin{align*}
\ell_i = \left\lceil \log i\cdot \max \left\{ \frac{M_{i+1}+(i+1)^{i+1}}{M_i}, \frac{M_{i-1}+(i-1)^{i-1}}{M_i} \cdot i \ell_{i-1} \right\} \right\rceil.
\end{align*}
Finally, we define the \emph{Champernowne sequence for $\mu$} as
$x_\mu=\mathbf{p}_{1,M_1}^{\ell_1}\mathbf{p}_{2,M_2}^{\ell_2}\mathbf{p}_{3,M_3}^{\ell_3}\mathbf{p}_{4,M_4}^{\ell_4}
\dots$. In Section 5.2, Mance and Madritsch
show \cite{madritsch2016construction} that $x_\mu$ is a $\mu$-normal
number. So, for any invariant measure $\mu$ (not necessarily
Bernoulli), this construction yields a $\mu$-normal number.  Hence,
the above construction along with Lemma \ref{lem:fsdofmunormal} and
Theorem \ref{thm:cylinderconvergenceimpliesweylconvergence} gives us
the following.

%Using Weyl's criterion for finite-state dimension, we proved in Lemma
%\ref{lem:fsdofmunormal} that for any invariant $\mu$, the dimension
%of every $\mu$-normal sequence is completely determined by the upper
%and lower average entropies of $\mu$.

\begin{theorem}
\label{thm:constructionofalphadimensionednumber}
For any $\alpha \in [0,1]$ and any invariant measure $\mu$ on
$\Sigma^\infty$ such that $H^+(\mu)=H^-(\mu)=\alpha$, the sequence
$x_\mu$ constructed above is such that 
$\dim_{FS}(x_\mu)=\Dim_{FS}(x_\mu)=\alpha$ and, $\lim_{n \to
  \infty}\frac{1}{n}\sum_{j=0}^{n-1} e^{2 \pi i k (v(T^j x_\mu)) }
=\int e^{2 \pi i k v(y) } d\mu$.
\end{theorem}

Given any computable $\alpha$, due to the computability of the Shannon
entropy, it is straightforward to compute a $p$ such that $-p\log_2 p
- (1-p)\log_2(1-p)=\alpha$. Then the Bernoulli measure with bias $p$
is computable (see \cite{DowneyHirschfeldtAlgorithmic}). If $\mu$ is a
Bernoulli measure with bias $p$, Mance and Madritsch
\cite{madritsch2016construction} show that by choosing $ M_i =
\min\{p,1-p\}^{-2i} $ and $\ell_i = i^{2i}$, the resulting $x_\mu$ is
a $\mu$-normal number. Since in this case $\mu$ is computable, it
follows that $x_\mu$ is a computable sequence. Therefore, we obtain
the following theorem.
 
\begin{theorem}
\label{thm:computableconstructionofalphadimensionednumber}
For any computable $\alpha \in [0,1]$ and the Bernoulli measure $\mu$
defined above, the sequence $x_\mu$ is a computable sequence such that
$\dim_{FS}(x_\mu)=\Dim_{FS}(x_\mu)=\alpha$ and, 
\begin{align*}
	\lim_{n \to
  \infty}\frac{1}{n}\sum_{j=0}^{n-1} e^{2 \pi i k (v(T^j x_\mu)) }
=\int e^{2 \pi i k v(y) } d\mu.
\end{align*}
\end{theorem}

\end{document}